\documentclass[acmsmall,screen,nonacm]{acmart}
\setcopyright{none}
\usepackage{amsmath,amsthm}
\usepackage{graphicx}
\usepackage{booktabs}
\usepackage{array}      
\usepackage{longtable}  
\usepackage{makecell}
\usepackage[ruled,vlined,linesnumbered]{algorithm2e}
\usepackage{placeins}
\usepackage{tikz}
\usetikzlibrary{arrows.meta}
\usepackage{cancel}

\newtheorem*{restatedtheorem}{Theorem \ref{thm:d-rfr} (restated)}
\newtheorem{theorem}{Theorem}
\newtheorem{lemma}[theorem]{Lemma}
\newtheorem{proposition}[theorem]{Proposition}
\newtheorem{corollary}[theorem]{Corollary}

\theoremstyle{definition}

\theoremstyle{remark}

\newcommand{\eg}{\emph{e.g.}}
\newcommand{\ie}{\emph{i.e.}}

\newtheorem*{mainresultrestated} {Theorem~\ref{thm:main} (restated)}

\newcommand{\fullalgo}{XYZ-Sketch}
\newcommand{\codeRef}{https://github.com/djwj233/XYZ-Sketch}
\title{Toward Optimal Time-Space Tradeoffs for Set Reconciliation}
\author{Rui Xu}
\email{rxu25@stu.pku.edu.cn}
\affiliation{
  \institution{Peking University}
  \city{Beijing}
  \country{China}
}

\author{Kangyang Zhou}
\email{zky233@mit.edu}
\affiliation{
  \institution{Massachusetts Institute of Technology}
  \city{Cambridge}
  \state{Massachusetts}
  \country{USA}
}

\author{Jiachen Xu}
\email{jucasonxu@gmail.com}
\affiliation{
  \institution{Peking University}
  \city{Beijing}
  \country{China}
}

\author{Jiarui Guo}
\email{ntguojiarui@pku.edu.cn}
\affiliation{
  \institution{Peking University}
  \city{Beijing}
  \country{China}
}

\author{Boyu Xian}
\email{loverintime2021@gmail.com}
\affiliation{
  \institution{Peking University}
  \city{Beijing}
  \country{China}
}

\author{Kaicheng Yang}
\email{ykc@pku.edu.cn}
\affiliation{
  \institution{Peking University}
  \city{Beijing}
  \country{China}
}

\author{Yong Cui}
\email{cuiyong@tsinghua.edu.cn}
\affiliation{
  \institution{Tsinghua University}
  \city{Beijing}
  \country{China}
}

\author{Tong Yang}
\authornote{Corresponding author.}
\email{yangtong@pku.edu.cn}
\affiliation{
  \institution{Peking University}
  \city{Beijing}
  \country{China}
}

\keywords{set reconciliation, streaming sketches, peeling, random hypergraphs}

\begin{document}

\begin{abstract}
Set reconciliation, where two parties each holding a large set of elements aim to identify their set difference, is a fundamental task in many areas.
Time (computation cost) and space (communication cost) are two important metrics in this problem, and much existing work optimizes one at the expense of the other, creating a clear trade-off.
We study this trade-off through algorithmic upper bounds and matching space lower bounds.
On one side, we propose {\fullalgo}.
For sufficiently large $d$, {\fullalgo} achieves high-probability reconciliation for sets whose differing elements are uniformly random, using only $(1+\varepsilon)d$ elements for communication, while achieving $O(\log(1/\varepsilon))$ insertion time and expected $O(d\log V\log^{3/2}(1/\varepsilon))$ decoding time, where $d$ and $V$ denote the size of the difference between two sets and the universe size, respectively.
On the other side, we establish a broad fixed-support canonical model and prove a matching space lower bound, establishing the asymptotic optimality of {\fullalgo}.
Our proof also establishes an extremality theorem resolving Walzer's conjecture.
Experiments demonstrate the practical advantages of {\fullalgo}, and the source code is available at \href{\codeRef}{\codeRef}.
\end{abstract}


\maketitle
\section{Introduction}

\subsection{The Set Reconciliation Problem}

The \emph{set reconciliation} problem asks two parties, Alice and Bob, holding sets $A,B\subseteq \mathcal U$, to recover their difference while communicating as little information as possible.
It is a fundamental problem for maintaining consistency among distributed replicas and databases \cite{DifferenceDigest-EppsteinGoodrichUyedaVarghese-2011,RobustSetReconciliation-ChenKonradYiYuZhang-2014,MultiPartySetReconciliation-MitzenmacherPagh-2018}, disseminating blocks and transaction sets in blockchain networks \cite{Graphene-OzisikAndresenLevineTappBissiasKatkuri-2019,RatelessSetReconciliation-YangGiladAlizadeh-2024}, and comparing packet observations across network measurement points to detect loss or inconsistency \cite{LossRadar-LiMiaoKimYu-2016,FermatSketch-YangWuMiaoYangLiuXuQiuZhaoLvJiXie-2023}.

Let $d:=|(A\setminus B)\cup(B\setminus A)|,V:=|\mathcal U|$.
We focus on the common regime $d\ll |A\cup B|$.
Typically, each party maintains a compact \emph{sketch} of its set.
After receiving Alice's sketch, Bob combines it with his own and decodes the difference.
Our main focus is the \emph{streaming} setting, where updates must keep pace with arriving elements, while decoding is performed offline when reconciliation is requested.
The key metrics are twofold:
\begin{enumerate}
    \item \textbf{Space}: the communication cost, which is measured by the number $\mathfrak C(d)$ of transmitted words in the universe $\mathcal U$.
    \footnote{Precisely, a communication cost of $(\mathfrak C(d)+o_V(1))\log_2 V$ bits corresponds to $\mathfrak C(d)$ universe words. We use this word-level metric for brevity and convenience, as the fully precise bit-level accounting of our algorithm is provided in the main text.}
    A standard information-theoretic lower bound \cite{CPI-MinskyTrachtenbergZippel-2003} gives $\mathfrak C(d)\ge d$ or, more precisely, requires at least $d\log_2(V/d)$ bits.
    \item \textbf{Time}: the computational cost of updating the sketch and decoding the difference.
\end{enumerate}

Most prior work can be classified into two categories.
\textbf{1) Algebraic methods}
(e.g., CPI \cite{CPI-MinskyTrachtenbergZippel-2003} and PinSketch \cite{PinSketch-DodisOstrovskyReyzinSmith-2008}) use algebraic encodings to approach optimal space, namely $\mathfrak C(d)=d+O(1)$, but do not support $O(1)$-time updates.
\textbf{2) Hashing-based methods}
(e.g., IBLT \cite{IBLT-GoodrichMitzenmacher-2011}, Difference Digests \cite{DifferenceDigest-EppsteinGoodrichUyedaVarghese-2011} and Graphene \cite{Graphene-OzisikAndresenLevineTappBissiasKatkuri-2019}) hash each item into a constant number of cells to support $O(1)$ updates, but require more space, typically $\mathfrak C(d)\ge 1.23d$.
The clear trade-off raises a natural question:

\smallskip
\noindent\hypertarget{question:optimal}{\textit{\textbf{Question 1:}}}
\emph{Can one approach optimal space while retaining $O(1)$-time updates?}
\smallskip

In fact, the answer is nuanced: there is both a positive and a negative answer.

\subsection{Positive Results via XYZ-Sketch}

As a positive answer, we propose \emph{XYZ-Sketch}, which approaches optimal space while keeping the update time $O_{\varepsilon}(1)$.
In addition to the theoretical guarantee below, experiments in Section~\ref{sec:experiments} demonstrate its practical advantages.

\begin{theorem}[Main Result of XYZ-Sketch]
\label{thm:main}
Assume $d\le V/2$ and the random-label setting\footnote{For fixed sizes of $A\setminus B$ and $B\setminus A$, these two sets are sampled uniformly from all disjoint pairs of those sizes.} described in Section~\ref{sec:verification}.
For every sufficiently small $\varepsilon,\delta>0$, XYZ-Sketch recovers the difference with probability at least $1-\delta$ for sufficiently large $d$, using:
\begin{itemize}
    \item At most $\lceil(1+\varepsilon)d\log_2V\rceil$ communicated bits, or $\mathfrak C(d)\le (1+\varepsilon)d$;
    \item $O(\log(1/\varepsilon))=O_\varepsilon(1)$ time per update;
    \item Expected $O(d\log V\log^{3/2}(1/\varepsilon))=O_\varepsilon(d\log V)$ decoding time.
\end{itemize}
The hidden constants in the explicit bounds above are independent of $\varepsilon$ and $\delta$.
\end{theorem}


\paragraph{Technical Overview.}

Our design builds on the IBLT architecture and combines three ingredients, with the key novelty lying in the first two.
Section~\ref{sec:technical-overview} gives a more detailed overview.

\textbf{1) High-capacity Cells.} In a conventional IBLT, we map items to some \emph{cells}. The decoder then follows a peeling process: it searches for cells that contain only one differing element ({\ie} those in $(A\setminus B)\cup (B\setminus A)$), then recovers the item and removes it from the other cells it touches.
Once no such cell remains but the decoding is incomplete, the decoding fails immediately.

We introduce \emph{high-capacity cells} that enlarge the capacity of cells from $1$ to $\ell>1$, making the decoding process less prone to failure.
Specifically, for any fixed constant $\ell>1$, a high-capacity cell can recover up to $\ell$ differing elements (e.g., $X,Y,Z$), hence the name {\fullalgo}.

\textbf{2) Degree-aware RFR.} The subproblem within a high-capacity cell can be seen as a set reconciliation instance with $d=\ell$.
In principle, existing algebraic methods such as CPI\cite{CPI-MinskyTrachtenbergZippel-2003} or PinSketch\cite{PinSketch-DodisOstrovskyReyzinSmith-2008} could serve this purpose.
However, neither directly meets our requirements for recovery within a cell (see Section \ref{sec:related-work} for details).
We therefore develop the \emph{Degree-aware Rational Function Reconstruction}
(\textsc{D-RFR}) technique and use it to implement the desired high-capacity cells.

\textbf{3) Spatial Coupling.} We leverage the spatial coupling technique from \cite{PeelingOrientabilityThreshold-Walzer-2025}, further reducing the number of cells required.
We also introduce the \emph{circular trick} and provide a heuristic analysis of key parameters in spatial coupling to guide its use in practice.

\subsection{Negative Results under the Fixed-Support Model}
Readers may note that the update time above still depends on $\varepsilon$.
We therefore seek an ideal algorithm that supports updates in $O(1)$ time with an absolute constant, while achieving $\mathfrak C(d)=(1+o_d(1))d$ as $d$ grows.
Does this algorithm exist? Our results point to a negative answer.

To establish rigorous impossibility results, we introduce a model of \emph{fixed-support canonical sketches}.
We suppose that a sketch consists of cells, each occupying $\ell$ universe words, and an update may modify at most $k$ cells, where $k,\ell$ are fixed constants.
We then require the algorithm to satisfy \emph{fixed-support} and \emph{canonicality} conditions.
\footnote{Informally, these conditions are defined as follows.
Fix the public randomness shared by both parties (e.g., hashing functions).
\emph{Fixed-support} means that each item's permitted write locations are predetermined.
\emph{Canonicality} means that the final sketch state depends only on the set, not on the arrival order.}
The model restricts write locations but allows arbitrary computation during updates and decoding.
As all existing set-reconciliation sketches known to us satisfy the two conditions, we believe this model captures a sufficiently broad class of sketches with $O(1)$-updates.
\footnote{Since constant-time updates can access only constantly many cells, any impossibility result under this relaxation also applies when computation time is bounded.
The additional fixed-support assumption requires each item's permitted write locations to be predetermined, so our results do not cover algorithms with adaptively chosen write locations.}

Within this model, asymptotically optimal space $\mathfrak C(d)=(1+o_d(1))d$ is impossible.
Specifically, for every $k\ge 2,\ell\ge 1$, any high-probability-successful sketch $\mathcal A$ must satisfy 
\begin{equation}
\liminf_{d\to\infty}\frac{\mathfrak C_{\mathcal A}(d)}{d}>1.
\end{equation}
(Corollary~\ref{thm:barrier2-precise})
Thus, such an "ideal algorithm" is unattainable within our model.

Another piece of bad news is that \emph{errors are also unavoidable}:
for sufficiently large $d$ and $V$, no $O(1)$-update algorithm in this model can guarantee zero-error recovery if its space depends only on $d$, even if $\mathfrak C(d)=d^{100}$ (Theorem~\ref{thm:barrier1-precise}).
Where, then, does the optimal trade-off lie for set reconciliation?

\subsection{Two Sides of the Same Coin: A Unified Frontier}

We approach this question by characterizing the minimum space
required for a given write-support bound $k$ and cell capacity $\ell$.
For every fixed $k,\ell\ge2$, this frontier of the space ratio $\mathfrak C(d)/d$ is $\ell/c^{\mathsf{orient}}_{k,\ell}$ in our model, where $c^{\mathsf{orient}}_{k,\ell}$ is the orientability threshold defined in Section~\ref{sec:background}.
We establish this frontier from two directions, corresponding to the positive and negative results above.

\textbf{1) Upper bound} (Lemma~\ref{thm:orientability-benchmark}).
Under the random-label setting above, {\fullalgo} supports updates in $O(k\ell)$ time and achieves $\mathfrak C_{\mathrm{XYZ}}(d)\le \left(\ell/c^{\mathsf{orient}}_{k,\ell}+o_d(1)\right)d.$

\textbf{2) Lower bound} (Theorem~\ref{thm:main-precise}).
Every high-probability-successful sketch family $\mathcal A$ in our model with the same fixed $k,\ell$ must satisfy
\begin{equation}
\liminf_{d\to\infty}\frac{\mathfrak C_{\mathcal A}(d)}{d}
\ge
\frac{\ell}{c^{\mathsf{orient}}_{k,\ell}}\ge \limsup_{d\to\infty}\frac{\mathfrak C_{\mathrm{XYZ}}(d)}{d}.
\end{equation}

Together, these bounds identify the asymptotic space frontier for fixed $k,\ell$ in our model.
Moreover, achieving $\mathfrak C(d)\le (1+\varepsilon)d$ requires $k\ell=\Omega(\log(1/\varepsilon))$, while XYZ-Sketch attains this optimal tradeoff (Theorem~\ref{thm:main-precise}).

Indeed, our positive and negative answers are just two sides of the same coin: the frontier is strictly greater than one for every fixed $k,\ell$, giving the negative answer, but approaches one as $k$ and $\ell$ increase, giving the positive answer.

\paragraph{Technical Overview.}
To explain where this frontier $\ell/c^{\mathsf{orient}}_{k,\ell}$ comes from, consider a hypergraph in which each cell is a vertex and each differing element is a hyperedge joining the cells it touches.
This hypergraph is \emph{$\ell$-orientable} if each hyperedge can be assigned to one of its vertices so that no vertex receives more than $\ell$ hyperedges.
Its density $c$ counts hyperedges per vertex ({\ie}, elements per cell), so the space ratio is $\ell/c$ as each cell occupies $\ell$ words.
When each element independently chooses $k$ cells uniformly at random, $c^{\mathsf{orient}}_{k,\ell}$ is the density threshold for $\ell$-orientability.

While the upper bound side follows from a more precise analysis of {\fullalgo}, our lower bound begins with an information-theoretic argument showing that high-probability success requires the induced hypergraph to be $\ell$-orientable (Lemma~\ref{lem:orientability-necessary}).
This reduces our lower-bound analysis to a graph-theoretic question: which strategy for choosing write locations allows the highest orientability threshold?
Walzer~\cite[Conclusion]{PeelingOrientabilityThreshold-Walzer-2025} conjectured that the uniform distribution is optimal.
Our extremality theorem (Theorem~\ref{thm:uniform-extremality}) proves a strengthened form of this conjecture, hence yielding the space lower bound $\ell/c^{\mathsf{orient}}_{k,\ell}$.

\subsection{Key Contributions \& Paper Organization}

\begin{itemize}
\item We propose {\fullalgo}, approaching optimal space $(1+\varepsilon)d$ while supporting $O(\log(1/\varepsilon))$-time updates. By introducing high-capacity cells implemented through \textsc{D-RFR}, we give the first practical application of $\ell$-peeling in set reconciliation, to the best of our knowledge.
\item We formulate the fixed-support canonical model, characterize the frontier in this model, and also establish the asymptotic optimality of {\fullalgo} for fixed $k,\ell$.
\item Along the way, we prove an extremality theorem (Theorem~\ref{thm:uniform-extremality}) resolving Walzer's conjecture.
Since our lower-bound argument does not rely on any particular algorithm, we expect it, together with the extremality theorem, to be useful in analyzing a broader class of probabilistic data structures, including static retrieval structures, cuckoo hashing, and XOR filters.
\item We implement XYZ-Sketch in C++ and release the source code through a \href{\codeRef}{public repository}. We also conduct experiments, demonstrating its practical advantages.
\end{itemize}

The remainder of this paper is organized as follows.
Section~\ref{sec:background} introduces the relevant background.
Section~\ref{sec:technical-overview} provides a high-level overview of {\fullalgo}.
Section~\ref{sec:algo} develops the algorithm in detail.
Section~\ref{sec:analysis} analyzes {\fullalgo} and compares it with IBLT.
Section~\ref{sec:barriers} develops our model, then establishes the lower bounds and optimality results.
Section~\ref{sec:experiments} presents the experimental evaluation, and Section~\ref{sec:conclusion} concludes.

\section{Background}
\label{sec:background}

\subsection{Problem Setting and Conventions}
\label{sec:prelim-conventions}

In real-world scenarios, the exact difference size may not be known in advance.
To address the issue, prior work represented by Strata Estimator~\cite{DifferenceDigest-EppsteinGoodrichUyedaVarghese-2011} estimates the difference size by exchanging only $O(\log V)$ data words in one round.
We therefore focus on the setting where both parties know an upper bound on the difference size, which we denote by $d$.

We work in the \emph{one-pass streaming} setting.
Alice and Bob maintain sketches $\mathsf{Enc}(A)$ and $\mathsf{Enc}(B)$ of their respective sets.
Each party processes a stream of insertions and deletions,
where each insertion adds an absent element and each deletion removes a present element.
Each update is processed exactly once and cannot be revisited; between updates, each party retains information only in its sketch state.
Given $A,B\subseteq\mathcal U^*$ satisfying $|(A\setminus B)\cup(B\setminus A)|\le d$, the decoder must recover both $A\setminus B$ and $B\setminus A$ from $(\mathsf{Enc}(A),\mathsf{Enc}(B))$.
We additionally require $d\le V/2$; otherwise, the problem is of little practical interest.

For simplicity, we identify the universe with a finite field $\mathcal U=\mathbb F_V$, where $V$ is a prime power.
Let $\mathcal U^*:=\mathcal U\setminus\{0\}$.
Throughout the paper, we assume that $A,B\subseteq\mathcal U^*$.
This does not make much difference, as membership of $0$ can be communicated separately using one bit.
We treat each elementary field operation over $\mathcal U$ as taking $O(1)$ time.

We measure communication in $\mathcal U$-word equivalents.
Formally, a sketch has communication cost $\mathfrak C(d)$ if, as $V\to\infty$,
it occupies $\bigl(\mathfrak C(d)+o_V(1)\bigr)\log_2 V$ bits.
This accounting includes all stored auxiliary information, such
as counters and fingerprints, but excludes shared public randomness (e.g. hashing functions).

We denote $[n]:=\{1,2,\cdots,n\}$ for $n\ge 0$ with the special case $[0]=\varnothing$.
For a finite set $A$ and an integer $d\ge 0$, we write
$\binom{A}{d}:=\{B\subseteq A:|B|=d\}$ and $\binom{A}{\le d}:=\{B\subseteq A:|B|\le d\}$.
For a map $f:\mathcal U^*\to\{-1,0,1\}$, we define $\operatorname{supp}(f):=\{v:f(v)\ne 0\}.$
The expression \emph{with high probability} (\emph{w.h.p.} for short) is used when the probability is $1-o_d(1)$ as $d\to\infty$.
Unless otherwise stated, every constant in $O(\cdot)$ is absolute.
Subscripts on $O(\cdot)$ ({\eg}, $O_{\varepsilon}(\cdot),O_{k,\ell}(\cdot)$)indicate parameters on which the implied constant may depend, whereas subscripts on $o(1)$ indicate the variable tending to infinity.

\subsection{Hypergraphs and Thresholds}
\label{sec:prelim-hypergraph}

A \emph{hypergraph} is defined as a pair $H=(\mathcal V,\mathcal E)$, where every \emph{hyperedge}
$e\in \mathcal E$ is a subset of $\mathcal V$. We allow repeated hyperedges.
For a vertex $v\in \mathcal V$, let $\deg_H(v)$ denote the number of hyperedges incident to $v$, counted with multiplicity.

An \emph{$\ell$-peeling step} selects a vertex $v$ satisfying $1\le \deg_H(v)\le \ell$,
and removes all hyperedges incident to $v$. A hypergraph is called
\emph{$\ell$-peelable} if repeated $\ell$-peeling removes every
hyperedge.
When $\ell=1$, this is the usual singleton-peeling process.

An \emph{$\ell$-orientation} of $H$ is a map $\phi:\mathcal E\to \mathcal V$
such that
\begin{equation}
\forall v\in\mathcal V,\  |\phi^{-1}(v)|\le \ell\quad \text{and}\quad  \forall e\in \mathcal E,\ 
\phi(e)\in e
\end{equation}
If such a map exists, then $H$ is \emph{$\ell$-orientable}.  In fact, every
$\ell$-peelable hypergraph is $\ell$-orientable, as we can orient each hyperedge
to the vertex that removes it in a peeling sequence. Note that the converse does
not hold in general.
Moreover, by Hall's theorem, $H$ is $\ell$-orientable iff every
submultiset $F$ of $\mathcal E$ satisfies
$|F|\le \ell\left|\bigcup_{e\in F}e\right|.$

Let $H^{(k)}_{n,m}$ denote the hypergraph on $n$ vertices formed by sampling $m$ hyperedges independently and uniformly from $\binom{[n]}{k}$.
Its edge density is $c:=\frac{m}{n}$.
For a family of random hypergraphs $(G_{n,c})_{n\in\mathbb N,\,c\ge0}$, we say that $c^*$ is a threshold for a property $\mathcal P$ if
\begin{equation}
    c^*
:=
\sup\left\{
c\ge0:
\lim_{n\to\infty}
\Pr\!\left[G_{n,c}\in\mathcal P\right]
=1
\right\}.
\end{equation}
For convenience, let
$c^{\mathsf{peel}}_{k,\ell}$ and $c^{\mathsf{orient}}_{k,\ell}$ denote the thresholds
for $\ell$-peelability and $\ell$-orientability, respectively, of $H^{(k)}_{n,\lfloor cn\rfloor}$.
Since every $\ell$-peelable hypergraph is $\ell$-orientable, $c^{\mathsf{peel}}_{k,\ell}\le c^{\mathsf{orient}}_{k,\ell}$.

For these two properties, the thresholds are sharp~\cite{PeelingOrientabilityThreshold-Walzer-2025} except for $(k,\ell)=(2,1)$.
In particular, when $c<c^{\mathsf{peel}}_{k,\ell}$, then $H^{(k)}_{n,\lfloor cn\rfloor}$ is $\ell$-peelable w.h.p., whereas if
$c>c^{\mathsf{peel}}_{k,\ell}$, it is not $\ell$-peelable w.h.p.
The analogous statement holds for $\ell$-orientability.

\subsection{Polynomial Arithmetic}
\label{sec:prelim-polynomial}

To avoid confusion between $X$ and $\chi$, we follow the convention of using $Z$ as the formal variable for polynomials.
A nonzero polynomial is \emph{monic} if its leading coefficient is $1$.
For polynomials over $\mathcal U$, $F\bmod G$ denotes the unique remainder of degree less than $\deg G$ obtained by dividing $F$ by $G$, and $F\equiv G\pmod H$ means that $H$ divides $F-G$.
For $F,G$ not both zero, $\gcd(F,G)$ denotes their monic greatest common divisor.

Throughout, we assume that degree-$k$ polynomials can be multiplied in $O(k\log k)$ field operations.
Their roots can be also found using the Cantor–Zassenhaus algorithm, a Las Vegas algorithm in expected $O(k\log^2 k\log V)$ field operations.

For any $S\subseteq\mathcal U^*$, define its \emph{characteristic
polynomial} as 
\begin{equation}
    \chi_S(Z):=\prod_{x\in S}(Z-x).
    \label{eq:characteristic-polynomial}
\end{equation}
It is clear that the polynomial $\chi_S$ is monic of degree $|S|$ and has root set $S$.
Also, for disjoint $S,T$, the polynomials $\chi_S,\chi_T$ are coprime and
$\chi_{S\cup T}=\chi_S\chi_T$.

For $F,G\in \mathcal U[Z]$ with $G(0)\neq0$, we write
$
\frac{F(Z)}{G(Z)}\bmod Z^\ell
$
for the unique polynomial $R(Z)$ with $\deg R<\ell$ satisfying
$R(Z)G(Z)\equiv F(Z)\pmod{Z^\ell}.$
Such an $R$ exists uniquely because $G(0)\neq0$ implies that $G$ is
invertible modulo $Z^\ell$.  In particular, for every
$x\in\mathcal U^*$, the polynomial $Z-x$ is invertible modulo
$Z^\ell$.  Hence products with negative exponents are well-defined.


\subsection{Related Work}
\label{sec:related-work}

We briefly review existing set reconciliation work here, also demonstrating that most prior work share the abstraction of \emph{fixed-support} and \emph{canonical} in our barrier model.

\paragraph{Algebraic Methods}

Algebraic methods are rooted in Error Correcting Codes (ECC).
Characteristic-polynomial schemes represented by CPI\cite{CPI-MinskyTrachtenbergZippel-2003} are essentially variants of the Reed-Solomon Code\cite{ReedSolomonCode-ReedSolomon-1960}, whereas PinSketch\cite{PinSketch-DodisOstrovskyReyzinSmith-2008} and MiniSketch\cite{MiniSketch-BitcoinCore-2018} are variants of BCH codes\cite{BCHCode-Hocquenghem-1959,BCHCode-BoseRayChaudhuri-1960}.
These sketches can be viewed as \emph{global}, as each update modifies almost all entries.

However, CPI typically requires field extensions, whereas PinSketch and Minisketch recover only the symmetric difference and cannot
distinguish $A\setminus B$ and $B\setminus A$.
\footnote{Recovering the direction of each item typically requires maintaining an additional membership-query data structure. It may be easy to do so when the full set is static and already stored, but less suitable for a one-pass streaming sketch, e.g. in network measurement and packet loss detection scenarios.}
Our D-RFR follows the same coding-theoretic tradition, but takes a different representation.
By directly maintaining the coefficients of a truncated characteristic polynomial over the native field, D-RFR avoids CPI's field-extension overhead while retaining directional information.

\paragraph{Hashing-Based Methods.}
IBLT~\cite{IBLT-GoodrichMitzenmacher-2011} and IBLT-based schemes such as Difference Digest~\cite{DifferenceDigest-EppsteinGoodrichUyedaVarghese-2011} and Graphene~\cite{Graphene-OzisikAndresenLevineTappBissiasKatkuri-2019} support local updates and recover differences through singleton peeling.
Our {\fullalgo} retains this outer architecture while replacing singleton recovery with recovery of up to $\ell$ signed residual items per cell.
Section~\ref{sec:comparison-sc-iblt} compares our construction with IBLT and its spatially coupled variants.

\paragraph{Hybrid, Rateless, and Zero-Error Schemes.}
In fact, there have been attempts to combine the
two routes.
Parity Bitmap Sketch~\cite{PBS-GongLiuLiuXuOgiharaYang-2020}
combines ECC with hashing, attaining higher space-efficiency than IBLT practically (but not theoretically).
A recent hybrid construction~\cite{SetReconciliationTradeoffs-BelazzouguiKucherovWalzer-2024} augments singleton peeling with a global BCH stash to recover the remaining differences, yielding stronger reliability.
In contrast, {\fullalgo} changes the peeling rule itself by introducing $\ell$-peeling.

Rateless IBLT~\cite{RatelessSetReconciliation-YangGiladAlizadeh-2024} transmits an extendable sequence of coded cells, avoiding the need to know $d$ in advance.
For zero-error recovery, IBLTs with listing guarantees~\cite{IBLTListingGuarantees-MizrahiBarLevYaakobiRottenstreich-2023} guarantee recovery of every set of size at most $d$, while CertainSync~\cite{CertainSync-KeniaginYaakobiRottenstreich-2025} extends it to rateless reconciliation.
Note that these zero-error algorithms do not contradict our zero-error barrier, as their communication cost strongly depends on the universe size $V$, typically $\Omega_k(V^{1/k}\log V)$ bits.
Table~\ref{tab:related-work} compares representative methods.
\begin{table}[htbp]
    \begin{center}
    \caption{Algorithm comparison. Randomized running-time bounds are in expectation.}
    \label{tab:related-work}
    \begin{tabular}{c | c | c | c}
    \specialrule{1pt}{0em}{0em}
    \textbf{Algorithm}
    & \textbf{Space} ($\mathfrak C$)
    & \textbf{Insertion Time}
    & \textbf{Decoding Time}\\
    \specialrule{1pt}{0em}{0em}

    CPI / CPISync~\cite{CPI-MinskyTrachtenbergZippel-2003,CPISync-TrachtenbergStarobinskiAgarwal-2002}
    & $d+O(1)$
    & $O(d)$
    & $O(d^3\log V)$\\
    \hline

    PinSketch / MiniSketch~\cite{PinSketch-DodisOstrovskyReyzinSmith-2008,MiniSketch-BitcoinCore-2018}
    & $d$
    & $O(d)$
    & $O(d^2\log V)$\\
    \hline

    IBLT~\cite{IBLT-GoodrichMitzenmacher-2011}
    & $(1.23+\varepsilon)d$
    & $O(1)$
    & $O(d)$\\
    \hline

    Rateless IBLT\cite{RatelessSetReconciliation-YangGiladAlizadeh-2024}
    & $\approx 1.35d$ (avg.)
    & $O(\log d)$
    & $O(d\log d)$\\
    \hline

    \makecell{{\fullalgo}\\\textit{(Our Work)}}
    & $(1+\varepsilon)d$
    & $O(\log(1/\varepsilon))$
    & $O(d\log V\log^{3/2}(1/\varepsilon))$\\
    \specialrule{1pt}{0em}{0em}
    \end{tabular}
    \end{center}
\end{table}
\paragraph{High Capacity and Spatial Coupling.}
Multi-item processing appears in early shared-memory simulations under
the $c$-collision rule~\cite{SharedMemory-DietzfelbingerMeyer-1993},
which serves all requests arriving at a memory module whenever their number
is at most $c$.
Luby et al.~\cite{AndOrTree-LubyMitzenmacherShokrollahi-1998}
developed And--Or tree analyses for iterative recovery processes.
These works provide early precedents for, respectively, capacity-based
local processing and the analysis of peeling-type dynamics.
Spatial coupling was developed in coding theory to achieve \emph{threshold saturation} in LDPC decoding~\cite{ThresholdSaturation-KudekarRichardsonUrbanke-2011}.
We use Walzer's extension to hashing-based structures~\cite{PeelingOrientabilityThreshold-Walzer-2025}, which brings the $\ell$-peelability threshold arbitrarily close to the $\ell$-orientability threshold.
\section{Overview of XYZ-Sketch}
\label{sec:technical-overview}

As discussed in Section~\ref{sec:background}, Alice and Bob maintain sketches in a one-pass streaming setting.
This section gives an intuitive overview of the design of XYZ-Sketch, while the formal description is provided in Section~\ref{sec:algo}.

\subsection{From IBLT Dynamics to High-Capacity Cells}

Just as in an IBLT~\cite{IBLT-GoodrichMitzenmacher-2011}, the sketch consists of $M$ \emph{cells}.
Each item $x$ is mapped to cells by hashing, and each cell stores a compact summary of the items mapped to it.
Hence, inserting or deleting $x$ updates the corresponding cells.
We fix two constant parameters: \textbf{1)} $k$, the number of cells touched by each element; \textbf{2)} $\ell$, the \emph{capacity} of each cell, which we will discuss below.

When a reconciliation is required, we conceptually \emph{subtract} Bob's sketch from Alice's cell by cell.
Since the same hash functions are used by Alice and Bob, common elements cancel, leaving ``positive'' elements from $A\setminus B$ and ``negative'' elements from $B\setminus A$, as illustrated in Figure \ref{fig:overview-peeling}(a).

During decoding, a conventional IBLT repeatedly looks for a cell containing exactly one residual element, called a \emph{pure cell}.
It decodes this cell to recover the element and its sign, then uses the hash functions to locate all $k$ cells to which the element maps and removes its contribution from each of them.
This removal may create new pure cells, allowing the process, called \emph{peeling}, to continue.
Decoding succeeds if all cells become empty; otherwise, it fails when no pure cells remain.

XYZ-Sketch follows the same dynamics but uses \emph{high-capacity cells},
which support recovery of up to $\ell$ residual elements, together with their signs, rather than just one.
Currently, a cell is pure if it contains at most $\ell$ residual elements, making the peeling process less prone to failure.
As illustrated in Figure~\ref{fig:overview-peeling}(b), peeling a pure
cell can make other cells pure.

\begin{figure}[htbp]
    \centering
    \resizebox{\linewidth}{!}{%
    \begin{tikzpicture}[
        font=\small,
        cell/.style={
            draw=black!55,
            minimum width=1.42cm,
            minimum height=0.58cm,
            inner sep=2pt
        },
        active/.style={
            cell,
            draw=blue!60!black,
            thick,
            fill=blue!9
        },
        flow/.style={
            -{Stealth[length=1.6mm]},
            semithick,
            draw=black!65
        },
        heading/.style={font=\footnotesize},
        rowlabel/.style={anchor=east, font=\footnotesize}
    ]

    \node at (1.45,1.2) {(a) Cell-wise Subtraction};

    \foreach \i/\x in {1/0,2/1.5,3/3} {
        \node[heading] at (\x,0.55) {Cell \i};
    }

    \node[rowlabel] at (-0.85,0) {$\mathsf{Enc}(A)$};
    \node[cell] at (0,0) {$a,\ {\color{black!45}c}$};
    \node[cell] at (1.5,0) {$a,d,e$};
    \node[cell] at (3,0) {$d,e,\ {\color{black!45}c}$};

    \node[rowlabel] at (-0.85,-0.85) {$\mathsf{Enc}(B)$};
    \node[cell] at (0,-0.85) {$b,\ {\color{black!45}c}$};
    \node[cell] at (1.5,-0.85) {$\varnothing$};
    \node[cell] at (3,-0.85) {$b,\ {\color{black!45}c}$};

    \draw[black!40] (-0.72,-1.28) -- (3.72,-1.28);
    \node[rowlabel] at (-0.85,-1.7) {Difference};
    \node[cell] at (0,-1.7) {$+a,-b$};
    \node[cell] at (1.5,-1.7) {$+a,+d,+e$};
    \node[cell] at (3,-1.7) {$-b,+d,+e$};

    \node[heading] at (1.5,-2.3)
        {$A=\{a,c,d,e\},\quad B=\{b,c\}$};

    \begin{scope}[xshift=5.25cm]
        \node at (2.1,1.2) {(b) High-capacity peeling};

        \node[heading] at (0,0.55) {Initial};
        \node[heading] at (2.1,0.55) {Peel $+a,-b$};
        \node[heading] at (4.2,0.55) {Peel $+d,+e$};

        \node[active] at (0,0) {$+a,-b$};
        \node[cell] at (0,-0.85) {$+a,+d,+e$};
        \node[cell] at (0,-1.7) {$-b,+d,+e$};

        \node[cell] at (2.1,0) {$\varnothing$};
        \node[active] at (2.1,-0.85) {$+d,+e$};
        \node[cell] at (2.1,-1.7) {$+d,+e$};

        \node[cell] at (4.2,0) {$\varnothing$};
        \node[cell] at (4.2,-0.85) {$\varnothing$};
        \node[cell] at (4.2,-1.7) {$\varnothing$};

        \draw[flow] (0.8,-0.85) -- (1.3,-0.85);
        \draw[flow] (2.9,-0.85) -- (3.4,-0.85);

        \node[heading] at (2.1,-2.3)
            {$k=\ell=2$};
    \end{scope}

    \end{tikzpicture}%
    }
    \caption{Subtraction and peeling with $k=\ell=2$.
    (a) Each item hashes to two cells; the common item $c$ cancels after subtraction.
    (b) Cells 1--3 are arranged top to bottom, with peeling proceeding
    from left to right. Recovering $+a,-b$ makes the other two cells pure;
    recovering either then completes decoding.
    Shading marks the cell selected for recovery.}
    \label{fig:overview-peeling}
\end{figure}

\subsection{Recovering Multiple Elements via D-RFR}

How can a cell recover up to $\ell$ residual elements?
We represent its items by a characteristic polynomial defined in (\ref{eq:characteristic-polynomial}), storing only
its lowest $\ell$ coefficients (i.e. $\chi_S(Z)\bmod Z^\ell$), together with an item count.
Inserting an item $x$ multiplies the polynomial by $Z-x$, while
deleting it divides by $Z-x$, always modulo $Z^\ell$.

When two cells are subtracted, we divide their polynomial representations,
so factors corresponding to common items cancel.
For example, the first cell in Figure~\ref{fig:overview-peeling}(a)
represents $\{a,c\}$ for Alice and $\{b,c\}$ for Bob.
After subtraction, it stores
\begin{equation}
    R(Z)
    = \frac{(Z-a)\cancel{(Z-c)}}{(Z-b)\cancel{(Z-c)}}
      \bmod Z^\ell
    = \frac{Z-a}{Z-b}\bmod Z^\ell.
    \label{eq:overview-cancellation}
\end{equation}
Although the original cells may contain many items, the total degree
of the numerator and denominator after cancellation is exactly
the number of residual elements.
For a pure cell, this total degree is at most $\ell$.
Our \textsc{D-RFR} algorithm exploits this bound to reconstruct the numerator and denominator from the truncated quotient.\footnote{The
reconstruction also uses the difference between the two cells'
item counts; see Section~\ref{sec:algo} for details.}
Finding their roots then yields the positive and negative residual
elements, respectively: $+a$ and $-b$ in the example above.

\subsection{Connection to Hypergraphs and Why Spatial Coupling Helps}

Readers may notice a resemblance between the peeling process in our decoding and the hypergraph peeling process introduced in
Section~\ref{sec:prelim-hypergraph}.
Indeed, if we view each cell as a vertex and each difference
element as a hyperedge joining its hash locations, then recovering a nonempty pure cell means deleting all hyperedges incident to a vertex of degree $\le \ell$.
At this abstract level, our decoding process is exactly $\ell$-peeling.

Consequently, when each element independently hashes to a uniformly chosen set of $k$ cells, and all other components work correctly, we can expect decoding success to be related to the peeling threshold $c^{\mathsf{peel}}_{k,\ell}$
(defined in Section~\ref{sec:prelim-hypergraph}).
In fact, the corresponding space factor is
$\ell/c^{\mathsf{peel}}_{k,\ell}$.
To see why, $c^{\mathsf{peel}}_{k,\ell}$ counts hyperedges per vertex, or elements per cell, so $1/c^{\mathsf{peel}}_{k,\ell}$ counts cells per element. We then multiply by $\ell$ as each cell stores $\ell$ elements for its polynomial representation.

However, increasing cell capacity alone does not make this
space factor approach $1$.
Specifically, the peeling ratio $\ell/c^{\mathsf{peel}}_{k,\ell}$ cannot be made arbitrarily close to $1$ when $k$ and $\ell$ varies, whereas the orientability ratio $\ell/c^{\mathsf{orient}}_{k,\ell}$ can.
The difference is further illustrated in Appendix~\ref{app:threshold-table}.
This motivates our use of \emph{spatial coupling}.
Deferring its technical details to Section~\ref{sec:algo},
the key benefit is that it allows the peeling threshold of the corresponding hypergraph family to approach the orientability threshold
$c^{\mathsf{orient}}_{k,\ell}$.
Thus, the space factor can approach
$\ell/c^{\mathsf{orient}}_{k,\ell}$ instead,
providing the basis for our near-optimal space guarantee.


\section{Design of {\fullalgo}}
\label{sec:algo}

\newcommand\fgp{h_{\mathrm{fp}}}

In this section, we present {\fullalgo}.
For specific implementation details, refer to the C++ implementation in our repository \href{\codeRef}{\codeRef}.

\subsection{Framework: Interface of High-Capacity Cells and $\ell$-Peeling}
\label{sec:framework}

\subsubsection{Structure of the \fullalgo}

Just as illustrated in Section \ref{sec:technical-overview}, we fix two constant parameters $k,\ell$, and maintain a sketch $\mathcal B$ containing $M$ \emph{cells}.
The cells are numbered from $1$ to $M$, where $i$-th cell is denoted by $\mathcal B[i]$.
We use $k$ mutually independent hash functions $h_1,\ldots,h_k:\mathcal U^*\rightarrow[M]$.\footnote{In our implementation, we use MurmurHash~\cite{MurmurHash3-Appleby-2016}, which is a widely adopted hashing scheme.}

We now describe the interface of high-capacity cells, deferring its detailed implementation to Section \ref{sec:d-rfr}.
Each cell $C$ can be viewed as maintaining a map $f_C:\mathcal U^*\to \{-1,0,1\}$, where $f_C(x)$ stands for the sign of $x$, defaults to zero.
Note that $f_C$ should be viewed as a variable rather than as a static map.
As an example, we take $\mathcal U=\mathbb F_5$ and write $f_C$ as a $4$-dimensional vector, as $\mathcal U^*=\{1,2,3,4\}$.
A cell $C$ containing $1,3$ corresponds to $f_C=(1,0,1,0)$.
When a new item $4$ is inserted, $f_C$ now becomes $(1,0,1,1)$.
If we subtract another cell containing $2,3$, the result will be $(1,0,1,1)-(0,1,1,0)=(1,-1,0,1)$.
In order to perform such desired operations, our cell should support point updates (for inserting / deleting a single element), global subtraction (for the cell-wise subtraction operation) and decoding.
Formally, our cell should support the following operations:
\begin{itemize}
    \item $\mathsf{Update}(C,x,\sigma)$, where
    $\sigma\in\{+1,-1\}$, which means $f_C(x)\gets f_C(x)+\sigma$;
    \item $\mathsf{Subtract}(C,C')$, which returns the cell state representing
    $f_C-f_{C'}$, thereby canceling elements shared by the two cells and retaining their differences;
    \item $\mathsf{Decode}_{\ell}(C)$, which returns exactly $\{(x,f_C(x)):x\in\mathcal U^*,f_C(x)\ne 0\}$ when $C$ is pure, and otherwise it may return $\bot$ or an arbitrary output.
\end{itemize}

Here, a cell is called \emph{pure} iff it contains at most $\ell$ nonzero entries, i.e. $|\operatorname{supp}(f_C)| \le \ell$.
Note that all operations require the resulting value to remain in $\{-1, 0, 1\}$; otherwise, the behavior is undefined.

\subsubsection{The Encoding \& Decoding Algorithm}

During encoding, when an element $x\in\mathcal U^*$ is inserted, we compute its $k$ hash locations $h_1(x),\ldots,h_k(x)$, remove duplicate locations\footnote{The de-duplication step is performed for simplicity.
It also ensures consistency with the definition of $f_C$.} to obtain
$\Gamma(x):=\{h_1(x),\ldots,h_k(x)\}$, and update each remaining cell once.
For each $j\in\Gamma(x)$, we apply
$\mathsf{Update}(\mathcal B[j],x,+1)$, while deletions use sign $-1$ instead.
For reconciliation, we first compute
\begin{equation}
\mathcal D[i]:=\mathsf{Subtract}(\mathcal B_A[i],\mathcal B_B[i]),\qquad i\in[M],
\end{equation}
which cancels common elements.
The decoder then attempts to recover these signed elements by running $\mathsf{Decode}_{\ell}$ on each cell.

As a detail that we have not mentioned before, decoding on a non-pure cell may also return a candidate.
Hence, every decoded candidate must pass a verification step, as described in Section~\ref{sec:high-capacity-cell-dec-and-ver}.
When a cell is decoded and passes verification, we remove every contribution of $(x,\sigma)$ in the cell, then recheck the affected cells for their purity.
The decoder reports success if all cells have returned to their initial empty states, and reports failure otherwise.
Furthermore, for simplicity of analysis, we have to set an upper bound $M+kd$ on the number of steps and the cumulative number of peeled signed items, counting repetitions, at $d$.
Algorithm~\ref{alg:pseudo:decode} gives the complete procedure.

\begin{algorithm}
\caption{{\fullalgo}: Peeling Decoder}
\label{alg:pseudo:decode}
\KwIn{Residual sketch $\mathcal D$ and difference bound $d$.}
\KwOut{$A\setminus B$ and $B\setminus A$, or $\bot$.}

$(S_{+1},S_{-1})\gets(\varnothing,\varnothing)$\;
Initialize $Q$ with all cell indices $1,\ldots,M$\;
$t\gets0,\ s\gets0$ \tcp*{They are only used to bound running time.}
\While{$Q$ is not empty}{
    $t\gets t+1$\;
    \lIf{$t>M+kd$}{\KwRet $\bot$}
    Remove the front index $i$ from $Q$\;
    $R\gets\mathsf{Decode}_{\ell}(\mathcal D[i])$\;
    \If{$R=\bot$ or $R=\varnothing$ or $R$ fails verification in cell $i$}{\textbf{continue}}
    \lIf{$s+|R|>d$}{\KwRet $\bot$}
    $s\gets s+|R|$\;
    \ForEach{$(x,\sigma)\in R$}{
        $S_{\sigma}\gets S_{\sigma}\cup\{x\}$\;
        \ForEach{$j\in\Gamma(x)$}{
            $\mathsf{Update}(\mathcal D[j],x,-\sigma)$\;
            Append $j$ to $Q$\;
        }
    }
}
\If{all cells are in their initial empty states}{
    \KwRet $(S_{+1},S_{-1})$
}
\KwRet $\bot$\;
\end{algorithm}




\subsection{Implementing the High-Capacity Cell via D-RFR}
\label{sec:d-rfr}

We now implement the high-capacity cell interface in Section~\ref{sec:framework} using D-RFR.

\subsubsection{Cell Representation and Updates}

Recall that each cell $C$ should represent a map $f_C:\mathcal U^*\to\{-1,0,1\}$ without storing it explicitly.
In our construction, it stores a count field $c_C \in \mathbb Z/(2\ell+1)\mathbb Z$, a polynomial field $P_C\in \mathcal U[Z]/(Z^\ell)$, and an optional fingerprint field $\mathsf{fp}_C\in \mathbb Z/q\mathbb Z$:
\begin{equation}
\label{eq:cell-state}
\begin{aligned}
c_C &:= \sum_{x\in \mathcal U^*} f_C(x) \pmod{2\ell+1},\\
P_C(Z) &:= \prod_{x\in \mathcal U^*} (Z-x)^{f_C(x)} \bmod Z^\ell,\\
\mathsf{fp}_C &:= \sum_{x\in \mathcal U^*} f_C(x)\fgp(x) \pmod q,
\end{aligned}
\end{equation}
where the \emph{fingerprint} function $\fgp:\mathcal U^*\to\mathbb Z_q$ is a random hash function independent of the placement hashes $h_{1\cdots k}$, with $q$ odd.
For example, a cell containing $3,8$ as positive items and $6$ as
negative has
\begin{equation}
\begin{aligned}
c_C&=2-1=1,\\
P_C(Z)
&=\frac{(Z-3)(Z-8)}{Z-6}\bmod Z^\ell\\
\mathsf{fp}_C &=\fgp(3)+\fgp(8)-\fgp(6).
\end{aligned}
\end{equation}

Initially, the state of an empty cell is $(0,1,0)$.
For a legal call of $\mathsf{Update}(C,x,\sigma)$, we set
\begin{equation}
\label{eq:cell-update}
\begin{aligned}
c_C &\gets c_C+\sigma \pmod{2\ell+1},\\
P_C(Z) &\gets P_C(Z)(Z-x)^\sigma \bmod Z^\ell,\\
\mathsf{fp}_C &\gets \mathsf{fp}_C+\sigma\fgp(x) \pmod q.
\end{aligned}
\end{equation}
Similarly, $\mathsf{Subtract}(C,D)$ returns
\begin{equation}
\label{eq:cell-subtract}
\left(
(c_C-c_{D}) \bmod (2\ell+1),\;
\frac{P_C}{P_{D}}\bmod Z^\ell,\;
(\mathsf{fp}_C-\mathsf{fp}_{D})\bmod q
\right).
\end{equation}
It is obvious that all the above operations preserve \eqref{eq:cell-state}, proving their correctness.
Detailed pseudocode for sketch updates and subtraction is given in Algorithms~\ref{alg:pseudo:update} and~\ref{alg:pseudo:subtract} in Appendix~\ref{app:pseudocode}.

\subsubsection{Decoding and Verification}
\label{sec:high-capacity-cell-dec-and-ver}

In this part, we show how to decode a cell.
For a pure cell $C$, let $\Delta_C^+:=\{x:f_C(x)=+1\}$ and $\Delta_C^-:=\{x:f_C(x)=-1\}$.
In the decoding phase, the former set contains residuals in $A\setminus B$, while the latter contains those in $B\setminus A$.
By \eqref{eq:cell-state},
\begin{equation}
\label{eq:cell-reconstruction}
P_C(Z)=
\frac{\chi_{\Delta_C^+}(Z)}{\chi_{\Delta_C^-}(Z)}
\bmod Z^\ell.
\end{equation}

We aim to exploit these properties to recover both the numerator and the denominator.
For this reason, we propose \emph{Degree-aware Rational Function Reconstruction} (D-RFR) as follows:

\begin{theorem}[The D-RFR Technique]
\label{thm:d-rfr}
Let $R(Z)$ be a polynomial with $\deg R<\ell$, and let
$m\in[-\ell,\ell]$.
Suppose that there exist monic coprime polynomials $F(Z),G(Z)$
satisfying
\begin{equation}
\label{eq:drfr-conditions}
\begin{gathered}
G(0)\ne0,\qquad \deg F+\deg G\le\ell,\qquad \deg F-\deg G=m,\\
\dfrac{F(Z)}{G(Z)}\equiv R(Z)\pmod{Z^\ell}.
\end{gathered}
\end{equation}
Then $F,G$ are uniquely determined by $(R,m)$ and can be recovered in $O(\ell\log^2\ell)$ field operations using an explicit algorithm.
\end{theorem}

We defer the discussion of D-RFR to Section \ref{sec:overview-d-rfr} and simply use it as a black-box subroutine.


For a pure cell, the numerator and denominator in
\eqref{eq:cell-reconstruction} are monic and coprime, with total degree
at most $\ell$.
Also, their degree difference is the unique $m\in[-\ell,\ell]$ satisfying $m\equiv c_C\pmod{2\ell+1}$.
Theorem~\ref{thm:d-rfr} therefore recovers them from $(P_C,m)$, hence finding their roots yields $(\Delta_C^+,\Delta_C^-)$.
If reconstruction fails or the returned pair does not satisfy \eqref{eq:drfr-conditions}, just report a failure.

However, a non-pure cell may also produce a candidate.
Hence further verification is necessary.
We provide two checks to verify as follows.

\textbf{1) Rehashing.} For a candidate $(S,T)$ for $(\Delta^+_C,\Delta^-_C)$ decoded from cell $i$, we first check that
\begin{equation}
\forall x\in (S\cup T), \qquad i\in \Gamma(x).
\end{equation}
This is called \emph{rehashing verification}, as these items must be mapped to cell $i$ by definition.

\textbf{2) Fingerprint.} When the fingerprint field is enabled, we additionally check
\begin{equation}
\label{eq:cell-verification}
\sum_{x\in S}\fgp(x)
-
\sum_{x\in T}\fgp(x)
\equiv \mathsf{fp}_C \pmod q.
\end{equation}
The overall decoding algorithm is presented in Algorithm \ref{alg:cell-decode}.










\begin{algorithm}
\caption{High-Capacity Cell: Decoding and Verification}
\label{alg:cell-decode}
\SetKwProg{Fn}{Function}{:}{}

\Fn{$\mathsf{Decode}_{\ell}(C)$}{
    $m\gets$ the unique integer in $[-\ell,\ell]$ such that
    $m\equiv c_C\pmod{2\ell+1}$\;
    $(F,G)\gets\mathsf{D\text{-}RFR}(P_C,m,\ell)$\;
    \lIf{reconstruction fails}{\KwRet $\bot$}
    \lIf{$F,G$ are not monic and coprime}{\KwRet $\bot$}
    \lIf{\eqref{eq:drfr-conditions} fails with $R=P_C$}{\KwRet $\bot$}
    $S\gets\mathsf{FindRoots}(F),\quad
     T\gets\mathsf{FindRoots}(G)$\;
    \lIf{$F\ne\prod_{x\in S}(Z-x)$ or
         $G\ne\prod_{x\in T}(Z-x)$}{\KwRet $\bot$}
    \KwRet $\{(x,+1):x\in S\}\cup\{(x,-1):x\in T\}$\;
}

\BlankLine

\Fn{$\mathsf{Verify}_i(C,R)$}{
    $S\gets\{x:(x,+1)\in R\},\quad
     T\gets\{x:(x,-1)\in R\}$\;
    \lIf{$\exists x\in S\cup T,\ i\notin\Gamma(x)$}{
        \KwRet $\mathrm{false}$
    }
    \lIf{fingerprint is enabled and \eqref{eq:cell-verification} fails}{
        \KwRet $\mathrm{false}$
    }
    \KwRet $\mathrm{true}$\;
}
\end{algorithm}

A pure cell always passes verification, while an incorrect candidate that passes constitutes a \emph{false positive}.
We will analyze this false positive rate in Section \ref{sec:analysis}.

\subsubsection{Technical Overview of D-RFR}
\label{sec:overview-d-rfr}

We now return to Theorem \ref{thm:d-rfr}.
In fact, the problem is essentially the classical Pad\'e-approximation problem.
A direct application of standard Rational Function Reconstruction (RFR)~\cite{RFR-KhodadadMonagan-2006,RationalReconstruction-CollinsEncarnacion-1995} would require the polynomials modulo $Z^{\ell+1}$, and hence $\ell+1$ coefficients for storage.
In contrast, our D-RFR leverages $m$ to save one coefficient:
knowing the polynomial state modulo $Z^\ell$ suffices.

The key idea is just an elementary trick: since both polynomials are monic, we can cancel the leading term of one polynomial and reduce
the total degree by at least one.
For example, if $F=Z^2+aZ+b$ and $G=Z+c$, then
\begin{equation}
\dfrac{F}{G}-Z=\dfrac{F-ZG}{G}=\dfrac{(a-c)Z+b}{Z+c}.
\end{equation}
Denoting $H:=(a-c)Z+b$, one may find $\deg H+\deg G\le2$, whereas $\deg F+\deg G=3$, reducing the total degree.
However, this trick requires knowing the degree difference $m$ to tell us how to align the leading terms.
To obtain the complete algorithm, it remains to handle several elementary but tedious cases.
We thus defer the full procedure and the uniqueness proof to Appendix~\ref{app:drfr}.

\subsection{Spatial Coupling}
\label{sec:spatial-coupling}

Spatial coupling originated in coding theory as a technique for improving
iterative-decoding thresholds, and was later adapted to hashing-based data
structures in \cite{PeelingOrientabilityThreshold-Walzer-2025}.
Rather than mapping the $k$ locations of an element independently over
the whole sketch, spatial coupling places them in a common short range.
Fix a coupling parameter $z>0$. Let
$g_0:\mathcal U^*\to[0,z)$ and
$g_1,\ldots,g_k:\mathcal U^*\to[0,1)$ be mutually independent uniform
random functions. For every $x\in\mathcal U^*$, define
\begin{equation}
h'_i(x):=
1+\left\lfloor
\frac{g_0(x)+g_i(x)}{z+1}M
\right\rfloor,
\qquad i\in[k].
\end{equation}

Conditioned on $g_0(x)$, all $k$ locations of $x$ lie in an interval of
length approximately $M/(z+1)$.
Intuitively, cells near the two ends have smaller
expected loads, and hence provide the initial pure cells from which
$\ell$-peeling propagates through the sketch.

\paragraph{Circular Trick.}
\label{para:circular}
For practical performance, we propose a circularized placement rule,
under which a coupled range may wrap around from the last cell to the
first.
Specifically, fix $a\in [0,1)$ and let $g_0:\mathcal U^*\to[0,z+a)$ and
$g_1,\ldots,g_k:\mathcal U^*\to[0,1)$ be mutually independent uniform
random functions, and replace the above mapping by
\begin{equation}
h_i^{(a)}(x):=
1+\left\lfloor
\left(
\frac{g_0(x)+g_i(x)}{z+1}
\bmod 1
\right)M
\right\rfloor,
\qquad i\in[k],    
\end{equation}
and correspondingly define
$\Gamma(x):=\{h_i^{(a)}(x):i\in[k]\}$.
When $a=0$, this reduces to the standard spatial coupling.
\section{Theoretical Analysis}
\label{sec:analysis}

In this section, we present the theoretical results and their proof
sketches, while full proofs are deferred to Appendix~\ref{app:analysis}.
Throughout this section, we assume $M=\Theta(d)$.
For convenience, we use $q=1$ to denote the version without a fingerprint.

\subsection{Complexity}

\begin{proposition}[Computational Complexity]
\label{prop:complexity}
Suppose $|(A\setminus B)\cup(B\setminus A)|\le d$.
Then {\fullalgo} supports each update in $O(k\ell)$ time
and performs decoding in expected $O(k\ell d(\log^2\ell\log V+k))=O_{k,\ell}(d\log V)$ time.\footnote{As in CPI and PinSketch, the expectation
arises solely from the randomized root-finding routine, while all remaining operations in the decoding procedure are deterministic.}
\end{proposition}

\begin{proof}[Proof sketch]
The only nontrivial point is the decoding complexity. During correct peeling, each differing item is recovered once and triggers at most $k$ cell updates, resulting in at most $M+kd$ decoding attempts. The claimed bound then follows by inspecting the pseudocode step by step.
\end{proof}

\begin{proposition}[Bit-level Space Complexity]
\label{prop:bit-space}
By jointly encoding all cells, the entire sketch can be represented using at most
\begin{equation}
\left\lceil M\left(
\ell\log_2 V
+\log_2(2\ell+1)
+\log_2 q
\right)\right\rceil
\end{equation}
bits for transmission.
Equivalently, its size is
$M\bigl(\ell\log_2 V+\log_2 q+O(\log\ell)\bigr)$ bits.
\end{proposition}

The three terms account for the $\ell$ polynomial coefficients,
the count field, and the optional fingerprint, respectively.

\subsection{Pure-Cell Verification}
\label{sec:verification}

In this part, we show that in the random label setting discussed below, rehashing verification alone yields a global false-positive probability of $O_{k,\ell}(d^{1-\ell})$.
To further suppress this error, especially at small $d$, it is suggested to use an independent fingerprint additionally, reducing the bound to $O_{k,\ell}(d^{1-\ell}/q)$.

\paragraph{Random Label.}
Fix $d_+,d_-\ge0$ with $d_++d_-\le d$.
We analyze the model in which $(A\setminus B,B\setminus A)$ is sampled uniformly from all disjoint ordered pairs of sizes $(d_+,d_-)$.
In fact, arbitrary fixed input instances can be reduced to this setting by applying an independent, uniformly random public bijection $\phi:\mathcal U^*\to\mathcal U^*$ before sketching and $\phi^{-1}$ after decoding.
Since we already treat fully random hash functions as constant-time public oracles, we adopt the same convention for $\phi$ and  $\phi^{-1}$, supposing that they can be evaluated in $O(1)$ time with no communication cost.
\footnote{Practical candidates for invertible relabeling include format-preserving encryption such as the Feistel-based FF1 scheme~\cite{FPE-Dworkin-2016}.Our analysis assumes an ideal uniformly random permutation.}

\begin{lemma}[False-Positive Rate]
\label{prop:false-positive}
In the random-label setting, the probability that a false positive occurs during decoding is $O_{k,\ell}\left(\frac{d^{1-\ell}}{q}\right)$.
Consequently, for $\ell\ge 2$ and $q=1$, the error rate is $O_{k,\ell}(d^{-1})$, thus negligible for sufficiently large $d$.
\end{lemma}

\begin{proof}[Proof sketch]
The intuition is that the state of a non-pure cell contains essentially random information.
Consequently, a spurious output of \textsc{Decode}$_\ell$ behaves roughly like a random candidate, with the dominant case containing $\ell$ elements.
For such a candidate to pass rehashing verification, every recovered element must hash back to the cell being tested.
Intuitively, each element hits any particular cell with probability
$O(k/M)$, so a random candidate containing $\ell$ elements would
pass rehashing with probability $O_{k,\ell}(M^{-\ell})$.
The independent fingerprint contributes another factor $1/q$.
With $M+O(kd)$ decoding attempts, this suggests a global
false-positive bound of $O_{k,\ell}(d^{1-\ell}/q)$.
Appendix~\ref{app:analysis-verification} makes this intuition rigorous using an algebraic collision bound and a union bound over all
possible residual states and candidates.

\end{proof}

\subsection{Space Overhead Analysis}
\label{sec:space-overhead}

In this part, we analyze the conventional coupling rule $h_i'$ from Section~\ref{sec:spatial-coupling}.
Fix $k,\ell\ge2$ and $c<c^{\mathsf{orient}}_{k,\ell}$, for every sufficiently large fixed $z$, choosing $M=\lceil(1+1/z)d/c\rceil$ makes the induced coupled hypergraph $\ell$-peelable w.h.p.~\cite{PeelingOrientabilityThreshold-Walzer-2025}.
Together with Lemma~\ref{prop:false-positive}, this yields $o(1)$ decoding failure probability even without fingerprints.

\begin{lemma}[XYZ-Sketch Approaches the Orientability Threshold]
\label{thm:orientability-benchmark}
Assume $d\le V/2$ and use the random-label setting from
Section~\ref{sec:verification}.
For every fixed $k,\ell\ge2$, there exist sequences $M_d,z_d$ which
depend only on $d,k,\ell$, such that {\fullalgo} solves every fixed instance of difference size at most $d$
w.h.p., with
\begin{equation}
\limsup_{d\to\infty}\frac{\mathfrak C(d)}{d}
\le \frac{\ell}{c^{\mathsf{orient}}_{k,\ell}}.
\end{equation}
\end{lemma}

\begin{proof}[Proof sketch]
The above construction has asymptotic space ratio $M\ell/d\approx(1+1/z)\ell/c$.
Taking $c\uparrow c^{\mathsf{orient}}_{k,\ell}$ and
$z\to\infty$ sufficiently slowly with $d$ gives the claim.
\end{proof}

Combining the above lemma with our complexity and bit-space bounds yields the following explicit time--space trade-off as follows:

\begin{mainresultrestated}
Assume $d\le V/2$ and use the random-label setting from
Section~\ref{sec:verification}.
For every sufficiently small $\varepsilon>0$ and
every $\delta>0$, there are parameters such that,
for all sufficiently large $d$, {\fullalgo} solves every fixed instance of difference size at most $d$ with probability at least
$1-\delta$, using $O(\log(1/\varepsilon))=O_{\varepsilon}(1)$ per-item update time,
expected $O(d\log V\log^{3/2}(1/\varepsilon))=O_{\varepsilon}(d\log V)$ decoding time, and at most
$\lceil(1+\varepsilon)d\log_2V\rceil$ communicated bits.
The hidden constants in the explicit bounds above are
independent of $\varepsilon$ and $\delta$.
\end{mainresultrestated}

\begin{proof}[Proof sketch]
For $k=\ell=t$, the orientability threshold formula gives
$\ell/c^{\mathsf{orient}}_{k,\ell}
\le 1+\exp(-t^2+O(t\log t))$ (Lemma \ref{lem:speed-approach}).
Choosing $t=\Theta(\sqrt{\log(1/\varepsilon)})$ makes
$\ell/c^{\mathsf{orient}}_{k,\ell}<1+\varepsilon/2$.
For sufficiently large fixed $z$ and sufficiently large $d$,
the coupling and metadata overheads together fit within
the remaining $\varepsilon/2$ margin.
The preceding analysis gives success w.h.p., while
Propositions~\ref{prop:complexity} and~\ref{prop:bit-space}
yield the stated bounds.
\end{proof}

\subsection{Comparison to Spatially Coupled IBLT}
\label{sec:comparison-sc-iblt}

A natural alternative is to equip a conventional singleton IBLT with spatial coupling.
For fixed $k\ge3$, spatial coupling allows its singleton-peeling threshold to approach $c^{\mathsf{orient}}_{k,1}$.
Letting $k\to\infty$ seemingly offers another route to approach $\mathfrak C(d)=(1+\varepsilon)d$ with $O_{\varepsilon}(1)$-time updates.
However, this theoretical promise does not translate into near-optimal communication in practice.
In contrast, XYZ-Sketch offers the following practical advantages:

\textbf{1) False-positive control.}
Recall that the global false-positive probability for XYZ-Sketch is $O_{k,\ell}(d^{1-\ell}/q)$, meaning that XYZ-Sketch is practical even without a fingerprint field.
\footnote{For example, for our recommended parameters $(k,\ell)=(2,6)$, the bound decays as $d^{-5}$, suggesting a very low false-positive rate even at moderate $d$.}
In contrast, a similar argument
\footnote{A false candidate passes rehashing and an independent fingerprint check with probability $O(1/(Mq))$. With $O(M)$ decoding checks, this gives a global bound of $O(1/q)$.
The central difference lies in rehashing verification: XYZ-Sketch typically rehashes $\ell$ times, leading to an $O_{k,\ell}(M^{-\ell})$ false-positive probability compared with the $O_k(M^{-1})$ of IBLT.} only yields a global false-positive bound of $O(1/q)$ for IBLT (and its spatially coupled variant).
This weaker bound makes fingerprints particularly useful for IBLT.
In practice, IBLT typically requires a large fingerprint field to guarantee a lower false-positive rate, incurring higher space overhead.

\textbf{2) Amortized Extra Space Overhead.}
In XYZ-Sketch, up to $\ell$ elements share one counter and one fingerprint. For $\ell=6$, six items share $\log_2(13q)$ auxiliary bits, whereas six singleton cells in IBLT use $6\log_2(3q)$ bits with modulo-$3$ counters.
For large fingerprints, this gives roughly a sixfold saving.

\textbf{3) Spatial locality.}
With our recommended $(k,\ell)=(2,6)$ configuration, each update
accesses only two cells, each storing its polynomial coefficients
contiguously. This reduces scattered memory accesses compared
with IBLTs, especially those using larger $k$.

These advantages are also illustrated by the experiments in Section~\ref{sec:experiments}.

\subsection{Heuristic Choices for Parameters $a$ and $z$}
\label{sec:heuristic-az}

We now discuss how to choose $a$ and $z$ in practice.
For circular placement, we suggest the following heuristic choices:
\begin{equation}
\label{eq:heuristic-az-choice}
a=C\frac{c^{\mathsf{peel}}_{k,\ell}}
        {c^{\mathsf{orient}}_{k,\ell}},
\qquad
z=D(1-a)^{2/3}
\left(\frac{M}{\log(1/\delta)}\right)^{1/3},
\end{equation}
where $\delta$ is the target peeling-failure probability, and constants $0<C<1$ and $D>0$ are calibrated experimentally.

We now explain the intuition.
Spatial coupling starts peeling at the sparsely loaded boundary and lets it propagate inward.
Increasing $a$ improves space utilization but raises the minimum boundary density to $\approx ad/M$.
If even this region is too dense, peeling may fail to start.
Since we aim for $d/M\approx c^{\mathsf{orient}}_{k,\ell}$, this is approximately $a c^{\mathsf{orient}}_{k,\ell}$, which motivates $a c^{\mathsf{orient}}_{k,\ell}\lesssim c^{\mathsf{peel}}_{k,\ell}$, with a constant $C<1$ leaving room for fluctuations.

Once $a$ is chosen, we turn to $z$.
Increasing $z$ reduces the space overhead factor
\footnote{The full coordinate range has length $z+1$,
while the range of $g_0$ has length $z+a$. Their ratio gives this factor.} $(z+1)/(z+a)$ and improves space utilization.
However, it also leaves fewer cells in each coupling window, increasing relative load fluctuations.
More spare capacity is then needed to keep peeling from stalling, and balancing these effects motivates our choice of $z$.
This trade-off follows the same intuition behind Walzer's heuristic $z=\Theta(M^{1/3})$ for conventional spatial coupling~\cite{PeelingOrientabilityThreshold-Walzer-2025}.
Appendix~\ref{app:heuristic-az} develops the argument for circular placement, explicitly accounting for $a$ and $\delta$.
\section{Barriers \& Optimality}
\label{sec:barriers}

The goal of this section is to formalize a model for set reconciliation sketches and to show that, within this model, {\fullalgo} already achieves the optimal asymptotic space bound for fixed $k,\ell$.
In this section, we will only outline proof intuitions and defer the rigorous discussion to Appendices~\ref{app:barriers} and~\ref{app:extremality}.

\subsection{Our Model}
\label{sec:barrier-model}

Recall the one-pass streaming setting defined in Section \ref{sec:background}, where elements of the input set arrive in an arbitrary order and cannot be revisited.
Between updates, the algorithm retains information only in its sketch state.
For simplicity of our proof, we now restrict attention to the insertion-only case, since any algorithm that succeeds w.h.p. must also succeed on insertion-only instances.

Fix $k,\ell\ge 1$.  For every $d,V$, a \emph{fixed-support canonical sketch family} $\mathcal A$ first fixes all public randomness\footnote{After fixing all public randomness, the encoding and decoding are deterministic.} $\omega$ including hashing functions and fingerprint functions, then maintains a sketch containing $M_d$ cells, where $M_d$ depends only on $d$, not on $V$, since the desired space overhead should be at most $O(\mathrm{poly}(d)\log V)$.
For every fixed $d$, each cell has $V^{\ell+o_V(1)}$ possible states, hence occupying $(\ell+o_V(1))\log_2 V$ bits.
\footnote{An $\ell$-word cell may correspond to a cache line, or fit within a single machine word when the universe is small. The $o_V(1)$ term allows extra space for counters and fingerprints, as used in IBLT-style sketches.}
By the definition of $\mathfrak C$, we have $\mathfrak C_{\mathcal A}(d)=\ell M_d$.
We impose the following natural assumptions on a practical streaming algorithm:

\textbf{1) Fixed-support.} After fixing $\omega$, every $x\in\mathcal U^*$ is assigned a predetermined write support $S(x)\subseteq[M_d]$.  An insertion of $x$ may inspect the entire sketch and perform arbitrary computation, but may modify only the cells in this support.
As an $O(1)$-time update can modify only $O(1)$ cells, we focus on the constant-support setting, where $|S(x)|\le k$ for every $x$.

\textbf{2) Canonicality.} The final sketch should depend only on the set, rather than the arriving order.
Only in this way, the encoding function $\mathsf{Enc}^\omega_{d,V}(A)$ is well-defined.

Given $\mathsf{Enc}(A)$ and
$\mathsf{Enc}(B)$, the decoder $\mathsf{Dec}$ may inspect both sketches and perform arbitrary computation to output $(A\setminus B,B\setminus A)$.
Note that we impose no implementation-specific restrictions, such as requiring the decoder to employ a peeling strategy.
We treat the decoder purely as a black box, also imposing no constraints on its time or space complexity.

To simplify the analysis, our lower bounds use the subfamily $(A,B)=(D,\varnothing)$ where $|D|\le d$. Intuitively, a reasonable algorithm should succeed w.h.p. on this simpler subfamily, thus it may serve as a necessary condition.
\footnote{This restriction is also natural from canonicality: writing
$P=A\setminus B$, $Q=B\setminus A$, and $C=A\cap B$, one may start with $(P,\varnothing)$, insert $Q$ on Bob's side, and then insert the common part $C$ on both sides.
Conditioned on $Q$, $C$, and $\omega$, the resulting pair of sketches $(\mathsf{Enc}(A),\mathsf{Enc}(B))$ is thus a deterministic processing of $(\mathsf{Enc}(P),\mathsf{Enc}(\varnothing))$, which cannot reveal additional information about $P$.}
Let $D\sim\mathrm{Unif}\bigl(\binom{\mathcal U^*}{d}\bigr)$.
We define
\begin{equation}
\operatorname{Succ}^\omega_{d,V}(\mathcal A)
:=
\Pr_D\!\left[
\mathsf{Dec}^\omega_{d,V}
\bigl(\mathsf{Enc}^\omega_{d,V}(D),
      \mathsf{Enc}^\omega_{d,V}(\varnothing)\bigr)
=
(D,\varnothing)
\right],
\end{equation}
For brevity, we may slightly abuse notation and write $\mathsf{Dec}(\mathsf{Enc}(D))$ for the first component of $\mathsf{Dec}(\mathsf{Enc}(D),\mathsf{Enc}(\varnothing))$.

We write the success probability $\operatorname{Succ}_{d,V}(\mathcal A):= \mathbb E_\omega[\operatorname{Succ}_{d,V}^\omega(\mathcal A)]$.
We call a family of sketch algorithms $\mathcal A$ \emph{w.h.p. successful} iff
$
\lim\limits_{d\to\infty}\liminf\limits_{V\to\infty}
\operatorname{Succ}_{d,V}(\mathcal A)=1.
$

\subsection{Zero-Error is Impossible}

We first give a conceptually simpler result: zero-error recovery with $O(1)$-support updates is impossible within our model.
The proof is a pigeonhole argument.

\begin{theorem}[Zero-Error Is Impossible]
\label{thm:barrier1-precise}
Fix constants $\ell,k\ge 1$ and $d\ge \ell k+1$.
For every fixed-support canonical sketch family, the following zero-error condition must fail for all sufficiently large $V$:
\begin{equation}
\forall A\in \tbinom{\mathcal U^*}{d}, \qquad\Pr_{\omega} \bigl[
\mathsf{Dec}^{\omega}(\mathsf{Enc}^{\omega}(A))=A
\bigr]=1.
\end{equation}
\end{theorem}


\begin{proof}[Proof sketch]
For fixed $d$, there are at most $N_d:=\sum_{i=0}^{k}\binom{M_d}{i}$ possible write supports, so by pigeonhole principle, some write support is shared by $(V-1)/N_d=\Omega_d(V)$ elements. Their $d$-element subsets give $\Omega_d(V^d)$ inputs, but the corresponding sketches have at most $V^{k\ell+o_V(1)}$ states.
Since $d>k\ell$, the pigeonhole principle rules out zero-error recovery. The full proof is given in Appendix~\ref{app:proof-zero-error}.
\end{proof}

\subsection{From Orientability to Optimality}
\label{sec:orientability-necessary}

\paragraph{Orientability is Necessary for Success}
We first show that successful recovery generally requires the hypergraph induced by the current input set to be $\ell$-orientable.
Formally, for $D\in\binom{\mathcal U^*}{d}$, let $H^\omega_{d,V}(D):=\bigl([M_d],\{S(x):x\in D\}\bigr)$ be the induced multihypergraph, where parallel edges are allowed.

\begin{lemma}[Key Lemma: Successful Recovery Requires Orientability]
\label{lem:orientability-necessary}
For every fixed $d$ and every fixed choice $\omega$ of the public randomness,
\begin{equation}
\Pr_D\!\left[
\mathsf{Dec}^{\omega}\bigl(\mathsf{Enc}^{\omega}(D)\bigr)=D
\ \land\
H^\omega_{d,V}(D)\text{ is not $\ell$-orientable}
\right]
=o_V(1),
\end{equation}
where the $o_V(1)$ term is uniform in $\omega$. Hence,
\begin{equation}
\forall \omega,\quad
\operatorname{Succ}^{\omega}_{d,V}(\mathcal A)
\le
\Pr_{D}\!\left[
H^\omega_{d,V}(D)\text{ is $\ell$-orientable}
\right]
+o_V(1).
\end{equation}
\end{lemma}
Note that the necessary condition here is orientability instead of peelability, which is consistent with the $\ell / c^{\mathsf{orient}}_{k,\ell}$ boundary.
\begin{proof}[Proof sketch.]
By Hall's theorem, a non-$\ell$-orientable hypergraph must contain an overloaded "core".
Specifically, this core consists of $t>\ell r$ item, while their write-support are confined to $r$ cells.
Those cells contain only $(\ell r+o_V(1))\log_2 V$ bits of information and therefore cannot distinguish all $V^{t-o_V(1)}$ cases.
The full proof is given in Appendix~\ref{app:proof-orientability}.
\end{proof}

\paragraph{Uniform Distribution Maximizes Orientability Threshold.}
For a distribution $\mu$ on $\binom{[n]}{\le k}$, let
$H^\mu_{n,m}$ be the multihypergraph obtained by sampling $m$
independent edges from $\mu$. Let $\alpha_\ell(H)$ denote the maximum number
of edges in an $\ell$-orientable subhypergraph of $H$, and let $\mu^*$
be the uniform distribution on $\binom{[n]}{k}$.

\begin{theorem}[Uniform Distribution Maximizes Expected Orientable Size]
\label{thm:uniform-extremality}
For every $n\ge k\ge1$, $\ell\ge1$, $m\ge0$, and every distribution $\mu$ on $\binom{[n]}{\le k}$, we have $\mathbb E[\alpha_\ell(H^\mu_{n,m})]\le \mathbb E[\alpha_\ell(H^{\mu^*}_{n,m})]$.
Consequently, for every fixed $k,\ell\ge1$ and $\varepsilon>0$,
\begin{equation}
\sup_{\mu}
\Pr\!\left[
H^\mu_{n,\lceil(c^{\mathsf{orient}}_{k,\ell}+\varepsilon)n\rceil}
\text{ is $\ell$-orientable}
\right]
=o_n(1),    
\end{equation}
where the supremum ranges over all distributions on
$\binom{[n]}{\le k}$.
\end{theorem}

This theorem proves a strengthened form of the informal conjecture of
Walzer~\cite[Conclusion]{PeelingOrientabilityThreshold-Walzer-2025}.
We defer the proof overview to Section~\ref{sec:extremality-overview}
and the full proof to Appendix~\ref{app:extremality}.

\paragraph{Optimality Results.}

In fact, Theorem~\ref{thm:uniform-extremality} immediately implies the optimality of {\fullalgo}.
\begin{theorem}[Optimality of {\fullalgo}]
\label{thm:main-precise}
For every fixed
$k,\ell\ge2$, every w.h.p. successful fixed-support
canonical sketch family $\mathcal A$ satisfies
\begin{equation}
\liminf_{d\to\infty}
\frac{\mathfrak C_{\mathcal A}(d)}{d}
\ge
\frac{\ell}{c^{\mathsf{orient}}_{k,\ell}}
\ge
\limsup_{d\to\infty}
\frac{\mathfrak C_{XYZ}(d)}{d}.
\end{equation}
Hence, {\fullalgo} is optimal within our model. In particular, achieving $\limsup_{d\to\infty}\mathfrak C_{\mathcal A}(d)/d \le1+\varepsilon$ requires $k\ell\ge\ln(1+1/\varepsilon)$.
\end{theorem}

Intuitively, Lemma~\ref{lem:orientability-necessary} reduces successful recovery to orientability, while Theorem~\ref{thm:uniform-extremality} bounds the achievable edge density by $c^{\mathsf{orient}}_{k,\ell}$, giving the space lower bound $\ell/c^{\mathsf{orient}}_{k,\ell}$.
{\fullalgo} attains this bound by Lemma~\ref{thm:orientability-benchmark}.
The full proof is given in Appendix~\ref{app:proof-optimality}.

\begin{corollary}
\label{thm:barrier2-precise}
For every fixed $k\ge2$ and $\ell\ge1$, every
w.h.p. successful fixed-support canonical sketch family satisfies $\liminf_{d\to\infty}\frac{\mathfrak C_{\mathcal A}(d)}{d}>1$.
\end{corollary}

The strict inequality follows from $c^{\mathsf{orient}}_{k,\ell}<\ell$ for every fixed $k\ge2$ and $\ell\ge1$. Note that the argument above also applies when $\ell=1$ for the lower-bound side (see Appendix~\ref{app:proof-overhead}).

\subsection{Proof Overview of Theorem \ref{thm:uniform-extremality}}
\label{sec:extremality-overview}

\paragraph{Remarks.}
Our argument uses a parallel-class decomposition underlying the discrete-Hessian analysis of matroid rank functions, which is developed in algorithmic mechanism design~\cite[Claim~4.5]{DughmiRoughgardenYan2011}.
Moreover, a similar argument also appears in Brändén-Huh's work \cite[Lemma~3.11]{BrandenHuh2020} on Lorentzian polynomials.
In this paper, we adopt their techniques to probabilistic data structure.
We next provide a condensed proof sketch, while readers interested in the full details are referred to Appendix~\ref{app:extremality}.

\begin{proof}[Proof Sketch.]
Let $\alpha_\ell(H)$ be the maximum number of edges of $H$ that can be $\ell$-oriented.
For a distribution $\mu$ on $\binom{[n]}{k}$, write $R(\mu):=\mathbb E[\alpha_\ell(H^\mu_{n,m})]$.
The main idea is to show that $R$ is concave in $\mu$.
Since relabeling vertices leaves $R$ unchanged, averaging over vertex permutations then immediately shows that the uniform distribution maximizes it under its concavity.
For $m\ge2$, fix a direction $h$ on the probability simplex (which means $\sum_{e\in\binom{[n]}{k}} h_e=0$) and let
$F(t):=R(\mu+th)$.
Differentiating twice gives
\begin{equation}
F''(0)
=
m(m-1)\mathbb E_{H\sim H^\mu_{n,m-2}}\!\left[
\sum_{e,f\in\binom{[n]}{k}}
h_eh_f\,\alpha_\ell(H+e+f)
\right],
\end{equation}
where $H+e+f$ denotes adding two edge occurrences to $H$.

To simplify this expression, we employ techniques from matroid theory.
In particular, we represent orientable edge sets as matchable sets in a transversal matroid and contract the elements corresponding to $H$ (see Appendix~\ref{app:matroid} for these matroid notions).
The resulting quadratic form can then be grouped by non-loop parallel classes (Lemma \ref{lem:extremality-quadratic-form}), yielding
\begin{equation}
F''(0)
=
m(m-1)\mathbb E_{H\sim H^\mu_{n,m-2}}\!\left[
-\sum_{C\in\mathcal P_H}
\left(\sum_{e\in C}h_e\right)^2
\right]
\le0,    
\end{equation}
where $\mathcal P_H$ consists of the contributing parallel classes, identified with their sets of supports.
As each class contributes a negative square, the quadratic form is negative semi-definite, proving concavity.

As the uniform distribution $\mu^*$ is optimal, no other distribution can have an orientability threshold exceeding $c^{\mathsf{orient}}_{k,\ell}$.
Combining expectation extremality with a concentration inequality gives the $o_n(1)$ bound above this threshold.
\end{proof}
\section{Experiments}
\label{sec:experiments}

\subsection{Experimental Setup}
\label{sec:experimental-setup}

We use shared datasets $A,B$ with $|(A\setminus B)\cup (B\setminus A)|=d$, split the difference nearly evenly between the two directions, and shuffle insertion
orders.
Unless otherwise stated, $|A|=|B|=10^7$, and each element is represented using $w=30$ bits. XYZ-Sketch uses $\mathcal U=\mathbb F_{2^{30}}$ and $q=1$.
We measure practical communication by
\begin{equation}
    \mathfrak R:=\dfrac{\text{transmitted bits}}{dw},
\end{equation}
which is an empirical counterpart of $\mathfrak C(d)/d$.

Baselines are MiniSketch~\cite{MiniSketch-BitcoinCore-2018},
CPISync~\cite{CPISync-TrachtenbergStarobinskiAgarwal-2002},
Rateless IBLT~\cite{RatelessSetReconciliation-YangGiladAlizadeh-2024},
and IBLT~\cite{IBLT_Cplusplus-Andresen-2014}.
We also implement IBLT+SC by adding standard spatial coupling(SC)
to the same IBLT implementation.
Appendix~\ref{app:experiment-details} gives implementation and
measurement details.

\subsection{Thresholds and Parameter Choices}
\label{sec:experimental-parameters}

\begin{figure}[htbp]
    \centering
    \begin{minipage}[t]{0.32\linewidth}
        \vspace{0pt}\centering
        \includegraphics[width=\linewidth]{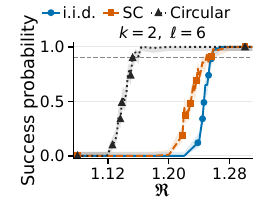}\par
        {\small (a) Sharp threshold\par}
    \end{minipage}\hfill
    \begin{minipage}[t]{0.66\linewidth}
        \vspace{0pt}\centering
        \includegraphics[width=\linewidth]{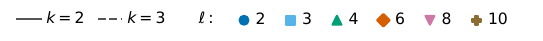}\par\nointerlineskip
        \begin{minipage}[t]{0.485\linewidth}
            \centering
            \includegraphics[width=\linewidth]{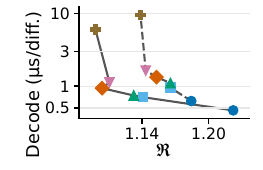}\par
            {\small (b) $k,\ell$: decoding tradeoff\par}
        \end{minipage}\hfill
        \begin{minipage}[t]{0.485\linewidth}
            \centering
            \includegraphics[width=\linewidth]{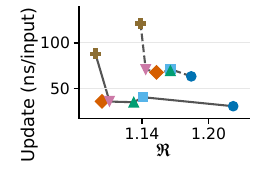}\par
            {\small (c) $k,\ell$: update tradeoff\par}
        \end{minipage}
    \end{minipage}

    \par\smallskip
    \begin{minipage}[t]{0.32\linewidth}
        \centering
        \includegraphics[width=\linewidth]{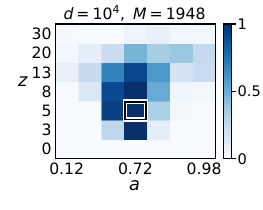}\par
        {\small (d) $a,z$ heatmap\par}
    \end{minipage}\hfill
    \begin{minipage}[t]{0.66\linewidth}
        \centering
        \includegraphics[width=\linewidth]{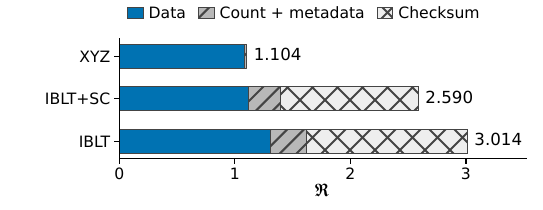}\par
        {\small (e) Communication breakdown at $d=10^5$\par}
    \end{minipage}

    \par\smallskip
    \includegraphics[width=\linewidth]{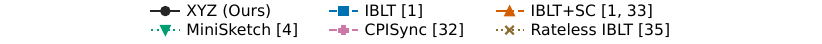}\par
    \begin{minipage}[t]{0.32\linewidth}
        \centering
        \includegraphics[width=\linewidth]{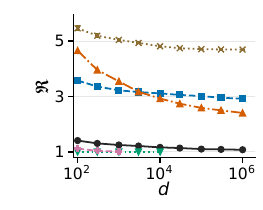}\par
        {\small (f) Communication\par}
    \end{minipage}\hfill
    \begin{minipage}[t]{0.32\linewidth}
        \centering
        \includegraphics[width=\linewidth]{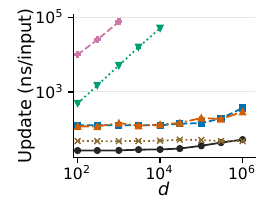}\par
        {\small (g) Update time\par}
    \end{minipage}\hfill
    \begin{minipage}[t]{0.32\linewidth}
        \centering
        \includegraphics[width=\linewidth]{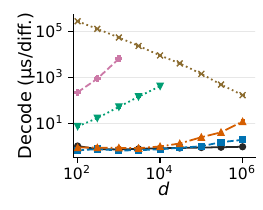}\par
        {\small (h) Decoding time\par}
    \end{minipage}

    \caption{Experimental results.
    (a) Sharp thresholds at $(k,\ell)=(2,6)$ and $d=10^4$.
    (b)--(c) $k\in\{2,3\}$, $\ell\in\{2,3,4,6,8,10\}$,
    $d=10^5$. Shapes and colors denote $\ell$; solid and dashed
    lines denote $k=2$ and $k=3$, respectively, connecting increasing $\ell$.
    (d) $a,z$ Heatmap at $d=10^4,M=1948$. The box marks our heuristic choice.
    (e) Communication breakdown at $d=10^5$. XYZ saves most of the  auxiliary fields costs.
    (f)--(h) End-to-End Tradeoff, where XYZ uses $(k,\ell)=(2,6)$.
    Fixed-size configurations target $90\%$ success.
    Probability bands and timing bars are $95\%$ Wilson and
    bootstrap intervals, respectively.}
    \label{fig:experimental-results}
\end{figure}

We calibrate $C,D$ using only smaller-scale instances, obtaining $C\approx 0.875$ and $D\approx1.188$, and then keep them fixed for the remaining experiments.
Figure~\ref{fig:experimental-results}(a) shows sharp success transitions, with circular SC improving over i.i.d.\ and standard SC.
Panel (d) shows that our heuristic choice of $(a,z)$ lies within the high-success region, suggesting that the rule works well in practice.

Panels (b)--(c) show the effect of $k$ and $\ell$.
For $k=2$, increasing $\ell$ from $2$ to $6$ reduces communication;
further increases yield little additional saving and cost more
decoding time in our implementation.
Together, these results suggest that $(k,\ell)=(2,6)$ offers a good balance, motivating its use in our main comparison.
Additional heatmaps appear in Appendix~\ref{app:experiment-heatmaps}.

\subsection{End-to-End Tradeoffs}
\label{sec:end-to-end}

Fix $(k,\ell)=(2,6)$ for XYZ-Sketch.
Figures~\ref{fig:experimental-results}(f)--(h) compare communication and computation over $d=10^2$--$10^6$.
XYZ's communication ratio decreases from $1.411$ to $1.079$, approaching MiniSketch's unit ratio.
At $d=10^5$, XYZ uses $57.4\%$ less communication than IBLT+SC, while maintaining a comparable computational cost.
Panel (e) shows that most of this gap comes from auxiliary fields, illustrating the arguments in Section~\ref{sec:comparison-sc-iblt}.
Here, IBLT uses a 32-bit checksum (the counterpart of XYZ's fingerprint). Substantially shortening it reduces decoding success, as illustrated in Appendix~\ref{app:experiment-fingerprint}.

Our implementation uses binary-field arithmetic and specialized low-degree polynomial routines to improve computational efficiency.
At $d=10^6$, XYZ takes $52.5\,\mathrm{ns}$ per update and $1.01\,\mu\mathrm{s}$ per difference element for decoding, which is even faster than $372.3\,\mathrm{ns}$ and $2.02\,\mu\mathrm{s}$ for IBLT.
Thus, the communication savings need not come at the cost of slower updates or decoding.
\section{Conclusion}
\label{sec:conclusion}

We studied whether set reconciliation can approach optimal space
while retaining constant-time updates.
Our positive answer is {\fullalgo}, which combines high-capacity
cells implemented via D-RFR with spatial coupling.
Under the random-label setting, it achieves
$(1+\varepsilon)d$ space with $O(\log(1/\varepsilon))$-time updates.
Our experiments demonstrate practical communication savings
while maintaining efficient updates and decoding.

Our negative results establish a matching boundary in the
fixed-support canonical model.
For fixed $k,\ell\ge 2$, the optimal asymptotic space ratio is
$\ell/c^{\mathsf{orient}}_{k,\ell}>1$, and {\fullalgo} attains it.
Moreover, approaching $(1+\varepsilon)d$ space requires
$k\ell=\Omega(\log(1/\varepsilon))$, matching our parameter choice.
Thus, the positive and negative results describe two sides of
the same frontier: optimal space cannot be approached with
fixed parameters, but can be approached as the parameters grow.

A natural question is to understand how the frontier changes beyond the fixed-support canonical model.
It would also be interesting to extend high-capacity cells and our lower-bound technique to multi-party and rateless set reconciliation, and to explore their use in other sketching and probabilistic data structures.

\section{Methods: Use and Disclosure of AI and LLMs}

The authors developed the core algorithmic idea and the initial C++ implementation.
The proof ideas were primarily developed by the authors, except for the main proof ideas of Theorem \ref{thm:uniform-extremality} and Lemma \ref{lem:app-analysis-poly-spread}, which were provided by generative AI.
Generative AI, including ChatGPT and Codex, also assisted with paper writing, literature review, refining proof details, performance optimizations, experiment scripting, execution, analysis, and plotting.
All claims and outputs were independently reviewed and verified by the authors.

\bibliographystyle{ACM-Reference-Format}
\bibliography{bibfile}

\newpage

\appendix
\section*{Appendix}

The appendices are organized as follows.
Appendix~\ref{app:tables} collects notation and threshold values.
Appendix~\ref{app:drfr} gives the D-RFR reconstruction algorithm
and its proofs.
Appendices~\ref{app:analysis}, \ref{app:barriers},
and~\ref{app:extremality} provide proofs for the theoretical analysis,
the lower bounds, and the uniform-distribution extremality theorem,
respectively.
Appendix~\ref{app:experiments} provides supplementary experimental details.
Appendix~\ref{app:pseudocode} gives additional pseudocode.

\section{Tables}
\label{app:tables}

\subsection{Notation Table}
\label{app:notation-table}

The main notation used in the paper is summarized below.

\begingroup
\small
\renewcommand{\arraystretch}{1.15}
\begin{longtable}{@{}p{0.23\textwidth}p{0.72\textwidth}@{}}
\caption{Main notation.}
\label{tab:notation}\\
\toprule
Notation & Meaning \\
\midrule
\endfirsthead

\multicolumn{2}{l}{\tablename~\thetable\ (continued)}\\
\toprule
Notation & Meaning \\
\midrule
\endhead

\bottomrule
\endfoot

\multicolumn{2}{@{}l}{\textbf{Problem setting and sketch parameters}}\\
\addlinespace

$A,B$
& Sets held by Alice and Bob. \\

$\mathcal U,\mathcal U^*$
& Element universe $\mathcal U$ and its nonzero
  subset $\mathcal U^*=\mathcal U\setminus\{0\}$. \\

$V$
& Universe size. \\

$d$
& Known upper bound on $|A\setminus B|+|B\setminus A|$. \\

$\mathfrak C(d)$
& Communication cost in $\mathcal U$-word equivalents:
  the sketch occupies
  $(\mathfrak C(d)+o_V(1))\log_2 V$ bits. \\

$\varepsilon,\delta$
& Target relative space overhead and failure probability,
  respectively. \\

$k$
& Number of placement hashes in {\fullalgo}; maximum write-support
  size in the fixed-support model. \\

$\ell$
& Maximum number of signed elements recoverable from a pure cell;
  also the leading number of $\mathcal U$-words stored per cell. \\

$M,M_d$
& Number of cells; the subscript emphasizes dependence on $d$. \\

$z,a$
& Spatial-coupling parameter and circular-placement parameter,
  respectively; $a=0$ gives conventional spatial coupling. \\

$h_i,h'_i,h_i^{(a)}$
& Placement hashes for independent hashing, conventional spatial
  coupling, and circular placement, respectively. \\

$\Gamma(x)$
& Set of distinct cells to which element $x$ is mapped. \\

\midrule
\multicolumn{2}{@{}l}{\textbf{Cell representation and reconstruction}}\\
\addlinespace

$\mathcal B_A,\mathcal B_B,\mathcal D$
& Sketches of $A$ and $B$, and their cell-wise difference. \\

$f_C(x)$
& Signed contribution of element $x$ in cell $C$,
  taking values in $\{-1,0,1\}$. \\

$\Delta_C^+,\Delta_C^-$
& Sets of positive and negative elements remaining in cell $C$. \\

$c_C$
& Signed count in cell $C$, stored modulo $2\ell+1$. \\

$P_C(Z)$
& Polynomial stored in cell $C$:
  $\prod_x(Z-x)^{f_C(x)}\bmod Z^\ell$. \\

$\mathsf{fp}_C$
& Optional fingerprint stored in cell $C$. \\

$h_{\mathrm{fp}},q$
& Independent fingerprint hash and its modulus;
  $q=1$ denotes the version without fingerprints. \\

$S,T$
& Candidate sets of positive and negative elements returned
  by cell reconstruction. \\

$Z$
& Formal variable for polynomials. \\

$\chi_S(Z)$
& Characteristic polynomial of $S$:
  $\chi_S(Z)=\prod_{x\in S}(Z-x)$. \\

$F,G,R$
& Numerator, denominator, and truncated quotient in rational
  function reconstruction, with $F\equiv GR\pmod{Z^\ell}$
  in D-RFR. \\

$m$ (D-RFR)
& Known degree difference $\deg F-\deg G$. \\


$\bot$
& Failure output. \\

\midrule
\multicolumn{2}{@{}l}{\textbf{Hypergraphs and thresholds}}\\
\addlinespace

$H=(\mathcal V,\mathcal E)$
& Hypergraph with vertex set $\mathcal V$ and edge multiset
  $\mathcal E$. \\

$n,m$ (hypergraphs)
& Numbers of vertices and hyperedges, respectively,
  with repeated edges counted separately. \\

$c$
& Edge density $m/n$ for uncoupled hypergraphs;
  the load parameter in the coupled construction. \\

$H^{(k)}_{n,m}$
& Random hypergraph with $n$ vertices and $m$ independently
  sampled uniform $k$-element edges. \\

$c^{\mathsf{peel}}_{k,\ell}$,
$c^{\mathsf{orient}}_{k,\ell}$
& Density thresholds for $\ell$-peelability and
  $\ell$-orientability of the uniform random hypergraph. \\

$\alpha_\ell(H)$
& Maximum number of edges in an $\ell$-orientable
  subhypergraph of $H$. \\

$\mu,\mu^*$
& Distribution on supports of size at most $k$, and the uniform
  distribution on $k$-element supports, respectively. \\

$H^\mu_{n,m}$
& Random hypergraph with $n$ vertices and $m$ edges sampled
  independently from $\mu$. \\

\midrule
\multicolumn{2}{@{}l}{\textbf{Fixed-support canonical model}}\\
\addlinespace

$\mathcal A$
& Sketch family. \\

$\omega$ & The public randomness, including hashing functions, fingerprint functions,\textit{etc}. \\

$S(x)$
& Predetermined write support of element $x$, with $|S(x)|\le k$. \\

$\mathsf{Enc},\mathsf{Dec}$
& Sketch encoding and reconciliation decoding functions. \\

$D$
& Uniformly sampled $d$-element set used in the lower bounds,
  reconciled against the empty set. \\

$H^\omega_{d,V}(D)$
& Hypergraph induced by the write supports of elements in $D$,
  after fixing $\omega$. \\

$\operatorname{Succ}_{d,V}(\mathcal A)$
& Recovery probability for $(D,\varnothing)$, averaged over
  $D$ and the public randomness. \\

\end{longtable}
\endgroup

\subsection{Threshold Table}
\label{app:threshold-table}


Table~\ref{tab:peel-orient-thresholds} lists the peelability and
orientability thresholds of fully random $k$-uniform hypergraphs.
The corresponding asymptotic space ratios $\ell/c^{\mathsf{peel}}_{k,\ell}$ and $\ell/c^{\mathsf{orient}}_{k,\ell}$ are listed in Table~\ref{tab:peel-orient-space-ratios}.

Specifically, let $Q(\xi,r):=\Pr[\operatorname{Pois}(\xi)\ge r]$.
Following results from~\cite{OrientabilityThresholds-FountoulakisKhoslaPanagiotou-2016,
PeelingOrientabilityThreshold-Walzer-2025}, for $k\ge 2$ and $(k,\ell)\ne(2,1)$, the thresholds are
\begin{equation}
\begin{aligned}
c^{\mathsf{peel}}_{k,\ell}
&=\min_{\xi>0}\frac{\xi}{kQ(\xi,\ell)^{k-1}},\\
c^{\mathsf{orient}}_{k,\ell}
&=\frac{\xi^*}{kQ(\xi^*,\ell)^{k-1}},
\qquad
\xi^*Q(\xi^*,\ell)=k\ell Q(\xi^*,\ell+1),
\end{aligned}
\end{equation}
where $\xi^*$ is the positive solution of the displayed equation.
For $(k,\ell)=(2,1)$, the values are
$c^{\mathsf{peel}}_{2,1}=0$ and
$c^{\mathsf{orient}}_{2,1}=1/2$.
For $k=1$, both thresholds are zero: at any fixed positive
edge density, some vertex receives more than $\ell$ singleton
edges with high probability, preventing both peeling and orientation.

\begin{table}[htbp]
\centering
\caption{Peelability and orientability thresholds at density $c=m/n$.
Each entry lists $c^{\mathsf{peel}}_{k,\ell}$ followed by
$c^{\mathsf{orient}}_{k,\ell}$, separated by a slash.
Values are rounded to four decimal places.}
\label{tab:peel-orient-thresholds}
\footnotesize
\setlength{\tabcolsep}{2.7pt}
\renewcommand{\arraystretch}{1.10}
\begin{tabular}{@{}c*{6}{c}@{}}
\toprule
$\ell\backslash k$ & $2$ & $3$ & $4$ & $5$ & $6$ & $7$ \\
\midrule
1 & 0.0000/0.5000 & 0.8185/0.9179 & 0.7723/0.9768 & 0.7018/0.9924 & 0.6371/0.9974 & 0.5818/0.9991 \\
2 & 1.6755/1.7940 & 1.5528/1.9764 & 1.3336/1.9965 & 1.1578/1.9994 & 1.0216/1.9999 & 0.9146/2.0000 \\
3 & 2.5747/2.8775 & 2.1745/2.9919 & 1.8109/2.9994 & 1.5457/3.0000 & 1.3488/3.0000 & 1.1977/3.0000 \\
4 & 3.3996/3.9215 & 2.7467/3.9970 & 2.2498/3.9999 & 1.9022/4.0000 & 1.6492/4.0000 & 1.4574/4.0000 \\
5 & 4.1827/4.9478 & 3.2894/4.9989 & 2.6654/5.0000 & 2.2395/5.0000 & 1.9333/5.0000 & 1.7030/5.0000 \\
6 & 4.9376/5.9644 & 3.8117/5.9996 & 3.0651/6.0000 & 2.5635/6.0000 & 2.2060/6.0000 & 1.9386/6.0000 \\
7 & 5.6721/6.9754 & 4.3189/6.9998 & 3.4528/7.0000 & 2.8776/7.0000 & 2.4703/7.0000 & 2.1668/7.0000 \\
8 & 6.3905/7.9828 & 4.8143/7.9999 & 3.8312/8.0000 & 3.1840/8.0000 & 2.7279/8.0000 & 2.3892/8.0000 \\
\bottomrule
\end{tabular}
\end{table}

\begin{table}[htbp]
\centering
\caption{Asymptotic space ratios corresponding to
Table~\ref{tab:peel-orient-thresholds}.
Each entry lists $\ell/c^{\mathsf{peel}}_{k,\ell}$ followed by
$\ell/c^{\mathsf{orient}}_{k,\ell}$, separated by a slash.
Values are rounded to four decimal places;
$\infty$ reflects $c^{\mathsf{peel}}_{2,1}=0$.}
\label{tab:peel-orient-space-ratios}
\footnotesize
\setlength{\tabcolsep}{2.7pt}
\renewcommand{\arraystretch}{1.10}
\begin{tabular}{@{}c*{6}{c}@{}}
\toprule
$\ell\backslash k$ & $2$ & $3$ & $4$ & $5$ & $6$ & $7$ \\
\midrule
1 & $\infty$/2.0000 & 1.2218/1.0894 & 1.2949/1.0238 & 1.4249/1.0076 & 1.5697/1.0026 & 1.7189/1.0009 \\
2 & 1.1937/1.1148 & 1.2880/1.0119 & 1.4997/1.0018 & 1.7275/1.0003 & 1.9577/1.0000 & 2.1868/1.0000 \\
3 & 1.1652/1.0426 & 1.3796/1.0027 & 1.6567/1.0002 & 1.9409/1.0000 & 2.2242/1.0000 & 2.5049/1.0000 \\
4 & 1.1766/1.0200 & 1.4563/1.0007 & 1.7780/1.0000 & 2.1029/1.0000 & 2.4254/1.0000 & 2.7445/1.0000 \\
5 & 1.1954/1.0106 & 1.5201/1.0002 & 1.8759/1.0000 & 2.2327/1.0000 & 2.5863/1.0000 & 2.9361/1.0000 \\
6 & 1.2152/1.0060 & 1.5741/1.0001 & 1.9575/1.0000 & 2.3406/1.0000 & 2.7199/1.0000 & 3.0951/1.0000 \\
7 & 1.2341/1.0035 & 1.6208/1.0000 & 2.0273/1.0000 & 2.4326/1.0000 & 2.8337/1.0000 & 3.2306/1.0000 \\
8 & 1.2518/1.0022 & 1.6617/1.0000 & 2.0881/1.0000 & 2.5126/1.0000 & 2.9327/1.0000 & 3.3484/1.0000 \\
\bottomrule
\end{tabular}
\end{table}
\section{Implementation of the D-RFR Procedure}
\label{app:drfr}


We now present the reconstruction procedure and uniqueness proof for
Theorem~\ref{thm:d-rfr} in Section~\ref{sec:d-rfr}.
The complete pseudocode is given in Appendix~\ref{app:drfr-pseudocode}.
Throughout this section, $\deg 0:=-\infty$.
Let $\mathsf M(n)$ denote the cost of multiplying two degree-$n$ polynomials in field operations. The reconstruction takes $O(\mathsf M(n)\log n)$ field operations, which is $O(n\log^2 n)$ for the common case of $\mathsf M(n)=O(n\log n)$.

\subsection{Proof of Uniqueness}

We first prove the uniqueness as follows.

\begin{proof}
Suppose that $(F_0,G_0)$ and $(F_1,G_1)$ are two valid solutions. We have
\begin{equation}
F_0(Z)G_1(Z)\equiv F_1(Z)G_0(Z)\pmod{Z^\ell}.
\end{equation}
Write $a_j:=\deg F_j$ and $b_j:=\deg G_j$ for $j\in\{0,1\}$.
Since $a_j-b_j=m$ and $a_j+b_j\le\ell$, we have
\begin{equation}
\deg(F_0G_1)
=
a_0+b_1
=
m+b_0+b_1
=
a_1+b_0
=
\deg(F_1G_0)
\le \ell.
\end{equation}
Both products are monic, so their leading terms cancel. Hence $\deg(F_0G_1-F_1G_0)<\ell.$
Since this difference is divisible by $Z^\ell$, it must be zero. Thus,
$F_0G_1=F_1G_0$. As $F_0$ and $G_0$ are coprime, this implies
$F_0\mid F_1$ and $G_0\mid G_1$. Since $F_1$ and $G_1$ are also coprime,
the common quotient must be a constant; monicity then gives
$F_0=F_1$ and $G_0=G_1$.
\end{proof}

\subsection{The Original RFR}

We first introduce the original RFR as follows.
\begin{lemma}[The original RFR algorithm~\cite{RFR-KhodadadMonagan-2006,RationalReconstruction-CollinsEncarnacion-1995}]
\label{lem:original-rfr}
    Let $N=2D+1$ be odd, and let $R(Z)$ be a polynomial with
    $\deg R<N$. Suppose that there exist coprime polynomials $F,G\in
    \mathcal U[Z]$ such that
    \begin{itemize}
        \item $\deg F\le D$, $\deg G\le D$, and $G(0)\ne0$;
        \item $\dfrac{F(Z)}{G(Z)}\equiv R(Z)\pmod{Z^N}$.
    \end{itemize}
    Then the original RFR algorithm recovers $F$ and $G$, up to
    multiplication by a common nonzero constant, in
    $O(N\log^2 N)$ field operations.
    
    In particular, if $G$ is known to be monic, dividing both recovered polynomials by the leading coefficient of the recovered denominator yields the real $F$ and $G$.
\end{lemma}

\begin{proof}
    If $R=0$, the coprimality and degree assumptions imply that
    $(F,G)$ is a nonzero scalar multiple of $(0,1)$, so we return $(0,1)$.
    We therefore assume $R\ne0$ below. Following the definition,
    \begin{equation}
    \begin{gathered}
    \frac{F(Z)}{G(Z)}\equiv R(Z)\pmod{Z^N}
    \\
    \Updownarrow
    \\
    F(Z)=G(Z)R(Z)+K(Z)Z^N
    \quad\text{for some }K\in\mathcal U[Z].
    \end{gathered}    
    \end{equation}

   Thus, our goal is to find $G(Z)$ and $K(Z)$ such that
    $\deg G\le D$ and
    $\deg(G(Z)R(Z)+K(Z)Z^{N})\le D$.

    Denoting $A:=Z^{N}$ and $B:=R(Z)$, we use an approach inspired by
    the Euclidean algorithm and formulate a recurrence relation as
    follows:
    \begin{equation}
    \begin{aligned}
    H_0(Z)&:=A,\qquad H_1(Z):=B,\\
    H_i(Z)&:=H_{i-2}(Z)\bmod H_{i-1}(Z),\qquad i\ge2,\\
    Q_i(Z)&:=\frac{H_{i-2}(Z)-H_i(Z)}{H_{i-1}(Z)}.
    \end{aligned}
    \end{equation}

    We continue the definition until $H_i(Z)=0$. Since
    $\mathcal U[Z]$ is a Euclidean domain, the process must end. It is easy to show that $\deg H_i<\deg H_{i-1}$ and $Q_i\in\mathcal U[Z];\ \deg Q_i=\deg H_{i-2}-\deg H_{i-1}$.
    We define
    $\mathbf G_i:=
        \begin{bmatrix}H_{i-1}\\H_i\end{bmatrix}
        \in \mathrm M_{2\times1}(\mathcal U[Z]).$
    Then by definition we have
    \begin{equation}
        \mathbf G_i=
        \begin{bmatrix}0&1\\1&-Q_i(Z)\end{bmatrix}
        \mathbf G_{i-1}.
    \end{equation}
    For convenience, denote
    $\begin{bmatrix}0&1\\1&-Q_i(Z)\end{bmatrix}$ by
    $\varphi(Q_i)$. We find $p$ such that
    $\deg H_{p-1}>D,\;\deg H_p\le D.$
    Then
    \begin{equation}
    \label{eq:stopping}
        \mathbf G_p=
        \varphi(Q_p)\varphi(Q_{p-1})\cdots\varphi(Q_2)
        \begin{bmatrix}A\\B\end{bmatrix}.
    \end{equation}
    We denote the matrix $\varphi(Q_p)\cdots\varphi(Q_2)$ by $\mathcal H_N(A,B)$.
    More generally, when the subscript is an arbitrary integer $t\ge1$, $\mathcal H_t(A,B)$ denotes the similar Euclidean prefix with the stopping threshold $\lceil t/2\rceil-1$ (\ie,\; $\deg H_{p-1}>\lceil t/2\rceil-1;\deg H_p\le \lceil t/2\rceil-1$).
    For every entry $x(Z)$ in matrix $\mathcal H_N(A,B)$,
    \begin{equation}
    \label{eq:half-gcd-degree-bound}
    \deg x
        \le \sum_{i=2}^{p}\deg Q_i
         =\deg H_0-\deg H_{p-1}
        \le N-(D+1)=D.
    \end{equation}
    Denoting
    \begin{equation}
        \mathcal H_N(A,B)=
        \begin{bmatrix}
        x_0(Z)&y_0(Z)\\
        x_1(Z)&y_1(Z)
        \end{bmatrix}.
    \end{equation}
    It follows that
    \begin{equation}
        H_p(Z)=x_1(Z)Z^N+y_1(Z)R(Z).
    \end{equation}
    Since $\deg H_p,\deg x_1,\deg y_1\le D$, the pair
    $(H_p,y_1)$ satisfies the degree bounds and the polynomial congruence required in our lemma statements.

    It remains to verify that $y_1(0)\ne0$ and $\gcd(H_p,y_1)=1$.
    From $H_p\equiv y_1R\pmod{Z^N}$ and $F\equiv GR\pmod{Z^N}$,
    we obtain
    \begin{equation}
        H_pG\equiv y_1F\pmod{Z^N}
    \end{equation}
    Since both products have degree at most $2D<N$, this gives
    $H_pG=y_1F$.
    Since $\gcd(F,G)=1$, we have $H_p=CF$ and $y_1=CG$ for some polynomial $C$.

    Write $F-GR=KZ^N$. Substitution into
    $x_1Z^N=H_p-y_1R=C(F-GR)$ gives $x_1=CK$.
    Since
    \begin{equation}
    x_0y_1-y_0x_1=\det\mathcal H_N(A,B)=(-1)^{p-1}=\pm1
    \end{equation}
    and $x_0y_1-y_0x_1=C(x_0G-y_0K)$, the common factor $C$ must be a nonzero constant.
    Hence $y_1(0)=G(0)C\ne0$ and $\gcd(H_p,y_1)=1$.
    If $G$ is monic, simply divide $H_p,y_1$ by the leading coefficient of $y_1$ to recover $F,G$.

    Next, we illustrate how to efficiently compute
    $\mathcal H_N(A,B)$ for $A,B\in\mathcal U[Z]$ with
    $\deg A=N>\deg B$. Set $s:=\lceil N/2\rceil$.
    If $N\le 1$ or $\deg B<s$, return the $2\times2$ identity matrix.
    Otherwise, let $L$ be an integer with $0<L<N$, and set $n_1:=N-L, s_1:=\lceil n_1/2\rceil$. Split $A,B$ into high- and low-degree parts:
    \begin{equation}
    \begin{bmatrix}A(Z)\\B(Z)\end{bmatrix}
    =Z^L\begin{bmatrix}A_1(Z)\\B_1(Z)\end{bmatrix}
    +\begin{bmatrix}A_2(Z)\\B_2(Z)\end{bmatrix},
    \qquad
    \max\{\deg A_2,\deg B_2\}<L.
    \end{equation}
    Then $\deg A_1=n_1>\deg B_1$. Let
    $M:=\mathcal H_{n_1}(A_1,B_1)$, we have
    \begin{equation}
    M\begin{bmatrix}A(Z)\\B(Z)\end{bmatrix}
    =M\begin{bmatrix}Z^LA_1(Z)+A_2(Z)\\
                        Z^LB_1(Z)+B_2(Z)\end{bmatrix}
    =M\begin{bmatrix}A_1(Z)&A_2(Z)\\
                        B_1(Z)&B_2(Z)\end{bmatrix}
      \begin{bmatrix}Z^L\\1\end{bmatrix}.
    \end{equation}
    Denote
    \begin{equation}
    M\begin{bmatrix}A_1(Z)&A_2(Z)\\B_1(Z)&B_2(Z)\end{bmatrix}
    =:
    \begin{bmatrix}C_1(Z)&C_2(Z)\\D_1(Z)&D_2(Z)\end{bmatrix}.    
    \end{equation}
    Note that
    $\mathcal H_{n_1}(A_1,B_1)\begin{bmatrix}A_1\\B_1\end{bmatrix}
    =
    M\begin{bmatrix}A_1\\B_1\end{bmatrix}
    =\begin{bmatrix}C_1\\D_1\end{bmatrix}$, hence the vector $\begin{bmatrix}C_1\\D_1\end{bmatrix}$
    is the stopping pair of the first recursive half (\ref{eq:stopping}), leading to
    \begin{equation}
        \deg C_1\ge s_1>\deg D_1.
    \end{equation}
    As in \eqref{eq:half-gcd-degree-bound}, every entry of $M$
    has degree at most $n_1-s_1$. Thus,
    \begin{equation}
        \deg C_2,\deg D_2
        \le n_1-s_1+\max\{\deg A_2,\deg B_2\}
        <n_1-s_1+L
        \le s_1+L.
    \end{equation}
    We obtain
    \begin{equation}
    M\begin{bmatrix}A(Z)\\B(Z)\end{bmatrix}
    =
    \begin{bmatrix}
    C_1(Z)Z^L+C_2(Z)\\
    D_1(Z)Z^L+D_2(Z)
    \end{bmatrix},
    \end{equation}
    and
    \begin{equation}
    \label{eq:d-rfr-rcsbound}
    \deg(C_1Z^L+C_2)\ge s_1+L> \deg(D_1Z^L+D_2).
    \end{equation}

    

    The computation of $\mathcal H_N(A,B)$ can therefore be expressed as follows:
    \begin{enumerate}
        \item Set $L:=s$ and recursively compute
        $M_1=\mathcal H_{n_1}(A_1,B_1)$. Then assign
        \begin{equation}
        \begin{bmatrix}A(Z)\\B(Z)\end{bmatrix}
            \gets M_1\begin{bmatrix}A(Z)\\B(Z)\end{bmatrix}.
        \end{equation}
        The new $B(Z)$ is exactly the polynomial $D_1(Z)Z^L+D_2(Z)$ discussed
        above, hence
        \begin{equation}
         \deg B\le s+s_1-1\le\lfloor3N/4\rfloor.
        \end{equation}
        While similarly there is $\deg A\ge \lfloor3N/4\rfloor$.
        Hence, if $\deg B<s$, just return $M_1$.

        \item Perform one Euclidean step, namely
        \begin{equation}
        \begin{bmatrix}A(Z)\\B(Z)\end{bmatrix}
            \gets
            \begin{bmatrix}B(Z)\\A(Z)\bmod B(Z)\end{bmatrix}.
        \end{equation}
        Let $Q$ be the quotient in this step.
        At this point, $\deg A,\deg B\le\lfloor3N/4\rfloor$.
        
        As the new $A$ is the resulting $B$ in step (1), we have $\deg A\ge s$ currently. If $\deg B<s$, return $\varphi(Q)M_1$.

        \item Let $d_A:=\deg A$, set $L:=2s-d_A$, and take the
        corresponding high parts $A_1,B_1$ of the current $(A,B)$ pair.
        Set $n_2:=\deg A_1=d_A-L=2(d_A-s)$ and recursively compute $M_2:=\mathcal H_{n_2}(A_1,B_1)$.
        Here $0<n_2\le\lfloor N/2\rfloor$ and $L>0$.
        Assign
        \begin{equation}
        \begin{bmatrix}A(Z)\\B(Z)\end{bmatrix}
        \gets M_2\begin{bmatrix}A(Z)\\B(Z)\end{bmatrix}.
        \end{equation}
        The same degree bound from (\ref{eq:d-rfr-rcsbound}) gives
        $\deg A\ge n_2/2+L=s>\deg B.$

        \item Return $M_2\varphi(Q)M_1$, where $Q$ is the quotient from step (2).
    \end{enumerate}

    Each recursive level performs at most two recursive calls on
    inputs of degree at most $\lfloor N/2\rfloor$, with
    $O(\mathsf M(N))$ additional work for polynomial and matrix arithmetic.
    Therefore,
    \begin{equation}
    T(N)\le2T(\lfloor N/2\rfloor)+O(\mathsf M(N))
        =O(\mathsf M(N)\log N). 
    \end{equation}
    Take $\mathsf M(N)=O(N\log N)$ gives $O(N\log^2 N)$.
\end{proof}

\subsection{RFR with Unequal Degree Bounds}

The following form keeps the interface of the original balanced RFR while allowing different degree bounds.

\begin{lemma}[RFR with unequal degree bounds]
\label{lem:shifted-rfr}
    Let $d_A,d_B,n$ be integers satisfying
    $0\le d_A\le d_B$ and $n\ge d_A+d_B+1$.
    Let $S\in\mathcal U[Z]$ have degree less than $n$.
    Suppose that coprime $A,B\in\mathcal U[Z]$ satisfy
    \begin{equation}
        \deg A\le d_A,\qquad \deg B\le d_B,\qquad B(0)\ne0, \qquad\frac{A}{B}\equiv S\pmod{Z^n}.
    \end{equation}
    Then $(A,B)$ can be recovered up to a common nonzero scalar in $O(n\log^2 n)$ field operations.
\end{lemma}

\begin{proof}
    Set
    \begin{equation}
        \widetilde S(Z):=Z^{d_B-d_A}S(Z)\bmod Z^{2d_B+1}.
    \end{equation}
    Since $n+d_B-d_A\ge2d_B+1$, we have
    \begin{equation}
        \frac{Z^{d_B-d_A}A(Z)}{B(Z)}
        \equiv\widetilde S(Z)\pmod{Z^{2d_B+1}}.
    \end{equation}
    Both polynomials have degree at most $d_B$ and remain
    coprime as $B(0)\ne0$.
    Lemma~\ref{lem:original-rfr} therefore recovers
    $(Z^{d_B-d_A}A,B)$ up to a common nonzero scalar.
    Dividing the first polynomial by $Z^{d_B-d_A}$ yields
    the desired pair.
    The cost is $O(n\log^2 n)$ since $d_B<n$.
\end{proof}

\subsection{Degree-Aware Reconstruction}

We now prove the reconstruction part of Theorem~\ref{thm:d-rfr}, following the leading-term cancellation described in Section~\ref{sec:overview-d-rfr}.
The reduction below first removes any power of $Z$ from the numerator. This is needed because the theorem assumes only $G(0)\ne0$ but not $F(0)\ne0$.

\begin{restatedtheorem}
    Let $R(Z)$ be a polynomial with $\deg R<\ell$, and let
    $m\in[-\ell,\ell]$ be an integer. Suppose that there exist
    monic coprime polynomials $F,G$ satisfying
    \begin{equation}
        \begin{gathered}
            G(0)\ne0,\qquad
            \deg F+\deg G\le\ell,\qquad
            \deg F-\deg G=m,\\
            \frac{F(Z)}{G(Z)}\equiv R(Z)\pmod{Z^\ell}.
        \end{gathered}
    \end{equation}
    Then $F,G$ are uniquely determined by $(R,m)$ and can be
    recovered in $O(\ell\log^2\ell)$ field operations using
    an explicit algorithm.
\end{restatedtheorem}

\begin{proof}[Proof of the Reconstruction Part]
    We first remove a possible power of $Z$.
    If $R=0$, then $F\equiv0\pmod{Z^\ell}$ because $G(0)\ne0$. The degree bound and
    monicity thus force $F=Z^\ell$ and $G=1$.

    Suppose now that $R\ne0$, and let
    $v:=\min\{j:[Z^j]R(Z)\ne0\}$, where $[Z^j]R(Z)$ denotes the coefficient of $Z^j$ in $R(Z)$.
    Since $G(0)\ne0$, the congruence implies $Z^v\mid F$.
    Write $F(Z)=Z^vF_0(Z)$ and $R(Z)=Z^vR_0(Z)$, and set
    $n:=\ell-v$ and $m_0:=m-v$.
    Then $F_0(0)\ne0$, $G(0)\ne0$, and
    \begin{equation}
        \frac{F_0(Z)}{G(Z)}\equiv R_0(Z)\pmod{Z^n},
        \qquad
        \deg F_0+\deg G\le n,
        \qquad
        \deg F_0-\deg G=m_0.
    \end{equation}

    If $m_0<0$, then $R_0$ is invertible modulo $Z^n$. We replace
    $R_0$ by $R_0^{-1}\bmod Z^n$, replace $m_0$ by $-m_0$, and swap
    the roles of $F_0$ and $G$. Thus, it remains to describe the reconstruction for $m_0\ge0$.
    After that reconstruction, we undo this swap if it was made.

    For the case $m_0=0$, if $R_0=1$, the congruence, coprimality, and monicity imply $F_0=G=1$. We then restore $Z^v$ and return
    $(F,G)=(Z^v,1)$.
    Otherwise, define $H(Z):=F_0(Z)-G(Z)$ and let
    $D:=\lfloor n/2\rfloor$. Since $F_0$ and $G$ are monic and
    have the same degree, their leading terms cancel. Under the
    existence assumption, $R_0\ne1$ also implies $n\ge2$, so $D\ge1$.
    Consequently,
    \begin{equation}
        \deg H\le D-1,\qquad \deg G\le D,\qquad
        \frac{H(Z)}{G(Z)}\equiv R_0(Z)-1\pmod{Z^n}.
    \end{equation}
    Also, $\gcd(H,G)=\gcd(F_0,G)=1$. Applying
    Lemma~\ref{lem:shifted-rfr} with $S=R_0-1$, $d_A=D-1$,
    and $d_B=D$ recovers $(H,G)$ up to a common scalar;
    the condition $d_A+d_B+1=2D\le n$ holds automatically.
    We normalize the second component so that $G$ is monic
    and recover $(F_0,G)=(H+G,G)$.

    It remains to handle $m_0\ge1$. Let
    $d_G:=\lfloor(n-m_0)/2\rfloor$ and $d_H:=d_G+m_0-1$,
    and define $H(Z):=F_0(Z)-Z^{m_0}G(Z)$.
    Since $F_0$ and $Z^{m_0}G$ are monic of the same degree, their
    leading terms cancel. Thus, $\deg G\le d_G$ and $\deg H\le d_H$.
    Furthermore, $H(0)=F_0(0)\ne0$,
    $\gcd(G,H)=\gcd(G,F_0)=1$, and
    \begin{equation}
        \frac{G(Z)}{H(Z)}
        \equiv\frac{1}{R_0(Z)-Z^{m_0}}\pmod{Z^n}.
    \end{equation}
    The inverse exists because
    $(R_0-Z^{m_0})(0)=R_0(0)\ne0$. Finally,
    \begin{equation}
        d_G+d_H+1
        =2\left\lfloor\frac{n-m_0}{2}\right\rfloor+m_0
        \le n.
    \end{equation}
    Lemma~\ref{lem:shifted-rfr}, applied with
    $S(Z):=(R_0(Z)-Z^{m_0})^{-1}\bmod Z^n$,
    $d_A=d_G$, and $d_B=d_H$,
    therefore recovers $(G,H)$ up to a common nonzero scalar.
    Here the normalization must be taken from the \emph{first}
    component: we scale both components so that $G$ is monic.
    We then recover $(F_0,G)=(H+Z^{m_0}G,G)$.

    Finally, if the reciprocal reduction was used, we swap the two
    reconstructed polynomials back. We then restore the removed power
    of $Z$ and output $(F,G)=(Z^vF_0,G)$.
\end{proof}

\subsection{Discussion on Practical Implementation}

The proof above assumes that a valid pair exists. On arbitrary inputs, D-RFR reports failure if any reconstruction step fails.
It accepts a candidate only if $F,G$ are monic and coprime, $G(0)\ne0$, and
\begin{equation}
    \deg F-\deg G=m,\qquad
    \deg F+\deg G\le\ell,\qquad
    F\equiv GR\pmod{Z^\ell}.
\end{equation}
The cell decoder additionally checks that $F,G$ split into distinct linear factors with roots in $\mathcal U^*$.
As the reductions and candidate checks take
$O(\ell\log^2\ell)$ field operations, the total complexity remains $O(\ell\log^2\ell)$.

\section{Proofs for the Theoretical Analysis}
\label{app:analysis}

This appendix gives the full proofs of the results in
Section~\ref{sec:analysis}.
As in Section~\ref{sec:analysis}, we assume $M=\Theta(d)$.
Throughout this appendix, $k$ and $\ell$ are fixed with respect to
$d$ and $V$, and we retain their explicit dependence in the complexity bounds.
The fingerprint modulus $q\ge3$ is odd. For convenience, we use $q=1$ to denote the version without a fingerprint.

\subsection{Proofs of Propositions~\ref{prop:complexity} and~\ref{prop:bit-space}}
\label{app:analysis-complexity}

\begin{proof}[Proof of Proposition~\ref{prop:complexity}]
We analyze Algorithms~\ref{alg:pseudo:update} and~\ref{alg:pseudo:subtract}
for updates and subtraction, and Algorithms~\ref{alg:pseudo:decode} and~\ref{alg:cell-decode} for cell decoding, verification, and peeling.
For each inserted item $x$, the algorithm updates at most $k$ cells.
Within one cell, it updates a signed counter $c_C$, a fingerprint
$\mathsf{fp}_C$, and the polynomial
$P_C(Z)\leftarrow P_C(Z)(Z-x)\bmod Z^\ell$.
The latter operation takes $O(\ell)$ field operations, while deletion uses division by $Z-x$ modulo $Z^\ell$ and has the same cost.
Hence one update costs $O(k\ell)$ time.

For decoding, we first insert all $M$ cell indices into a queue.
Whenever a cell is decoded and passes verification, it recovers at most $\ell$ residual items. Each recovered item is removed from the cells in its support, which can be recomputed from the item itself.
During correct peeling, every differing item is recovered at most once, so the total number of such cell updates is at most $kd$.

As a cell is tested only initially or after one of its incident items has been removed, the total number of decoding attempts is at most $M+kd=O(kd)$.
Algorithm~\ref{alg:pseudo:decode} enforces the same budget on erroneous trajectories: it returns $\bot$ before processing a candidate that would make the cumulative number of peeled signed items exceed $d$, counting repetitions.
This guard never aborts a correct peeling trajectory and ensures at most $kd$ point updates, even if a false positive occurs.

One decoding attempt invokes D-RFR in $O(\ell\log^2\ell)$ field operations
and finds the roots of polynomials of total degree at most $\ell$ in
expected $O(\ell\log^2\ell\log V)$ field operations.
The algebraic checks and reconstruction of the products of roots take $O(\ell\log^2\ell)$ field operations.
Rehashing verification evaluates at most $k$ hashes for each of at most $\ell$ recovered items, and fingerprint verification takes $O(\ell)$ time.
Sketch subtraction takes $O(M\ell\log\ell)$ field operations.
The total expected decoding time is therefore
\begin{equation}
\begin{aligned}
&O\!\left(M\ell\log\ell+(M+kd)
   (\ell\log^2\ell\log V+k\ell)+kd\ell\right)\\
&\qquad=O\!\left(k\ell d(\log^2\ell\log V+k)\right)
 =O_{k,\ell}(d\log V).
\end{aligned}
\end{equation}
The expectation arises solely from the randomized root-finding routine, while all remaining operations in the decoding procedure are deterministic.
\end{proof}

\begin{proof}[Proof of Proposition~\ref{prop:bit-space}]
Each cell has $V^\ell$ possible polynomial states, $2\ell+1$ counter
values, and $q$ fingerprint values. Thus, the entire sketch has at most
$(V^\ell(2\ell+1)q)^M$ possible states. Jointly encoding these states uses at most
\begin{equation}
    \left\lceil M\left(\ell\log_2 V+\log_2(2\ell+1)+\log_2 q\right)\right\rceil
\end{equation}
bits, proving the claim.
\end{proof}

\subsection{Proof of Lemma~\ref{prop:false-positive}}
\label{app:analysis-verification}

We first bound the probability that an incorrect candidate passes rehashing.
An independent fingerprint then contributes an additional factor $1/q$.

Let $N:=|\mathcal U^*|=V-1$, and let
$\widetilde\Gamma(x):=\Gamma(\phi(x))$ denote the write locations of the original element $x$.
For every fixed $\phi$, the locations $\widetilde\Gamma(x)=\Gamma(\phi(x))$ are independently sampled from the same distribution. Thus, $\widetilde\Gamma$ is independent of $\phi$.
We equivalently expose the randomness as follows: each original element $x$ first receives its write locations $\widetilde\Gamma$, then receives its label $\phi(x)$.
For every cell $j$, $\Pr[j\in\widetilde\Gamma(x)]\le C_0/M$,
where the constant $C_0$ depends only on the placement distribution ({\eg}, spatial coupling or circular).
We assume $d\le V/2$.
For a signed map $f:\mathcal U^*\to\{-1,0,1\}$, define
$\Psi_f(Z):=\prod_{x\in\mathcal U^*}(Z-\phi(x))^{f(x)}\bmod Z^\ell$.

\begin{lemma}
\label{lem:app-analysis-poly-spread}
Let $f\ne g$ be fixed signed maps with
$|\operatorname{supp}(g)|\le\ell$ and
$|\operatorname{supp}(f)\cup\operatorname{supp}(g)|\le d+\ell$.
Then
\begin{equation}
    \Pr_\phi\!\left[
    \begin{gathered}
        \Psi_f(Z)=\Psi_g(Z),\\
        \sum_{x\in\mathcal U^*}f(x)
        \equiv\sum_{x\in\mathcal U^*}g(x)
        \pmod{2\ell+1}
    \end{gathered}
    \right]
    =O_\ell(N^{-\ell}).
\end{equation}
\end{lemma}

The proof of this lemma is deferred to Appendix~\ref{app:analysis-missing-proof}.

\begin{lemma}
\label{lem:app-analysis-rehash-error}
The probability that a non-pure cell passes rehashing verification at any point before the first erroneous peeling step is $O_{k,\ell}(M^{1-\ell})$.
\end{lemma}

\begin{proof}
For each cell $j$, let $L_j:=|\widetilde\Gamma^{-1}(j)\cap((A\setminus B)\cup(B\setminus A))|$, where
$\widetilde\Gamma^{-1}(j):=\{x:j\in\widetilde\Gamma(x)\}$.
Before the first erroneous step, the residual content of the cell $j$ is obtained by removing a subset of these $L_j$ elements and has at most $2^{L_j}$ possible values.


If a candidate $(S,T)$ passes rehashing in cell $j$, then
$S\cup T\subseteq \widetilde\Gamma^{-1}(j)$ and $|S|+|T|\le\ell$.
Here, $S,T$ denote the candidate sets after applying $\phi^{-1}$
to their elements.
As $|\Gamma^{-1}(j)|=|\widetilde\Gamma^{-1}(j)|$, the number of such $(S,T)$ is at most
\begin{equation}
\sum_{t=0}^{\ell}2^t\binom{|\widetilde\Gamma^{-1}(j)|}{t}
\le \sum_{t=0}^{\ell}\dfrac{2^t}{t!}|\Gamma^{-1}(j)|^t\le 2\sum_{t=0}^{\ell}|\Gamma^{-1}(j)|^t\le 2(1+|\Gamma^{-1}(j)|)^\ell,
\end{equation}
where the inequalities follow from
$\binom{u}{v}\le u^v/v!$, $2^t/t!\le2$, and the binomial
expansion with $\binom{\ell}{t}\ge1$, respectively.

Fix $\widetilde\Gamma$.
For each cell $j$, let $f_j$ range over the $2^{L_j}$ possible signed maps.
For each candidate $(S,T)$, let $g$ be its signed map.
Define $E_j(S,T,f_j)$ as the bad event that:
\begin{itemize}
    \item $(S,T)$ is incorrect: $g\ne f_j$;
    \item $(S,T)$ matches the cell's counter state: $\sum_{x\in\mathcal U^*}g(x)
\equiv\sum_{x\in\mathcal U^*}f_j(x)\pmod{2\ell+1}$;
    \item $(S,T)$ matches the cell's polynomial state: $\Psi_g=\Psi_{f_j}$;
    \item $(S,T)$ passes rehashing in cell $j$.
\end{itemize}
Let $E:=\bigcup_{j,S,T,f_j}E_j(S,T,f_j)$. Any first erroneous peeling step implies that $E$ occurs.

Conditioning on $\widetilde\Gamma$, each fixed incorrect candidate matches both the counter and the polynomial state of $f_j$ with probability $O_\ell(N^{-\ell})$ by Lemma~\ref{lem:app-analysis-poly-spread}.
For cell $j$, there are $2^{L_j}$ possible residual maps and at most
$2(1+|\widetilde\Gamma^{-1}(j)|)^\ell$ candidates passing rehashing.
A union bound therefore gives
\begin{equation}
\begin{aligned}
\Pr_\phi[E\mid\widetilde\Gamma]
&\le
\sum_{j=1}^M
\sum_{f_j}
\sum_{\substack{(S,T)\text{ passing rehashing}\\g\ne f_j}}
\Pr_\phi[E_j(S,T,f_j)\mid\widetilde\Gamma]\\
&\le
O_\ell(N^{-\ell}) \sum_{j=1}^M
2^{L_j}(1+|\widetilde\Gamma^{-1}(j)|)^\ell.
\end{aligned}
\end{equation}

Let $K_j:=|\widetilde\Gamma^{-1}(j)|-L_j=|\widetilde\Gamma^{-1}(j)\setminus ((A\setminus B)\cup(B\setminus A))|$.
$L_j$ and $K_j$ are independent sums of independent Bernoulli variables.
Since $\Pr[j\in\widetilde\Gamma(x)]\le 2k/M$ \footnote{For both spatial coupling and its circular variant, each cell
occupies an interval of length $(z+1)/M$. Conditioning on either
$g_0(x)$ or $g_i(x)$ gives
$\Pr[h_i^{(a)}(x)=j]\le
(z+1)/(M\max\{z+a,1\})\le2/M$,
where $a=0$ covers conventional spatial coupling.
A union bound over the $k$ positions gives
$\Pr[j\in\widetilde\Gamma(x)]\le2k/M$.}, their means satisfy $\mathbb E[L_j]\le\frac{2kd}{M}=O_{k,\ell}(1);\; \mathbb E[K_j]\le\frac{2kN}{M}$.

The standard moment-generating-function bound for sums of independent Bernoulli variables gives
$\mathbb E[3^{L_j}]\le e^{2\mathbb E[L_j]}=O_{k,\ell}(1)$.
Using $(1+u)^\ell=O_\ell((3/2)^u)$ for $u\ge0$, and applying
the same bound to $K_j$,
\footnote{For a sum $X$ of independent Bernoulli variables with
mean $\mu$, write $X=\sum_i B_i$, where the $B_i$ are independent
Bernoulli variables with success probabilities $p_i$,
and let $\mu:=\mathbb E[X]=\sum_i p_i$.
Independence and $1+y\le e^y$ give
$\mathbb E[e^{tX}]\le \prod_i(1-p_i+p_ie^t)\le e^{\mu(e^t-1)}$.
Setting $t=\ln3$ gives $\mathbb E[3^X]\le e^{2\mu}$.
For $K_j$, set $t=(1+\mu)^{-1}\in[0,1]$, then $e^t-1\le2t$.
We have $\mathbb E[e^{t(1+X)}]\le e^{t+2\mu t}\le e^{(2\mu+2)t}=e^2$.
Finally, the exponential series gives
$[t(1+X)]^\ell/\ell!\le e^{t(1+X)}$.
Taking expectations yields
$\frac{t^\ell}{\ell!}\mathbb E[(1+X)^\ell]\le e^2$, hence $\mathbb E[(1+X)^\ell]\le e^2\ell!t^{-\ell}=e^2\ell!(1+\mu)^\ell=O_{\ell}((1+\mu)^\ell)$.}
we obtain
\begin{equation}
\begin{aligned}
    \mathbb E[2^{L_j}(1+L_j)^\ell]
    &\le O_\ell(\mathbb E[3^{L_j}])
    =O_{k,\ell}(1),\\
    \mathbb E[(1+K_j)^\ell]
    &\le O_\ell((1+\mathbb E[K_j])^\ell)
    =O_{k,\ell}((1+N/M)^\ell).
\end{aligned}
\end{equation}

Since $L_j,K_j$ are independent and
$1+L_j+K_j\le(1+L_j)(1+K_j)$, averaging the union bound gives
\begin{equation}
\begin{aligned}
    \Pr[E]&\le O_{\ell}(N^{-\ell})\mathbb E\left[\sum_{j=1}^M
2^{L_j}(1+L_j+K_j)^\ell\right]\\
    &\le O_{\ell}(N^{-\ell})\sum_{j=1}^M\mathbb E\left[
2^{L_j}(1+L_j)^\ell\right]\mathbb E\left[
(1+K_j)^\ell\right]\\
    &\le O_{k,\ell}\left(N^{-\ell}M(1+N/M)^\ell\right)
    =O_{k,\ell}(d^{1-\ell}),
\end{aligned}
\end{equation}
where we use $M=\Theta(d)$ and $d\le V/2$.
\end{proof}

\begin{lemma}
\label{lem:app-analysis-fingerprint-error}
For every fixed pair of signed maps $f\ne g$, their fingerprints
agree with probability $q^{-1}$, independently of the polynomial
collision event.
\end{lemma}

\begin{proof}
The case of $q=1$ is trivial. Otherwise, choose $x_0$ with $f(x_0)\ne g(x_0)$. Fingerprint equality is equivalent to
\begin{equation}
\sum_x(f(x)-g(x))h_{\mathrm{fp}}(x)\equiv0\pmod q.
\end{equation}
The coefficient of $h_{\mathrm{fp}}(x_0)$ belongs to
$\{-2,-1,1,2\}$ and is a unit in $\mathbb Z_q$ because $q$ is odd.
After fixing all other fingerprint values, there is exactly one value of $h_{\mathrm{fp}}(x_0)$ which satisfies the equation.
\end{proof}

\begin{proof}[Proof of Lemma~\ref{prop:false-positive}]
In the union bound of Lemma~\ref{lem:app-analysis-rehash-error}, every fixed incorrect
pair incurs an additional independent factor $q^{-1}$ by
Lemma~\ref{lem:app-analysis-fingerprint-error}. The same argument therefore
gives $O_{k,\ell}(M^{1-\ell}/q)=O_{k,\ell}(d^{1-\ell}/q)$.
\end{proof}

\subsection{Proofs of Lemma~\ref{thm:orientability-benchmark} and Theorem~\ref{thm:main}}
\label{app:analysis-space-overhead}

We first use spatial coupling to approach the orientability threshold.
We then choose $k,\ell$ to obtain the time--space trade-off in
Theorem~\ref{thm:main}. Throughout this subsection, we use the conventional
coupling rule $h_i'$ from Section~\ref{sec:spatial-coupling}.

\begin{lemma}[Coupled Peeling]
\label{lem:app-analysis-coupled-peeling}
Fix $k,\ell$ with $k\ge2$ and $k+\ell\ge4$, and let
$0<c<c^{\mathsf{orient}}_{k,\ell}$. For every sufficiently large fixed $z>0$, let
\begin{equation}
M=\left\lceil \left(1+\frac1z\right)\frac{d}{c}\right\rceil.
\end{equation}
Then, for every $t\le d$, the hypergraph induced by the spatial coupling rule of $t$ independently sampled edges is $\ell$-peelable w.h.p..
\end{lemma}

\begin{proof}
Walzer~\cite{PeelingOrientabilityThreshold-Walzer-2025} proves that such a
coupled hypergraph on $M$ vertices with
$\lfloor cMz/(z+1)\rfloor$ independently sampled hyperedges is
$\ell$-peelable w.h.p.
Let $m':=\left\lfloor \frac{cMz}{z+1}\right\rfloor.$ Our choice of $M$ ensures that $m'\ge d\ge t$. Generate an auxiliary coupled hypergraph with $m'$ independent hyperedges and retain its first $t$ edges.
These $t$ edges have exactly the distribution induced by our placement rule.
The claim follows because $\ell$-peelability is preserved under deleting hyperedges.
\end{proof}

\begin{proof}[Proof of Lemma~\ref{thm:orientability-benchmark}]
For each integer $j\ge1$, choose $\eta_j>0$ such that
$
(1+\eta_j)^2
\frac{\ell}{c^{\mathsf{orient}}_{k,\ell}}
\le
\frac{\ell}{c^{\mathsf{orient}}_{k,\ell}}
+
\frac1j.
$
Choose $0<c_j<c^{\mathsf{orient}}_{k,\ell}$ such that
$
\frac{\ell}{c_j}
<
(1+\eta_j)
\frac{\ell}{c^{\mathsf{orient}}_{k,\ell}}.
$
By the result for spatial coupling
\cite{PeelingOrientabilityThreshold-Walzer-2025}, there exists a fixed
coupling parameter $z_j$ such that $ 1+\frac1{z_j}<1+\eta_j$
and the coupled hypergraph with density $c_j$ is
$\ell$-peelable w.h.p.

For fixed $j$, let
$
M_{j,d}:=
\left\lceil
\left(1+\frac1{z_j}\right)\frac{d}{c_j}
\right\rceil.
$
By Lemma~\ref{lem:app-analysis-coupled-peeling}, the probability that the induced hypergraph is not $\ell$-peelable is $o(1)$ as $d\to\infty$.
Moreover, with $q=1$,
Lemma~\ref{prop:false-positive} gives a false-positive
probability of $O_{k,\ell}\left(M_{j,d}^{1-\ell}/q\right)=o(1)$,
since $\ell\ge2$ and $M_{j,d}=\Theta(d)$.

Hence, for every $j$, there exists $D_j$ such that, for all
$d\ge D_j$, the decoding failure probability under the parameters
$(M_{j,d},z_j)$ is at most $1/j$. Choose the sequence
$
D_1<D_2<\cdots
$
strictly increasing and let
$
j(d):=\max\{j:D_j\le d\},
$
then instantiate {\fullalgo} with
$
c_d:=c_{j(d)},
z_d:=z_{j(d)},
M_d:=M_{j(d),d},q=1.
$
Since $j(d)\to\infty$, the decoding failure probability is at most
$1/j(d)=o(1)$.
Thus,

\begin{equation}
    \frac{\mathfrak C(d)}{d}
=
\frac{M_d\ell}{d}<
(1+\eta_{j(d)})^2
\frac{\ell}{c^{\mathsf{orient}}_{k,\ell}}
+
\frac{\ell}{d}\ \ \implies\ \ 
\limsup_{d\to\infty}
\frac{\mathfrak C(d)}{d}
\le
\frac{\ell}{c^{\mathsf{orient}}_{k,\ell}}.
\end{equation}
\end{proof}

We first quantify how fast $\ell/c^{\mathsf{orient}}_{k,\ell}$ approaches one.
\begin{lemma}
\label{lem:speed-approach}
For sufficiently large integers $t$, we have
\begin{equation}
\label{eq:app-analysis-orientability-gap}
1\le\frac{t}{c^{\mathsf{orient}}_{t,t}}
\le1+\exp\bigl(-t^2+O(t\log t)\bigr).
\end{equation} 
\end{lemma}
\begin{proof}
Set $k=\ell=t\ge3$ and write
$Q(x,r):=\Pr[\operatorname{Pois}(x)\ge r]$.
The orientability threshold formula~\cite{OrientabilityThresholds-FountoulakisKhoslaPanagiotou-2016}
gives
\begin{equation}
\label{eq:app-analysis-orientability-formula}
t^2=\frac{\xi Q(\xi,t)}{Q(\xi,t+1)},
\qquad
c^{\mathsf{orient}}_{t,t}
=\frac{\xi}{tQ(\xi,t)^{t-1}},
\end{equation}
where $\xi>0$ is the unique solution of the first equation.
We first locate $\xi$ and then bound the space ratio.

Let $b_t(x):=e^{-x}x^t/t!$.
For $x\in[t^2-1,t^2]$, Stirling's formula gives
$\log b_t(x)=-t^2+t\log t+O(t)$.
Moreover, as
\begin{equation}
t!/(t-j)!=t(t-1)\cdots(t-j+1)\le t^j\implies \dfrac{t!}{(t-j)!x^j}\le \left(t/x\right)^j,
\end{equation}
we have
\begin{equation}
\label{eq:app-analysis-poisson-tail}
1-Q(x,t+1)
=b_t(x)\sum_{j=0}^{t}\frac{t!}{(t-j)!x^j}
\le b_t(x)\sum_{j=0}^{t}\left(t/x\right)^j\le\frac{b_t(x)}{1-t/x}
\le2b_t(x).
\end{equation}
Thus, choosing a sufficiently large constant $C$ and setting
$\eta_t:=\exp(-t^2+Ct\log t)$, we have
$b_t(x)\le 1-Q(x,t+1)\le\eta_t$ for all $x\in [t^2-1,t^2]$,
with $\lim_{t\to \infty}t^2\eta_t=0$.

To locate $\xi$, define $H_t(x):=xQ(x,t)/Q(x,t+1) =x+xb_t(x)/Q(x,t+1)$.
Clearly $H_t(t^2)>t^2$, whereas
\begin{equation}
H_t(t^2-1)
\le t^2-1+\frac{(t^2-1)\eta_t}{1-\eta_t}
<t^2
\end{equation}
for sufficiently large $t$.
Continuity and uniqueness therefore imply
$t^2-1<\xi<t^2$.

Applying the tail bound at $\xi$ and using
\eqref{eq:app-analysis-orientability-formula}, we obtain
\begin{equation}
\frac{t}{c^{\mathsf{orient}}_{t,t}}
=\frac{t^2}{\xi}Q(\xi,t)^{t-1}
=\frac{Q(\xi,t)^t}{Q(\xi,t+1)}
\le\frac{Q(\xi,t)^t}{1-\eta_t}
\le\frac{1}{1-\eta_t}
\end{equation}
The last step follows from
$Q(\xi,t)=\Pr[\operatorname{Pois}(\xi)\ge t]\le1$.
For sufficiently large $t$, there must be $\eta_t\le 1/2$ as
$\eta_t\to 0$, then
\begin{equation}
\frac{t}{c^{\mathsf{orient}}_{t,t}}\le \frac{1}{1-\eta_t}\le1+2\eta_t.
\end{equation}
Finally, a $t$-orientable hypergraph on $n$ vertices has at most
$tn$ edges, so $c^{\mathsf{orient}}_{t,t}\le t$.
This proves (\eqref{eq:app-analysis-orientability-gap}).
\end{proof}

\begin{proof}[Proof of Theorem~\ref{thm:main}]
Following Lemma \ref{lem:speed-approach}, for every sufficiently small $\varepsilon>0$, we may choose
$t=\Theta(\sqrt{\log(1/\varepsilon)})$ and fix $k=\ell=t$ such that
$\ell/c^{\mathsf{orient}}_{k,\ell}<1+\varepsilon/4$.
Apply Lemma~\ref{thm:orientability-benchmark} with these $k,\ell$,
which yields parameter schedules $M_d,z_d$ with $q=1$ such that
{\fullalgo} succeeds w.h.p. and
\begin{equation}
\limsup_{d\to\infty}
\frac{\mathfrak C(d)}{d}
\le
\frac{\ell}{c^{\mathsf{orient}}_{k,\ell}}
<1+\frac{\varepsilon}{4}.
\end{equation}
Hence, for all sufficiently large $d$,
$\mathfrak C(d)\le(1+\varepsilon/2)d$.
Moreover, the failure probability is $o(1)$, and is therefore at most
$\delta$ for all sufficiently large $d$.

Proposition~\ref{prop:complexity} gives an update cost of
$O(t^2)=O(\log(1/\varepsilon))$ and expected decoding time
\begin{equation}
\begin{aligned}
O\!\left(t^2d(\log^2t\log V+t)\right)
&=O\!\left(d\log V\,(t^2\log^2t+t^3)\right)\\
&=O\!\left(d\log V\log^{3/2}(1/\varepsilon)\right).
\end{aligned}
\end{equation}
These bounds have absolute implied constants, since
$M_d\ell\le(1+\varepsilon/2)d$ bounds $M_d/d$ by an absolute constant.
The dependence on $\delta$ enters only through how large $d$ must be.

Since $\log_2(2\ell+1)\le2\ell$, the counter fields use at most
$2M_d\ell=O(d)$ bits, with an absolute implied constant.
Finally, Proposition~\ref{prop:bit-space} gives
\begin{equation}
\begin{aligned}
M_d\left(\ell\log_2 V+\log_2(2\ell+1)\right)
&\le\left(1+\frac{\varepsilon}{2}\right)d\log_2 V+O(d)\\
&\le(1+\varepsilon)d\log_2 V.
\end{aligned}
\end{equation}
The last inequality holds for sufficiently large $d$ since $V\ge2d$.
Taking ceilings gives the claimed bit bound.
\end{proof}

\subsection{Proof of Lemma~\ref{lem:app-analysis-poly-spread}}
\label{app:analysis-missing-proof}
\paragraph{Preliminary knowledge.} We first recall some simple facts used below.

\textbf{1)} If $F(0)\ne0$, then $F$ is invertible modulo $Z^\ell$, so $F'(Z)/F(Z)\bmod Z^\ell$ is well defined. Moreover, if
$F(Z)\equiv1\pmod{Z^\ell}$, then
$F'(Z)/F(Z)\equiv0\pmod{Z^{\ell-1}}$.

\textbf{2)} For every
$a\in\mathcal U^*$, we have
$\frac{1}{Z-a}
\equiv
-\sum_{r\ge1}a^{-r}Z^{r-1}\pmod{Z^\ell}$.

\textbf{3)} If $m$ polynomials in $m$ variables have degrees $D_1,\ldots,D_m$, then they have at most $\prod_{i=1}^m D_i$ \textit{isolated} common zeros over the algebraic closure. A common zero at which the Jacobian has full rank is \textit{isolated}. The proposition is a standard consequence of the B\'ezout theorem in algebraic geometry.

\begin{lemma}
\label{lem:bounded-system-solutions}
Fix integers $e_1,\ldots,e_\ell\in\{-2,-1,1,2\}$ whose images in $\mathcal U$ are nonzero \footnote{In characteristic $2$, this requires $e_i\in\{-1,1\}$.}, and a polynomial $T\in\mathcal U[Z]$ with $\deg T<\ell$ and $T(0)\ne0$.
For $a=(a_1,\ldots,a_\ell)$, Define
\begin{equation}
    P(Z):=\prod_{i=1}^{\ell}(Z-a_i)^{e_i}.
\end{equation}
Then the congruence $T(Z)P(Z)\equiv1\pmod{Z^\ell}$ has at most $(2\ell)^\ell=O_\ell(1)$ solutions
$a\in(\overline{\mathcal U}\setminus\{0\})^\ell$
with pairwise distinct coordinates, where
$\overline{\mathcal U}$ is the algebraic closure of $\mathcal U$.
\end{lemma}

\begin{proof}
We first clear the denominators. Let
$D(Z):=\prod_{e_i<0}(Z-a_i)^{-e_i}$ and define
\begin{equation}
    H(Z):=D(Z)(T(Z)P(Z)-1)
    =T(Z)\prod_{e_i>0}(Z-a_i)^{e_i}-D(Z).
\end{equation}
Since $D(0)\ne0$, the original congruence is equivalent to
$H(Z)\equiv0\pmod{Z^\ell}$.
For $r=0,\ldots,\ell-1$, let $H_r(a):=[Z^r]H(Z)$,
where $[Z^r]$ denotes the coefficient of $Z^r$.
We therefore need to count the solutions of
$H_0(a)=\cdots=H_{\ell-1}(a)=0$.

In fact, each $H_r$ has total degree at most $2\ell$.
To apply B\'ezout's theorem, we show that every solution with
nonzero, pairwise distinct coordinates has a nonsingular Jacobian.
Since only $(Z-a_i)^{e_i}$ in $P$ depends on $a_i$,
we have $\partial P/\partial a_i=-e_iP/(Z-a_i)$.
As $T$ is fixed, differentiating $H=D(TP-1)$ gives
\begin{equation}
\begin{aligned}
    \frac{\partial H}{\partial a_i}
    &=\frac{\partial D}{\partial a_i}(TP-1)
      +DT\frac{\partial P}{\partial a_i}\\
    &=\frac{\partial D}{\partial a_i}(TP-1)
      -\frac{e_iDTP}{Z-a_i}\\
\end{aligned}
\end{equation}
At a solution, $TP\equiv1$, so the first term vanishes
and $TP$ in the second term becomes $1$. Hence,
\begin{equation}
\frac{\partial H}{\partial a_i}\equiv -\frac{e_iD(Z)}{Z-a_i}\pmod{Z^\ell}.
\end{equation}
Here, division by $Z-a_i$ is valid since $a_i\ne0$.

At a fixed solution, write
$D(Z)\equiv\sum_{r=0}^{\ell-1}d_rZ^r\pmod{Z^\ell}$.
Expanding
$-e_i/(Z-a_i)\equiv
\sum_{r=0}^{\ell-1}e_i a_i^{-r-1}Z^r\pmod{Z^\ell}$,
the derivative formula above gives the Jacobian
\begin{equation}
\begin{aligned}
    \left(\frac{\partial H_r}{\partial a_i}\right)_{
        \substack{0\le r<\ell\\1\le i\le\ell}}
    &=
    \begin{pmatrix}
        d_0&0&\cdots&0\\
        d_1&d_0&\cdots&0\\
        \vdots&\vdots&\ddots&\vdots\\
        d_{\ell-1}&d_{\ell-2}&\cdots&d_0
    \end{pmatrix}
    \begin{pmatrix}
        e_1a_1^{-1}&\cdots&e_\ell a_\ell^{-1}\\
        e_1a_1^{-2}&\cdots&e_\ell a_\ell^{-2}\\
        \vdots&&\vdots\\
        e_1a_1^{-\ell}&\cdots&e_\ell a_\ell^{-\ell}
    \end{pmatrix}.
\end{aligned}
\end{equation}
The first matrix is triangular with nonzero diagonal
$d_0=D(0)$, so it is invertible.
Multiplying column $i$ of the second matrix by $a_i/e_i$
gives a Vandermonde matrix on
$a_1^{-1},\ldots,a_\ell^{-1}$, which are pairwise distinct.
Thus the Jacobian is nonsingular, and every such solution
is isolated.

B\'ezout's theorem applied to $H_0,\ldots,H_{\ell-1}$
therefore bounds the number of these solutions by
$(2\ell)^\ell=O_\ell(1)$.
\end{proof}

\begin{proof}[Proof of Lemma~\ref{lem:app-analysis-poly-spread}]
If $|\operatorname{supp}(f)|\le\ell$, the cell is pure,
so its counter and polynomial state uniquely determine $f$
by Theorem~\ref{thm:d-rfr}. Thus $f=g$, a contradiction.
Thus, it remains to consider $|\operatorname{supp}(f)|>\ell$.
Suppose that counters of $f$ and $g$ agree, or the event in the lemma cannot occur.

Define $e(x):=f(x)-g(x)$ and $J:=\{x:e(x)\ne0\}$, and let $s:=|J|$.
For every $x\in J$, the integer $e(x)$ belongs to $\{-2,-1,1,2\}$.
Let $J_0\subseteq J$ consist of those $x$ for which $e(x)$ is nonzero in $\mathcal U$,
\footnote{This distinction is needed in characteristic $2$,
where $e(x)=\pm2$ becomes zero as a field coefficient. In other
characteristics, $J_0=J$.} and let $s_0:=|J_0|$.
Since $|\operatorname{supp}(f)|>\ell\ge|\operatorname{supp}(g)|$, some $x$ belongs to $\operatorname{supp}(f)$ but not
$\operatorname{supp}(g)$.
For this $x$, $e(x)=\pm1$, so $J_0\ne\varnothing$.

Suppose that their polynomial states agree ({\ie}, $\Psi_f(Z)=\Psi_g(Z)$). Then
\begin{equation}
    R(Z):=
    \prod_{x\in J}(Z-\phi(x))^{e(x)}
    \equiv1\pmod{Z^\ell}.
\end{equation}
Since every $\phi(x)$ is nonzero, all factors are invertible
modulo $Z^\ell$. We have
\begin{equation}
    \frac{R'(Z)}{R(Z)}
    =
    \sum_{x\in J_0}\frac{e(x)}{Z-\phi(x)}
    \equiv0\pmod{Z^{\ell-1}}.
\end{equation}
Expanding each fraction and comparing coefficients of $Z^{r-1}$, we obtain
\begin{equation}
    \sum_{x\in J_0}e(x)\phi(x)^{-r}=0,
    \qquad r=1,\ldots,\ell-1.
\end{equation}

We first show that $s_0\ge\ell$. Otherwise, write
$J_0=\{x_1,\ldots,x_{s_0}\}$ and let
$y_i:=\phi(x_i)^{-1}$.
The first $s_0\le \ell-1$ equations above give
\begin{equation}
    \begin{pmatrix}
        y_1&y_2&\cdots&y_{s_0}\\
        y_1^2&y_2^2&\cdots&y_{s_0}^2\\
        \vdots&\vdots&&\vdots\\
        y_1^{s_0}&y_2^{s_0}&\cdots&y_{s_0}^{s_0}
    \end{pmatrix}
    \begin{pmatrix}
        e(x_1)\\
        e(x_2)\\
        \vdots\\
        e(x_{s_0})
    \end{pmatrix}
    =0.
\end{equation}
The determinant is
$\left(\prod_{i=1}^{s_0}y_i\right)
\prod_{1\le i<j\le s_0}(y_j-y_i)\ne0$,
because $\phi$ is a bijection into $\mathcal U^*$.
Hence every $e(x_i)$ is zero in $\mathcal U$,
contradicting the definition of $J_0$.
Therefore, $s_0\ge\ell$.

Choose distinct $x_1,\ldots,x_\ell\in J_0$ and fix the labels
of all elements in $J\setminus\{x_1,\ldots,x_\ell\}$.
Let $e_i:=e(x_i)$ and $a_i:=\phi(x_i)$, and define
\begin{equation}
    T(Z):=
    \prod_{x\in J\setminus\{x_1,\ldots,x_\ell\}}
        (Z-\phi(x))^{e(x)}
    \bmod Z^\ell.
\end{equation}
Then $T$ is fixed and $T(0)\ne0$.
The collision condition becomes
\begin{equation}
    T(Z)\prod_{i=1}^{\ell}(Z-a_i)^{e_i}
    \equiv1\pmod{Z^\ell}.
\end{equation}
Since each $e_i$ is nonzero in $\mathcal U$,
Lemma~\ref{lem:bounded-system-solutions} implies that there
are only $O_\ell(1)$ possible tuples $(a_1,\ldots,a_\ell)$
with nonzero, pairwise distinct coordinates.

Finally, after fixing the other $s-\ell$ labels, the ordered
tuple $(\phi(x_1),\ldots,\phi(x_\ell))$ is sampled uniformly
without replacement from the remaining $N-s+\ell$ labels.
As $s\le d+\ell$ and $d\le V/2$, we have $N-s+\ell\ge N-d\ge(V-2)/2=\Omega(N)$.
The number of possible ordered tuples is hence
$\binom{N-s+\ell}{\ell}\ell!=\Omega_\ell(N^\ell)$.
Therefore, for every fixing of the other $s-\ell$ labels,
\begin{equation}
    \Pr_\phi\!\left[
        \Psi_f(Z)=\Psi_g(Z)
        \,\middle|\,
        \text{the fixed labels}
    \right]
    =O_\ell(N^{-\ell}).
\end{equation}
Averaging over the fixed labels gives the claimed bound.
\end{proof}

\subsection{Heuristic Choices for Parameters $a$ and $z$}
\label{app:heuristic-az}

This subsection explains the heuristic choices in
\eqref{eq:heuristic-az-choice}.
We combine Walzer's density-evolution analysis
\cite{PeelingOrientabilityThreshold-Walzer-2025}
with the finite-length scaling picture for spatially coupled codes
\cite{ScalingLawSPCC-OlmosUrbanke-2015,
FiniteLengthScalingWD-SokolovskiiGraellAmatBrannstrom-2020},
and adapt them to our circular placement rule and capacity-$\ell$ cells.

\subsubsection{Choosing $a$}

We first describe how $a$ affects the density along the circle.
Recall that circular placement arranges the $M$ cells around a
circle of circumference $L:=z+1$, with each placement window
having length $1$.

Write $\mathbb T_L:=\mathbb R/L\mathbb Z$ for this circle and
$N:=M/L$ for the number of cells per unit length.
For each edge, we choose the starting position of its window,
called its \emph{anchor} $s$, uniformly from $[0,z+a)$.
Its $k$ endpoints are then sampled independently from
$[s,s+1]\subseteq\mathbb T_L$.
Windows wrap around the circle, so all intervals below are
understood modulo $L$.
Thus, $a$ controls the length of the arc available for anchors.
There are $d/(z+a)$ anchors per unit length on average
within this arc.
Dividing by the $N$ cells per unit length, we define the
density parameter
\begin{equation}
    c:=\frac{d}{N(z+a)},
\end{equation}
which is the average number of anchors per cell within
the allowed anchor region.

To understand the \emph{load} at position $u$, consider which anchors can place an endpoint there.
A window $[s,s+1]$ contains $u$ precisely when $s\in[u-1,u]$,
and the anchor must also belong to $[0,z+a)$.
We therefore define its load as
\begin{equation}
    w_a(u):=\bigl|[u-1,u]\cap[0,z+a)\bigr|,
    \qquad u\in\mathbb T_L,
\end{equation}
where $|\cdot|$ denotes arc length.
Each unit of anchor length contributes $cN$ edges on average,
and each edge places $k$ endpoints among approximately $N$ cells.
Thus, $cw_a(u)$ describes the edge density at position $u$,
corresponding to an expected degree of approximately $ckw_a(u)$.
In particular, $\min_u w_a(u)=a$, attained on $u\in[0,a]$.

We now use density evolution to explain why peeling is expected
to start in the sparsely loaded region $[0,a]$.
Let $q(u)$ denote the probability that a cell at position $u$ cannot peel a given incident edge, assuming that the edge survives at its other endpoints.
This happens when at least $\ell$ other incident edges remain.
Following Walzer~\cite{PeelingOrientabilityThreshold-Walzer-2025}, we use a Poisson approximation for their number.
This gives the update
\footnote{Intuitively, the inner integral accounts for an edge surviving at its other endpoints, while the outer integral combines these edges to estimate whether the cell remains too crowded to peel.}
\begin{equation}
    (\mathbf P_{a,z,c}q)(u)
    :=
    Q\left(
        ck\int_{[u-1,u]\cap[0,z+a)}
        \left(\int_0^1q(s+t)\,dt\right)^{k-1}ds,
        \ell
    \right),
\end{equation}
where $Q(\lambda,\ell):=\Pr[\operatorname{Pois}(\lambda)\ge\ell]$.
Starting from $q\equiv1$, each round replaces $q$ by
$\mathbf P_{a,z,c}q$.
For the first round, the inner integral equals $1$, so
\begin{equation}
    (\mathbf P_{a,z,c}1)(u)=Q(ckw_a(u),\ell).
\end{equation}
Since $Q(\lambda,\ell)$ increases with $\lambda$, this probability
is smallest where $w_a(u)$ is smallest, namely on $[0,a]$.
Cells in this region are therefore the most likely to peel
their incident edges.
Removing these edges reduces the load on nearby cells,
allowing peeling to propagate along the circle.

Recall that increasing $a$ improves space utilization but raises the minimum density to $ac$.
If even this sparse region $[0,a]$ is too dense, peeling may fail to start.
Since we aim for $c\approx c^{\mathsf{orient}}_{k,\ell}$,
requiring $ac\lesssim c^{\mathsf{peel}}_{k,\ell}$ suggests
\begin{equation}
    a=C\frac{c^{\mathsf{peel}}_{k,\ell}}
            {c^{\mathsf{orient}}_{k,\ell}},
    \qquad 0<C<1,
\end{equation}
where $C$ leaves a margin for finite-size fluctuations and is
calibrated experimentally.

\subsubsection{Choosing $z$}

Once $a$ is chosen, we turn to $z$.
Recall that $N=M/(z+1)$ and $d=cN(z+a)$, so
\begin{equation}
    \frac dM
    =c\frac{z+a}{z+1}
    =c\left(1-\frac{1-a}{z+1}\right).
\end{equation}
Increasing $z$ brings this ratio closer to $c$, improving space
utilization.
However, it also reduces the number $N$ of cells in each coupling
window, making relative load fluctuations larger.
More spare capacity is then needed to keep peeling from stalling.
Following Walzer~\cite{PeelingOrientabilityThreshold-Walzer-2025},
we balance these two sources of density loss.

The first comes from the sparsely loaded boundary.
Its relative contribution is
\begin{equation}
    \Delta_{\mathrm{len}}
    =\frac{1-a}{z+1}
    \approx\frac{1-a}{z},
\end{equation}
which decreases as $z$ grows.

To estimate the second loss term, we consider the number of cells available for peeling.
Let $s$ be the number of edge-removal steps, $\tau=s/N$ the normalized peeling time, and
\begin{equation}
A(\tau)
:=
\left|
\left\{
v\in[M]:
1\le
\deg^{(\lfloor N\tau\rfloor)}(v)
\le\ell
\right\}
\right|
\end{equation}
the number of currently pure cells, where $\deg^{(r)}(v)$
denotes the residual degree of cell $v$ after $r$ edge removals.
By definition, peeling stalls if $A(\tau)$ reaches zero while edges remain.
Following the finite-length scaling picture for spatially coupled codes~\cite{ScalingLawSPCC-OlmosUrbanke-2015, FiniteLengthScalingWD-SokolovskiiGraellAmatBrannstrom-2020}, let $\Delta:=c^{\mathsf{orient}}_{k,\ell}-c>0$ denote the density margin reserved to accommodate finite-size fluctuations.
During propagation, we assume the linear approximation
\footnote{This follows the near-threshold linear approximation in~\cite[Section III-D]{ScalingLawSPCC-OlmosUrbanke-2015}}
$\mathbb E[A(\tau)]=\Theta(N\Delta)$ and a sub-Gaussian lower-tail bound:
\begin{equation}
    \Pr\!\left[
        A(\tau)-\mathbb E[A(\tau)]\le-x
    \right]
    \lesssim
    \exp\left(-\Theta\left(\frac{x^2}{N}\right)\right).
\end{equation}
Reaching zero requires a downward fluctuation of order $N\Delta$, giving the exponential factor $\exp(-\Theta(N\Delta^2))$.

Nearby edge-removal steps are strongly correlated.
Following the same scaling picture, we assume that fluctuations become substantially less correlated after $\Theta(N)$ edge removals, corresponding to $\Theta(1)$ units of normalized time.
The total normalized peeling time is
$d/N=c(z+a)=\Theta(z)$, suggesting $\Theta(z)$ opportunities
for peeling to stall.
We therefore estimate
\begin{equation}
    \Pr[\text{peeling stalls}]
    \lesssim z\exp\bigl(-\Theta(N\Delta^2)\bigr).
\end{equation}
For a target failure probability $\delta$, this suggests
a density margin of
\begin{equation}
    \Delta_{\mathrm{fluc}}
    =\Theta\left(\sqrt{\frac{\log(z/\delta)}{N}}\right)
    =\Theta\left(\sqrt{\frac{z\log(z/\delta)}{M}}\right).
\end{equation}
Unlike the boundary loss, this term increases with $z$.

Combining the two contributions from $\Delta_{\mathrm{len}}$ and $\Delta_{\mathrm{fluc}}$, we estimate the total density
loss by
\begin{equation}
    \kappa_1\frac{1-a}{z}
    +\kappa_2\sqrt{\frac{z\log(z/\delta)}{M}},
\end{equation}
where $\kappa_1,\kappa_2>0$ are constants.
Balancing the two terms suggests
\begin{equation}
    z^3\log(z/\delta)=\Theta\bigl((1-a)^2M\bigr).
\end{equation}
This relation determines $z$ implicitly.

To obtain a simple explicit rule without introducing the
Lambert $W$ function, we omit the slowly varying $\log z$
contribution in $\log(z/\delta)=\log(1/\delta)+\log z$ and use
\begin{equation}
    z=D(1-a)^{2/3}
    \left(\frac{M}{\log(1/\delta)}\right)^{1/3},
\end{equation}
where $D>0$ is calibrated experimentally.
\section{Proofs for Barriers and Optimality}
\label{app:barriers}

This appendix gives the full proofs of the results in
Section~\ref{sec:barriers}, using the model and notation of Section~\ref{sec:barrier-model}. The proof of Theorem~\ref{thm:uniform-extremality} is given separately in Appendix~\ref{app:extremality}.

\subsection{Proof of Theorem~\ref{thm:barrier1-precise}}
\label{app:proof-zero-error}

\begin{proof}
Suppose for contradiction that the stated zero-error guarantee holds.
Since $\binom{\mathcal U^*}{d}$ is finite, there exists a fixed choice
$\omega_0$ of the public randomness such that
$\mathsf{Dec}^{\omega_0}(\mathsf{Enc}^{\omega_0}(A))=A$
for every $A\in\binom{\mathcal U^*}{d}$.
Fix $\omega_0$ and omit it from the notation.

There are at most $N_d:=\sum_{i=0}^{k}\binom{M_d}{i}$
possible write supports. Hence some $R\subseteq[M_d]$ has a preimage
$Q_R:=\{x\in\mathcal U^*:S(x)=R\}$
of size at least $(V-1)/N_d=\Omega_d(V)$.
Let $u:=|R|\le k$ and consider all $d$-element subsets of $Q_R$.
Since every inserted element can modify only cells in $R$, all
resulting sketches differ only on these $u$ cells. Hence there are
at most $V^{\ell u+o_V(1)}$ possible final states.
On the other hand, there are
\begin{equation}
    \binom{|Q_R|}{d}=\Omega_d(V^d)
\end{equation}
valid inputs. Since $d\ge\ell k+1>\ell u$, for sufficiently large
$V$ the number of inputs exceeds the number of possible states.
Thus, two distinct valid inputs yield the same sketch, contradicting
zero-error recovery.
\end{proof}

\subsection{Proof of Lemma~\ref{lem:orientability-necessary}}
\label{app:proof-orientability}

\begin{proof}
Fix $d$ and $\omega$. We will use the following equivalent way to sample $D$: first sample the multiplicity of each write support type ({\ie} sample $S(D):=\{S(x):x\in D\}$, with multiplicities retained), and then sample the \emph{identities} within each type uniformly.
Note that orientability depends only on these multiplicities.
Once the type is fixed, the remaining identities can be sampled in any order.

Call a write support $R\subseteq[M_d]$ \emph{good} if its preimage under $S$ has size at least $V/\log_2 V$. 
Hence for a fixed non-good support $R$, we have $\Pr_{X\sim \mathrm{Unif}(\mathcal U^*)} [S(X)=R]\le\frac{V/\log_2 V}{V-1}.$
Since there are at most $N_d:=\sum_{i=0}^{k}\binom{M_d}{i}$ possible supports, the union bound gives that $D$ contains an element with a non-good write support with probability at most $dN_d\frac{V/\log_2 V}{V-1}=o_V(1)$.

Condition on $S(D)$ being entirely good, if $H^\omega_{d,V}(D)$ is not $\ell$-orientable,
Hall's theorem yields a submultiset $\mathcal E$ of $t$ edges with vertex union $R:=\bigcup_{e\in\mathcal E}e$ of size $r$ such that $t>\ell r$.
Choose such $\mathcal E$ deterministically from the multiplicities, and first sample all the identities of elements outside $\mathcal E$.
Let $P$ be this fixed outside set. For each support type $G$, let $m_G$ denote its multiplicity in $\mathcal E$.
The remaining identities are uniform over
\begin{equation}
    \prod_{G:m_G>0}
    \binom{|S^{-1}(G)\setminus P|}{m_G}
    =V^{t-o_V(1)}
\end{equation}
possible candidates, since each good support has at least $(V/\log_2 V)-d$ unused elements.

By canonicality, we can insert $P$ before the remaining elements.
Every remaining element modifies only cells in $R$, so all resulting sketches agree outside $R$.
Hence they have at most $V^{\ell r+o_V(1)}$ possible states.
We use the elementary fact that a uniformly random variable on $N$ possibilities cannot be recovered from an observation with at most $M$ possible values with probability exceeding $M/N$, even with randomized encoders and decoders.
Therefore, the conditional average success probability is at most
\begin{equation}
\frac{N_{\mathrm{state}}}{N_{\mathrm{candidates}}}
\le
V^{\ell r-t+o_V(1)}
\le
V^{-1+o_V(1)}
=
o_V(1),
\end{equation}
Averaging over the writing support multiplicities and all outside identities, and then restoring the $o_V(1)$ probability of a non-good support, proves the first claim.
Splitting the success event according to orientability immediately gives the stated upper bound on $\operatorname{Succ}_{d,V}^{\omega}(\mathcal A)$.
\end{proof}


\subsection{Proof of Theorem~\ref{thm:main-precise}}
\label{app:proof-optimality}

Due to its complexity, we defer the proof of Theorem \ref{thm:uniform-extremality} to Appendix \ref{app:extremality} and directly prove Theorem~\ref{thm:main-precise} based on it.
In short, we will first pass from a uniformly sampled set to independent samples so that Theorem~\ref{thm:uniform-extremality} applies.

\begin{proof}
Fix $\omega$, let $X_1,\ldots,X_d$ be independent uniform samples
from $\mathcal U^*$, and define
\begin{equation}
    \widetilde H^\omega_{d,V}
    :=\bigl([M_d],\{S(X_i):i\in[d]\}\bigr).
\end{equation}
Conditioned on $X_1,\ldots,X_d$ being distinct, their underlying set is uniform over $\binom{\mathcal U^*}{d}$, and the induced hypergraph
has the same distribution as $H^\omega_{d,V}(D)$.
This conditioning fails with probability at most $\binom d2/(V-1)=o_V(1)$ for fixed $d$.
Therefore, Lemma~\ref{lem:orientability-necessary} gives
\begin{equation}
\label{eq:app-barriers-iid}
    \operatorname{Succ}_{d,V}^{\omega}(\mathcal A)
    \le
    \Pr\!\left[
        \widetilde H^\omega_{d,V}\text{ is $\ell$-orientable}
    \right]+o_V(1).
\end{equation}

For a w.h.p. successful family, $M_d\ge d/\ell$ for all sufficiently
large $d$, otherwise orientability is impossible, contradicting w.h.p. success by Lemma~\ref{lem:orientability-necessary}.
In particular, $M_d\to\infty$ as $d\to\infty$.

Suppose that the claimed lower bound fails. Then, for some
$\varepsilon>0$ and infinitely many $d$,
$d/M_d\ge c^{\mathsf{orient}}_{k,\ell}+\varepsilon$.
For every $\omega$ and $V$, the hypergraph
$\widetilde H^\omega_{d,V}$ has distribution $H^\mu_{M_d,d}$,
where $\mu$ is the distribution of $S(X_1)$.
Since orientability is preserved under deleting edges,
Theorem~\ref{thm:uniform-extremality} gives,
\begin{equation}
    \Pr\!\left[
        \widetilde H^\omega_{d,V}\text{ is $\ell$-orientable}
    \right]=o_d(1),
\end{equation}
where $o_d(1)$ is uniform in $\omega$ and $V$.
Averaging \eqref{eq:app-barriers-iid} over $\omega$ yields
$\operatorname{Succ}_{d,V}(\mathcal A)\le o_d(1)+o_V(1)$.
Taking $\liminf_{V\to\infty} \operatorname{Succ}_{d,V}(\mathcal A)\le o_d(1)$ and then letting $d\to\infty$
along this subsequence contradicts w.h.p. success.
The final inequality about XYZ-Sketch follows from
Lemma~\ref{thm:orientability-benchmark}.

Finally, for every $c<c^{\mathsf{orient}}_{k,\ell}$, the uniform random hypergraph is $\ell$-orientable w.h.p. and has $(e^{-kc}+o(1))n$ isolated vertices.
\footnote{With $\lfloor cn\rfloor$ independently sampled edges,
each vertex is isolated with probability
$(1-k/n)^{\lfloor cn\rfloor}=e^{-kc}+o(1)$.
The isolation indicators of two distinct vertices have covariance
$O(1/n)$, so the number of isolated vertices has variance $O(n)$.
The claim follows from Chebyshev's inequality.}
Since each nonisolated vertex can receive at most $\ell$ edges, orienting all $cn+O(1)$ edges requires $c\le\ell(1-e^{-kc})$ in the limit, or equivalently, $\frac{\ell}{c}\ge \frac{1}{1-e^{-kc}}$.
Letting $c\uparrow c^{\mathsf{orient}}_{k,\ell}$ and using
$c^{\mathsf{orient}}_{k,\ell}\le\ell$ gives
\begin{equation}
    \frac{\ell}{c^{\mathsf{orient}}_{k,\ell}}
    \ge \frac{1}{1-e^{-k\ell}}.
\end{equation}
Thus, a space ratio at most $1+\varepsilon$ requires
$k\ell\ge\ln(1+1/\varepsilon)$.
\end{proof}

\subsection{Proof of Corollary~\ref{thm:barrier2-precise}}
\label{app:proof-overhead}

\begin{proof}
The lower-bound argument above applies to every fixed $k\ge2$ and
$\ell\ge1$; only the matching upper bound for {\fullalgo} requires
$\ell\ge2$. Since $c^{\mathsf{orient}}_{k,\ell}<\ell$ for every fixed $k\ge 2$ and $\ell\ge 1$~\cite{PeelingOrientabilityThreshold-Walzer-2025}, we obtain
\begin{equation}
    \liminf_{d\to\infty}
    \frac{\mathfrak C_{\mathcal A}(d)}{d}
    \ge \frac{\ell}{c^{\mathsf{orient}}_{k,\ell}}>1.
\end{equation}
\end{proof}
\section{Optimality of the Uniform Distribution}
\label{app:extremality}

In this section, we first introduce the related matroid notions, then prove Theorem~\ref{thm:uniform-extremality}.
Recall that $\alpha_\ell(H)$ denotes the maximum number of edges in an $\ell$-orientable subhypergraph of $H$, and  $\mu^*$ denotes the uniform distribution on $\binom{[n]}{k}$.

\newtheorem*{restatedextremality}{Theorem~\ref{thm:uniform-extremality} (restated)}
\begin{restatedextremality}
Recall that $n$ is the number of vertices, $m$ is the number of edges, $k$ is the maximum edge size, and $\ell$ is the capacity of each vertex in an $\ell$-orientation.
Recall also that $\alpha_\ell(H)$ is the maximum number of edges
in an $\ell$-orientable subhypergraph of $H$.
For every $n\ge k\ge1$, $\ell\ge1$, $m\ge0$, and every distribution
$\mu$ on $\binom{[n]}{\le k}$, we have:
\begin{enumerate}
    \item \label{thm:expectation} Expectation Extremality:
    \begin{equation}
        \mathbb E[\alpha_\ell(H^\mu_{n,m})]\le \mathbb E[\alpha_\ell(H^{\mu^*}_{n,m})].
    \end{equation}
    \item \label{thm:probability-bound} Probability Bound: for every fixed $k,\ell\ge1$ and $\varepsilon>0$,
    \begin{equation}
    \sup_{\mu}
    \Pr\!\left[
    H^\mu_{n,\lceil(c^{\mathsf{orient}}_{k,\ell}+\varepsilon)n\rceil}
    \text{ is $\ell$-orientable}
    \right]
    =o_n(1),
    \end{equation}
    where the supremum ranges over all distributions on
    $\binom{[n]}{\le k}$.
\end{enumerate}
\end{restatedextremality}

\subsection{Matroid Preliminaries}
\label{app:matroid}

We briefly introduce the related matroid notions, while readers may refer to Oxley~\cite[Chapters~1 and~3]{Oxley2011MatroidTheory} for a comprehensive introduction.

A \emph{matroid} $\mathcal M=(E,\mathcal I)$ consists of a finite
set $E$, called the \emph{ground set}, and a nonempty family $\mathcal I\subseteq 2^E$ of \emph{independent sets} with the following properties:
\begin{itemize}
    \item If $X\subseteq Y$ and $Y\in \mathcal I$, then $X\in \mathcal I$.
    \item If $X,Y\in\mathcal I$ with $|X|<|Y|$, then
$\exists e\in Y\setminus X,\ X\cup\{e\}\in\mathcal I.$
\end{itemize}

The \emph{rank function} of a matroid $\mathcal M$ is
\begin{equation}
    r_{\mathcal M}(X):=\max\{|Y|:Y\subseteq X,\ Y\in\mathcal I\}.
\end{equation}
It can be shown that, the rank function is \emph{monotone}: $X\subseteq Y\implies r_{\mathcal M}(X)\le r_{\mathcal M}(Y)$. It is also \emph{submodular}: for all
$X,Y\subseteq E$,
\begin{equation}
    r_{\mathcal M}(X)+r_{\mathcal M}(Y)
    \ge r_{\mathcal M}(X\cup Y)+r_{\mathcal M}(X\cap Y).
\end{equation}

An element $e\in E$ is a \emph{loop} if $r_{\mathcal M}(\{e\})=0$.
Two distinct non-loop elements $e,f$ are \emph{parallel} if
$r_{\mathcal M}(\{e,f\})=1$.
On the non-loop elements, being equal or parallel is an equivalence relation,
\footnote{Reflexivity and symmetry are immediate. For transitivity,
write $r=r_{\mathcal M}$. If $x,y,z$ are non-loops with
$r(\{x,y\})=r(\{y,z\})=1$, monotonicity and submodularity gives
$r(\{x,z\})\le r(\{x,y,z\})\le r(\{x,y\})+r(\{y,z\})-r(\{y\})=1$.
Thus, $x$ and $z$ are equal or parallel.}
whose equivalence classes are called \emph{parallel classes}.

For $T\subseteq E$, the \emph{contraction} $\mathcal M/T$ is the matroid on $E\setminus T$ whose rank satisfies
\footnote{The rank function uniquely determines
the independent sets, since $I$ is independent iff $r(I)=|I|$, and hence determines the matroid.}
\begin{equation}
\label{eq:contraction-formula}
    r_{\mathcal M/T}(U)
    =r_{\mathcal M}(U\cup T)-r_{\mathcal M}(T),
    \qquad U\subseteq E\setminus T.
\end{equation}
A bipartite graph defines a \emph{transversal matroid} on one of its vertex classes, where a subset is independent precisely when it can be matched to distinct vertices in the other class.

\subsection{A Quadratic Inequality for Matroid Rank}

The following lemma proves the negative semidefiniteness of a quadratic form by parallel classes.

\begin{lemma}[Negative Semidefiniteness of the Parallel-class Quadratic Form]
\label{lem:extremality-quadratic-form}
Let $r$ be the rank function of a matroid $(E,\mathcal I)$, and let $\{x_e,y_e:e\in\Omega\}\subseteq E$ be pairwise distinct elements, where $\Omega$ is a finite index set.
Assume that, for every $e\in\Omega$ and every
$U\subseteq E\setminus\{x_e,y_e\}$, there is 
\begin{equation}
r(U\cup\{x_e\})=r(U\cup\{y_e\}).
\end{equation}
Then for every real vector $h=(h_e)_{e\in\Omega}$ with $\sum_{e\in\Omega}h_e=0$,
\begin{equation}
    \sum_{e,f\in\Omega}h_eh_f\,r(\{x_e,y_f\})\le0.
\end{equation}
\end{lemma}

\begin{proof}
Taking $U=\varnothing$ in the assumption gives $r(\{x_e\})=r(\{y_e\})$.
Define $c_e:=r(\{x_e\})=r(\{y_e\}),\ d_{e,f}:=c_e+c_f-r(\{x_e,y_f\}).$
Expanding the quadratic form gives
\begin{equation}
    \begin{aligned}
    \sum_{e,f\in\Omega}h_eh_f\,r(\{x_e,y_f\})
    &=\sum_{e,f}h_eh_f(c_e+c_f-d_{e,f})\\
    &=\left(\sum_eh_ec_e\right)\left(\sum_fh_f\right)
      +\left(\sum_eh_e\right)\left(\sum_fh_fc_f\right)
      -\sum_{e,f}h_eh_fd_{e,f}\\
    &=-\sum_{e,f}h_eh_fd_{e,f}.
    \end{aligned}
\end{equation}
The first two terms vanish because $\sum_eh_e=0$.
For the last term related to $d_{e,f}$, considering separately the cases where at least one element is a loop, both are non-loops and parallel, or both are non-loops and not parallel, we obtain
\begin{equation}
    d_{e,f}=
    \begin{cases}
        1, & \text{if $x_e,y_f$ are parallel non-loop elements},\\
        0, & \text{otherwise}.
    \end{cases}
\end{equation}

Consider parallel classes among the non-loop elements $x_e,y_e$.
Specifically, let $C$ be a parallel class containing $x_e$ and $y_f$.
If $e=f$, then $x_e,y_e\in C$ already.
If $e\ne f$, the assumption gives
\begin{equation}
    r(\{y_e,y_f\})=r(\{x_e,y_f\})=1,
    \qquad
    r(\{x_e,x_f\})=r(\{x_e,y_f\})=1.
\end{equation}
Since $y_e$ and $x_f$ are also non-loops, they belong to $C$.
Applying this argument to each element of $C$ shows that
$x_e\in C$ if and only if $y_e\in C$.
Let $\mathcal P$ denote these classes, identified with their sets of indices. Grouping the last term by these classes gives
\begin{equation}
\label{eq:extremality-negative-squares}
    \sum_{e,f\in\Omega}h_eh_f\,r(\{x_e,y_f\})
    =-\sum_{e,f}h_eh_fd_{e,f}=-\sum_{C\in\mathcal P}\left(\sum_{e\in C}h_e\right)^2
    \le0.
\end{equation}
\end{proof}

\subsection{Proof of the Expectation Extremality}
\label{app:extremality-expectation}

In this part, we prove (\ref{thm:expectation}) in Theorem \ref{thm:uniform-extremality}.

\begin{proof}

We first reduce distributions on $\binom{[n]}{\le k}$ to distributions on $\binom{[n]}{k}$, as adding vertices to an edge only gives it more choices for its orientation.
For each sampled edge $e$, independently choose a $k$-element set containing $e$, uniformly among all such sets, and replace $e$ by it.
The resulting edges are still independent and identically distributed, and $\alpha_\ell$ cannot decrease.
It therefore suffices to prove the expectation bound for distributions on $\Omega:=\binom{[n]}{k}$.

The case for $m=0$ or $m=1$ is trivial as $\mathbb E[\alpha_\ell(H^\mu_{n,m})]=m$ for every such distribution, so we assume $m\ge2$.

\paragraph{Matching representation of Orientability.}
We construct a bipartite graph to represent $\ell$-orientations.
\begin{itemize}
    \item On the left, for each $i\in[m]$ and $e\in\Omega$, create a vertex $a_{i,e}$ representing the possibility that $E_i=e$, where $E_i$ is the $i$-th sampled edge.
    \item On the right, create $\ell$ slots $(v,1),\ldots,(v,\ell)$ for each hypergraph vertex $v\in[n]$, representing its capacity to receive up to $\ell$ edges.
    \item Connect $a_{i,e}$ to $(v,j)$ if and only if $v\in e$.
\end{itemize}

For a realization $E_1,\ldots,E_m$, select the left vertices
$\{a_{1,E_1},\ldots,a_{m,E_m}\}$.
Matching $a_{i,E_i}$ to a slot $(v,j)$ means assigning edge $E_i$ to vertex $v$.
Since each slot can be used at most once, a matching assigns at most $\ell$ edges to each vertex and therefore gives an $\ell$-orientation of the matched edges.
Conversely, any $\ell$-orientation can be also represented by a matching, where the edges assigned to each vertex use distinct slots of that vertex.
Thus, $\alpha_\ell(H^\mu_{n,m})$ equals the maximum number of selected left vertices that can be matched.

As introduced in Appendix \ref{app:matroid}, the subsets of the left vertex set that can be matched to distinct right vertices form a transversal matroid $\mathcal M$.
Its rank is the maximum size of a matchable subset, so
\begin{equation}
\label{eq:nu-equal-rank}
    \alpha_\ell(H^\mu_{n,m})
    =r_{\mathcal M}(\{a_{1,E_1},\ldots,a_{m,E_m}\}),
\end{equation}
where $H^\mu_{n,m}$ is the multihypergraph formed by the sampled edges $E_1,\ldots,E_m$.

\paragraph{Concavity.}
Write $R(\mu):=\mathbb E[\alpha_\ell(H^\mu_{n,m})]$, our goal is to prove the concavity of $R(\mu)$.\footnote{We regard each
distribution $\mu$ as a vector $(\mu_e)_{e\in\Omega}$, so $R$ is a
function on the probability simplex
$\{\mu\in\mathbb R_{\ge0}^{\Omega}:\sum_{e\in\Omega}\mu_e=1\}$,
which is a convex set.}
In fact, it suffices to show that its restriction to every line segment in the probability simplex is concave.
Specifically, fix a direction $h=(h_e)_{e\in\Omega}$ with $\sum_{e\in\Omega}h_e=0$, and let $F(t):=R(\mu+th),\; t\in\mathbb R$ along such a segment, it suffices to show that $F''(t)\le0$ for every direction $h$.

Moreover, showing that $F''(0)\le0$ is also sufficient: for any $t_0$ on the segment, set
$\widetilde\mu:=\mu+t_0h$ and define $\widetilde F(s):=R(\widetilde\mu+sh)=F(t_0+s)$.
Since $\widetilde\mu$ is also a distribution, the same argument
applied at $\widetilde\mu$ gives
$F''(t_0)=\widetilde F''(0)\le0$.
As $t_0$ is arbitrary on the segment, this proves that $F''(t)\le0$ throughout the segment.
We now begin the specific calculation.

Since $F(t)=\mathbb E[\alpha_\ell(H^{\mu+th}_{n,m})]$, expanding the expectation and substituting (\ref{eq:nu-equal-rank}) into it gives
\begin{equation}
    F(t)=\sum_{(e_1,\ldots,e_m)\in\Omega^m}
        r_{\mathcal M}(\{a_{1,e_1},\ldots,a_{m,e_m}\})
        \prod_{i=1}^m(\mu_{e_i}+th_{e_i}).
\end{equation}
Differentiating twice in $t$ selects an ordered pair of distinct factors in the product.
By symmetry, all $m(m-1)$ choices contribute equally, so we may take the differentiated factors to be those corresponding to $e_{m-1}$ and $e_m$.
Denoting $T:=\{a_{1,E_1},\ldots,a_{m-2,E_{m-2}}\}$, we have: 
\begin{equation}
\label{eq:extremality-second-derivative}
    F''(0)
    =m(m-1)\mathbb E_{E_1,\ldots,E_{m-2}\overset{\mathrm{i.i.d.}}{\sim}\mu}
    \!\left[
        \sum_{e,f\in\Omega}h_eh_f\,
        r_{\mathcal M}(T\cup\{a_{m-1,e},a_{m,f}\})
    \right].
\end{equation}

Fix the first $m-2$ edges, thus $T$ is also fixed.
By the contraction formula (\ref{eq:contraction-formula}),
\begin{equation}
    r_{\mathcal M}(T\cup\{a_{m-1,e},a_{m,f}\})
    =r_{\mathcal M/T}(\{a_{m-1,e},a_{m,f}\})+r_{\mathcal M}(T).
\end{equation}
The constant term $r_{\mathcal M}(T)$ vanishes because
$\sum_{e,f}h_eh_fr_{\mathcal M}(T)=r_{\mathcal M}(T)(\sum_eh_e)^2=0$.
Thus, it suffices to show that, for every $T$,
\begin{equation}
    \sum_{e,f\in\Omega}h_eh_f\,
    r_{\mathcal M/T}(\{a_{m-1,e},a_{m,f}\})\le0.
\end{equation}

We now use Lemma~\ref{lem:extremality-quadratic-form} to conclude.
In fact, the elements $a_{m-1,e},a_{m,e}$ have identical neighbors, so exchanging them preserves ranks, both before and after contracting $T$.
Formally, let $E$ be the ground set of $\mathcal M$ and take any
$U\subseteq E\setminus(T\cup\{a_{m-1,e},a_{m,e}\})$.
Since the two elements have identical neighbors, replacing one by
the other in a matching preserves its size. Hence,
\begin{equation}
    \begin{aligned}
    r_{\mathcal M/T}(U\cup\{a_{m-1,e}\})
    &=r_{\mathcal M}(T\cup U\cup\{a_{m-1,e}\})-r_{\mathcal M}(T)\\
    &=r_{\mathcal M}(T\cup U\cup\{a_{m,e}\})-r_{\mathcal M}(T)\\
    &=r_{\mathcal M/T}(U\cup\{a_{m,e}\}).
    \end{aligned}
\end{equation}
Lemma~\ref{lem:extremality-quadratic-form} therefore applies to the matroid $\mathcal M/T$, with $x_e=a_{m-1,e}$ and $y_e=a_{m,e}$.
Hence $F''(0)\le0$ and $R$ is concave on the probability simplex.

\paragraph{Permutation averaging.}
Let $\mathfrak S_n$ be the set of permutations of $[n]$, and let $\pi\mu$ denote the distribution obtained by sampling an edge from $\mu$ and relabeling each vertex $v\in[n]$ as $\pi(v)$. Equivalently, $(\pi\mu)(e):=\mu(\pi^{-1}e).$

As relabeling vertices of the hypergraph leaves $\alpha_\ell$ unchanged, we have $R(\pi\mu)=R(\mu)$.
The average of $\pi\mu$ over all vertex permutations is $\mu^*$.
Jensen's inequality therefore gives
\begin{equation}
    R(\mu)
    =\frac1{n!}\sum_{\pi\in\mathfrak S_n}R(\pi\mu)
    \le R\!\left(\frac1{n!}\sum_{\pi\in\mathfrak S_n}\pi\mu\right)
    =R(\mu^*).
\end{equation}
Together with the initial enlargement, this proves the expectation
bound for every distribution on $\binom{[n]}{\le k}$.
\end{proof}

\subsection{Proof of the Probability Bound}
\label{app:extremality-probability}

The overall idea of this part is that, above $c^{\mathsf{orient}}_{k,\ell}$, even the optimal distribution $\mu^*$ falls short of orienting all edges on average, and so does every other distribution $\mu$.
We first establish this gap for $\mu^*$ by Lemma \ref{lem:extremality-uniform-gap}.
Expectation extremality then implies that the gap is at least as large for every other distribution $\mu$.
A concentration inequality turns this expectation gap into the desired probability bound.

\begin{lemma}[Positive Deficit under the Uniform Distribution]
\label{lem:extremality-uniform-gap}
Fix $k,\ell\ge1$ and $\varepsilon>0$, and set
$c:=c^{\mathsf{orient}}_{k,\ell}+\varepsilon$ and $m:=\lceil cn\rceil$.
There exists a constant $\eta_\varepsilon>0$ such that, for all
sufficiently large $n$,
\begin{equation}
\label{eq:extremality-uniform-gap}
    \mathbb E[\alpha_\ell(H^{\mu^*}_{n,m})]
    \le m-\eta_\varepsilon n.
\end{equation}
\end{lemma}

\begin{proof}
For $k\ge2$, Lelarge~\cite[Theorem~2]{Lelarge2012Orientation} shows that, above the orientability threshold, even the maximum partial orientation leaves a positive fraction of edges unoriented with high probability.
Thus, the expected number of unoriented edges is linear in $n$, giving the claimed bound.
His model samples distinct edges, whereas ours allows repetitions.
However, the expected number of repetitions is only $O(1)$, and replacing each repeated edge changes $\alpha_\ell$ by at most one, so this difference affects the expectation by only $O(1)$.

For $k=1$, we have $c^{\mathsf{orient}}_{1,\ell}=0$ (see Appendix~\ref{app:threshold-table}).
Let $N_v$ count the sampled singleton edges at vertex $v$.
Each such edge can only be assigned to $v$, which can accept at
most $\ell$ edges. Hence $\alpha_\ell(H)=\sum_{v=1}^n\min\{\ell,N_v\}$.

For each $v$, we have $N_v\sim\operatorname{Bin}(m,1/n)$.
Since $m/n\to c$, the distribution of $N_v$ converges to
$\operatorname{Pois}(c)$.
Thus, for $P\sim\operatorname{Pois}(c)$, symmetry among vertices and the boundedness of $\min\{\ell,N_v\}$ give
\begin{equation}
    \frac{\mathbb E[\alpha_\ell(H^{\mu^*}_{n,m})]}{n}
    =\mathbb E[\min\{\ell,N_1\}]
    \longrightarrow \mathbb E[\min\{\ell,P\}].
\end{equation}
By definition we have $\Pr[P\ge \ell+1]>0$.
Therefore, truncating $P$ at $\ell$ strictly decreases its expectation:
$\mathbb E[\min\{\ell,P\}]<\mathbb E[P]=c$.
Together with $m/n\to c$, this gives the claimed linear deficit.
\end{proof}

We now prove (\ref{thm:probability-bound}) in
Theorem~\ref{thm:uniform-extremality}.

\begin{proof}
Fix $k,\ell\ge1$ and $\varepsilon>0$, and set
$m:=\lceil(c^{\mathsf{orient}}_{k,\ell}+\varepsilon)n\rceil$.
For any distribution $\mu$ on $\binom{[n]}{\le k}$, write
$X:=\alpha_\ell(H^\mu_{n,m})$.
The expectation bound and Lemma~\ref{lem:extremality-uniform-gap} give
\begin{equation}
    \mathbb E[X]
    \le\mathbb E[\alpha_\ell(H^{\mu^*}_{n,m})]
    \le m-\eta_\varepsilon n.
\end{equation}

Now we apply concentration inequalities to derive the required bound. 
Let $E_1,\ldots,E_m$ be the independent edges sampled from $\mu$.
For each $i\in[m]$, let $E_i'$ be an independent sample from $\mu$,
independent of all $E_1,\ldots,E_m$.
Let $H^{(i)}$ be the hypergraph obtained by replacing $E_i$ with
$E_i'$ ({\ie}, $\{E_1,\cdots,E_{i-1},E_i',E_{i+1},\cdots ,E_m\}$), and set $X^{(i)}:=\alpha_\ell(H^{(i)})$.
Replacing one edge changes the maximum number of orientable edges
by at most one, so $|X-X^{(i)}|\le1$.
The Efron--Stein inequality%
\footnote{For independent random variables $Z_1,\ldots,Z_m$,
let $Y=f(Z_1,\ldots,Z_m)$ have finite variance, and let $Y^{(i)}$
be obtained by replacing $Z_i$ with an independent copy.
The Efron--Stein inequality states that
$\operatorname{Var}(Y)\le
\frac12\sum_{i=1}^m\mathbb E[(Y-Y^{(i)})^2]$.}
gives
\begin{equation}
    \operatorname{Var}(X)
    \le\frac12\sum_{i=1}^m\mathbb E[(X-X^{(i)})^2]
    \le\frac m2.
\end{equation}
Full orientability requires $X=m$, so Chebyshev's inequality yields
\begin{equation}
    \begin{aligned}
    \Pr[H^\mu_{n,m}\text{ is $\ell$-orientable}]
    &=\Pr[X=m]\\
    &\le\Pr[|X-\mathbb E[X]|\ge\eta_\varepsilon n]\\
    &\le\frac{m}{2\eta_\varepsilon^2n^2}
    =O(1/n)=o_n(1).
    \end{aligned}
\end{equation}
The bound is independent of $\mu$, so taking the supremum over $\mu$
proves the claim.
\end{proof}
\section{Additional Pseudocode for {\fullalgo}}
\label{app:pseudocode}

This appendix gives the detailed pseudocode used by Section~\ref{sec:algo}.
Algorithm~\ref{alg:pseudo:update} corresponds to the high-capacity cell update;
Algorithm~\ref{alg:pseudo:subtract} corresponds to sketch subtraction.
Algorithms~\ref{alg:pseudo:drfr}--\ref{alg:pseudo:half-gcd} give
D-RFR reconstruction and its subroutines.

\subsection{Update and Subtraction}

In the following pseudocode, the procedure $\mathsf{UpdateCell}$ modifies cell $C$ in place and implements $\mathsf{Update}(C,x,\sigma)$ from Section~\ref{sec:framework}.
Here, $\sigma=+1$ is used for inserting $x$ into the maintained set, and $\sigma=-1$ is used for deleting $x$ from it.

\begin{algorithm}
\caption{{\fullalgo}: Update Procedure}
\label{alg:pseudo:update}
\SetKwFunction{FUpdateCell}{UpdateCell}
\SetKwFunction{FUpdate}{Update}
\SetKwProg{Pn}{Procedure}{:}{}
\SetKwFor{For}{for}{do}{end}

\Pn{\FUpdateCell{$C,x,\sigma$}}{
    $c_C\gets (c_C+\sigma)\bmod (2\ell+1)$\;
    \eIf{$\sigma=+1$}{
        $P_C\gets P_C\cdot (Z-x)\bmod Z^\ell$\;
    }{
        $P_C\gets \dfrac{P_C}{Z-x}\bmod Z^\ell$\;
    }
    \If{fingerprint is enabled}{
        $\mathsf{fp}_C\gets (\mathsf{fp}_C+\sigma\fgp(x))\bmod q$\;
    }
}

\Pn{\FUpdate{$\mathcal B,x,\sigma$}}{
    \For{$i\in\Gamma(x)$}{
        \FUpdateCell{$\mathcal B[i],x,\sigma$}\;
    }
}
\end{algorithm}


\begin{algorithm}
\caption{{\fullalgo}: Subtract Procedure}
\label{alg:pseudo:subtract}
\SetKwFunction{FSubtract}{Subtract}
\SetKwProg{Fn}{Function}{:}{}
\SetKwFor{For}{for}{do}{end}

\Fn{\FSubtract{$\mathcal B_A,\mathcal B_B$}}{
    Initialize an empty residual sketch $\mathcal D$\;
    \For{$i\in[M]$}{
        $c_{\mathcal D[i]}\gets (c_{\mathcal B_A[i]}-c_{\mathcal B_B[i]})\bmod(2\ell+1)$\;
        $P_{\mathcal D[i]}\gets P_{\mathcal B_A[i]}/P_{\mathcal B_B[i]}\bmod Z^\ell$\;
        \If{fingerprint is enabled}{
            $\mathsf{fp}_{\mathcal D[i]}\gets (\mathsf{fp}_{\mathcal B_A[i]}-\mathsf{fp}_{\mathcal B_B[i]})\bmod q$\;
        }
    }
    \KwRet $\mathcal D$\;
}
\end{algorithm}

\subsection{D-RFR Reconstruction}
\label{app:drfr-pseudocode}

The following algorithms implement the reconstruction proof in
Appendix~\ref{app:drfr}, including the original RFR and its half-GCD
recursion. All polynomials are over $\mathcal U$, and $\deg 0:=-\infty$.
For nonzero $P$, let $\operatorname{lc}(P)$ denote its leading coefficient.
Every assignment to a tuple uses the old values on its right-hand side.
If a subroutine fails, an inverse does not exist, or a division required
to be exact fails or divides by zero, the calling algorithm returns
$\bot$.

\begin{algorithm}[htbp]
\caption{Degree-aware rational function reconstruction}
\label{alg:pseudo:drfr}
\SetKwProg{Fn}{Function}{:}{}
\KwIn{$\ell\ge1$, $m\in\mathbb Z$, and $R\in\mathcal U[Z]$ with $\deg R<\ell$.}
\KwOut{Monic coprime $F,G$ with $G(0)\ne0$, $\deg F-\deg G=m$,
$\deg F+\deg G\le\ell$, and $F\equiv GR\pmod{Z^\ell}$; or $\bot$.}
\Fn{$\mathsf{D\text{-}RFR}(R,m,\ell)$}{
    \lIf{$|m|>\ell$}{\KwRet $\bot$}
    \If{$R=0$}{
        \lIf{$m=\ell$}{\KwRet $(Z^\ell,1)$}
        \KwRet $\bot$\;
    }
    Let $v$ be the smallest exponent with a nonzero coefficient in $R$\;
    $(R_0,n,m_0)\gets(R/Z^v,\ell-v,m-v)$\;
    \lIf{$|m_0|>n$}{\KwRet $\bot$}
    $\mathit{swapped}\gets\mathrm{false}$\;
    \If{$m_0<0$}{
        $(R_0,m_0)\gets(R_0^{-1}\bmod Z^n,-m_0)$\;
        $\mathit{swapped}\gets\mathrm{true}$\;
    }
    \eIf{$m_0=0$}{
        \eIf{$R_0=1$}{
            $(F_0,G)\gets(1,1)$\;
        }{
            $D\gets\lfloor n/2\rfloor$\;
            $(H,G)\gets\mathsf{ImbalancedRFR}(R_0-1,n,D-1,D)$\;
            $(H,G)\gets(H/\operatorname{lc}(G),G/\operatorname{lc}(G))$\;
            $F_0\gets H+G$\;
        }
    }{
        $(d_G,d_H)\gets(\lfloor(n-m_0)/2\rfloor,
            \lfloor(n-m_0)/2\rfloor+m_0-1)$\;
        $S\gets(R_0-Z^{m_0})^{-1}\bmod Z^n$\;
        $(G,H)\gets\mathsf{ImbalancedRFR}(S,n,d_G,d_H)$\;
        $(G,H)\gets(G/\operatorname{lc}(G),H/\operatorname{lc}(G))$\;
        $F_0\gets H+Z^{m_0}G$\;
    }
    \lIf{$\mathit{swapped}$}{$(F_0,G)\gets(G,F_0)$}
    $F\gets Z^vF_0$\;
    \If{$F,G$ are not both monic, $G(0)=0$, or $\gcd(F,G)\ne1$}{
        \KwRet $\bot$\;
    }
    \If{$\deg F-\deg G\ne m$, $\deg F+\deg G>\ell$, or
        $F\not\equiv GR\pmod{Z^\ell}$}{
        \KwRet $\bot$\;
    }
    \KwRet $(F,G)$\;
}
\end{algorithm}

\begin{algorithm}[htbp]
\caption{Imbalanced RFR and the original RFR}
\label{alg:pseudo:rfr}
\SetKwProg{Fn}{Function}{:}{}
\KwOut{Both routines return a coprime numerator--denominator pair
up to a common nonzero scalar, or $\bot$.}
\Fn{$\mathsf{ImbalancedRFR}(S,n,d_A,d_B)$}{
    \tcp{$S\in\mathcal U[Z]$; $n,d_A,d_B$ are integers.}
    \If{$d_A<0$, $d_B<d_A$, $n<d_A+d_B+1$, or $\deg S\ge n$}{
        \KwRet $\bot$\;
    }
    $s\gets d_B-d_A$\;
    $\widetilde S\gets Z^sS\bmod Z^{2d_B+1}$\;
    $(A,B)\gets\mathsf{OriginalRFR}(\widetilde S,d_B)$\;
    $A\gets A/Z^s$ \tcp*{Exact polynomial division.}
    \If{$\deg A>d_A$, $\deg B>d_B$, or $A\not\equiv BS\pmod{Z^n}$}{
        \KwRet $\bot$\;
    }
    \KwRet $(A,B)$\;
}
\BlankLine
\Fn{$\mathsf{OriginalRFR}(R,D)$}{
    \tcp{$D\ge0$ is an integer; $\deg R<2D+1$.}
    \lIf{$R=0$}{\KwRet $(0,1)$}
    $N\gets2D+1$\;
    $M\gets\mathsf{HalfGCD}(Z^N,R)$\;
    Let $(x_1,y_1)$ be the second row of $M$\;
    $H_p\gets x_1Z^N+y_1R$\;
    \If{$y_1(0)=0$, $\gcd(H_p,y_1)\ne1$, or
        $\max\{\deg H_p,\deg y_1\}>D$}{
        \KwRet $\bot$\;
    }
    \KwRet $(H_p,y_1)$\;
}
\end{algorithm}

For polynomials $A,B$ with $\deg A=N>\deg B$,
$\mathsf{HalfGCD}(A,B)$ returns the Euclidean transformation matrix
$\mathcal H_N(A,B)$: its transformed pair has first degree at least
$\lceil N/2\rceil$ and second degree less than $\lceil N/2\rceil$.
Here $I_2$ is the $2\times2$ identity matrix, and
\begin{equation}
    \varphi(Q):=\begin{bmatrix}0&1\\1&-Q\end{bmatrix}.
\end{equation}
The quotient on division by $Z^L$ keeps the high-degree part and
shifts its exponents down by $L$.

\begin{algorithm}[htbp]
\caption{The half-GCD recursion used by the original RFR}
\label{alg:pseudo:half-gcd}
\SetKwProg{Fn}{Function}{:}{}
\KwIn{$A,B\in\mathcal U[Z]$ with $\deg A\ge1$ and $\deg A>\deg B$.}
\KwOut{The $2\times2$ polynomial matrix $\mathcal H_{\deg A}(A,B)$.}
\Fn{$\mathsf{HalfGCD}(A,B)$}{
    $N\gets\deg A$, $s\gets\lceil N/2\rceil$\;
    \lIf{$N\le1$ or $\deg B<s$}{\KwRet $I_2$}
    $L\gets s$\;
    Let $A_1,B_1$ be the quotients of $A,B$ on division by $Z^L$\;
    $M_1\gets\mathsf{HalfGCD}(A_1,B_1)$\;
    $\begin{bmatrix}A\\B\end{bmatrix}
        \gets M_1\begin{bmatrix}A\\B\end{bmatrix}$\;
    \lIf{$\deg B<s$}{\KwRet $M_1$}
    Let $Q$ be the quotient of $A$ on division by $B$\;
    $(A,B)\gets(B,A\bmod B)$\;
    \lIf{$\deg B<s$}{\KwRet $\varphi(Q)M_1$}
    $L\gets2s-\deg A$\;
    Let $A_1,B_1$ be the quotients of the current $A,B$ on division by $Z^L$\;
    $M_2\gets\mathsf{HalfGCD}(A_1,B_1)$\;
    \KwRet $M_2\varphi(Q)M_1$\;
}
\end{algorithm}

\FloatBarrier
\section{Supplementary Experimental Details}
\label{app:experiments}

This appendix gives implementation and measurement details for
Section~\ref{sec:experiments}, followed by additional heatmaps.

\subsection{Implementation and Measurement}
\label{app:experiment-details}

\paragraph{XYZ-Sketch}
Our implementation uses $\mathbb F_{2^{30}}$ with polynomial basis defined by $X^{30}+X+1$, carry-less multiplication, and specialized routines for small-degree polynomials $\ell\le 8$.
A cell transmits $\ell$ 30-bit coefficients and a $\lceil\log_2(2\ell+1)\rceil$-bit counter.
We use no fingerprint ($q=1$); algebraic, root, and rehashing checks remain enabled.

\paragraph{IBLT and IBLT+SC}
Both variants retain the default configurations of 30-bit key and a 32-bit checksum per cell.
We found the checksum necessary in practice: removing or substantially shortening it led to an unacceptably high decoding failure rate.

To reduce counter overhead, we use lossless block compression: each block of up to 128 counters stores its minimum in 32 bits, a 6-bit offset width, and the offsets from that minimum.
The decoder restores the exact counters before subtraction and peeling. 
MiniSketch uses its native serialization, and CPISync communication is measured from its actual protocol transcript.

\paragraph{Rateless IBLT}
We use the official implementation, retaining its 64-bit checksum
and generating symbols until decoding succeeds.
Following~\cite[Section~6]{RatelessSetReconciliation-YangGiladAlizadeh-2024},
we compress each count by encoding its deviation from the expected value.
We count all transmitted bytes, including metadata, and average
communication over 100 independent trials per scale.

\paragraph{Timing}
Experiments run on a server with an Intel Core i9-10980XE CPU
at 3.00\,GHz.
Update time is both parties' insertion CPU time divided by $|A|+|B|$. Full decoding includes sender serialization, an in-memory transfer, and receiver parsing, subtraction, reconstruction, and output normalization, divided by $d$. Input generation and the final exact-output comparison are outside the timed region.

C++ implementations use GCC 9.4.0 with \texttt{-O3}, and Rateless IBLT uses Go 1.21.13. Timing runs use one pinned core. We average three repetitions on each of five successful datasets, then average the five means; $95\%$ intervals use $10{,}000$ bootstrap resamples of these means.

For Rateless IBLT, update time measures input ingestion into the two
encoders. Full decoding includes sender symbol generation, receiver
processing, and wire encoding and decoding; receiver processing
includes local symbol generation, subtraction, and peeling.


CPISync times out at $d=3{,}000$ with a 3,600-second per-trial limit.
MiniSketch times out $d=30{,}000$.
Larger scales for these methods are not run.

\subsection{Parameter Calibration and Heatmaps}
\label{app:experiment-heatmaps}

We use $C=0.875$ and $D\approx1.18845$ in the heuristic rule of
Section~\ref{sec:heuristic-az}, with $\delta=0.1$ and $z$ rounded to
the nearest integer. These constants are selected using peeling
simulations and fixed before the evaluation. For $(k,\ell)=(2,6)$,
this gives $a\approx0.72437$ and
$z\approx0.381M^{1/3}$.



XYZ uses circular placement with the heuristic parameters above.
For IBLT+SC, we use $a=0$ and tune $k\in\{3,4,5,6,7\}$ and $z$,
selecting the configuration with the lowest communication cost
at the $90\%$ success target.
Calibration and evaluation each use 100 independent trials.

Figure~\ref{fig:fixed-M-heatmaps} shows how $(a,z)$ affects peeling
success in an ideal-cell simulation, where any cell containing
at most $\ell$ remaining elements can be decoded.
The heuristic choices generally lie within or near the high-success
regions, supporting the rule in Appendix~\ref{app:heuristic-az}.

\begin{figure}[htbp]
    \centering
    \includegraphics[width=0.49\linewidth]{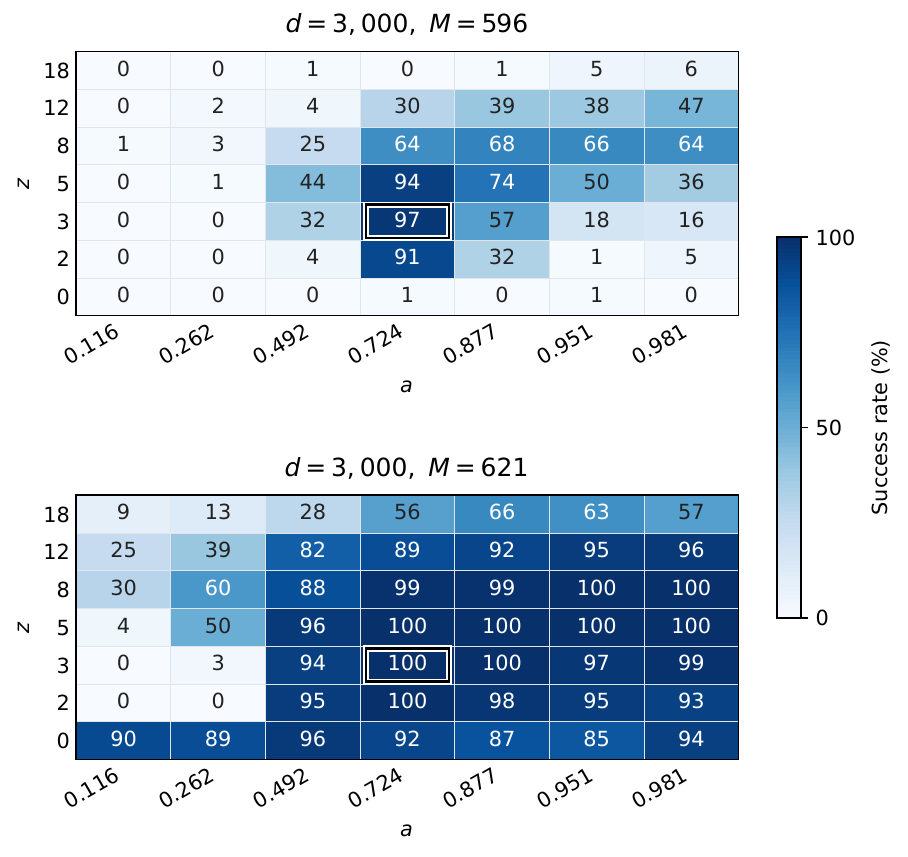}\hfill
    \includegraphics[width=0.49\linewidth]{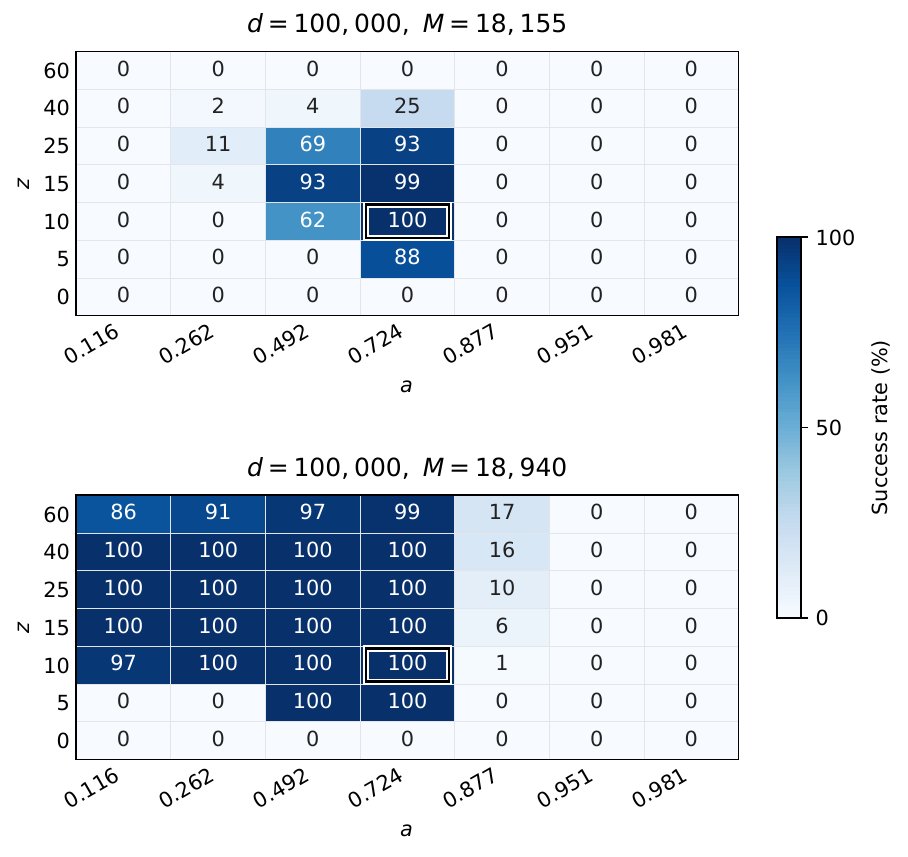}
    \caption{Supplementary $(a,z)$ heatmaps for $(k,\ell)=(2,6)$,
    with 100 ideal-cell peeling trials per grid point. Each panel
    states its $(d,M)$; outlined boxes mark the heuristic choices.}
    \label{fig:fixed-M-heatmaps}
\end{figure}

\subsection{Effect of Fingerprint Length}
\label{app:experiment-fingerprint}

We vary the fingerprint length of ordinary IBLT at $d=10^5$
and $10^6$, keeping $k=4$ and the calibrated cell count fixed
at each scale. Each configuration uses the same 100 input pairs
within each scale. Success requires exact recovery of the signed
difference; runs exceeding $10d$ peeling operations or $100M$
cell inspections count as failures.

Figure~\ref{fig:iblt-fingerprint} shows that shortening the
fingerprint substantially reduces decoding reliability.
With 20-bit fingerprints, success rates are $96\%$ and $47\%$
at $d=10^5$ and $10^6$, respectively; with 24 bits, they are
$100\%$ and $94\%$. Both 28-bit and 32-bit fingerprints recover
all 100 inputs at each scale. These results support retaining
a substantial fingerprint field (32 bits) in the IBLT baseline.

\begin{figure}[htbp]
    \centering
    \includegraphics[width=0.72\linewidth]
        {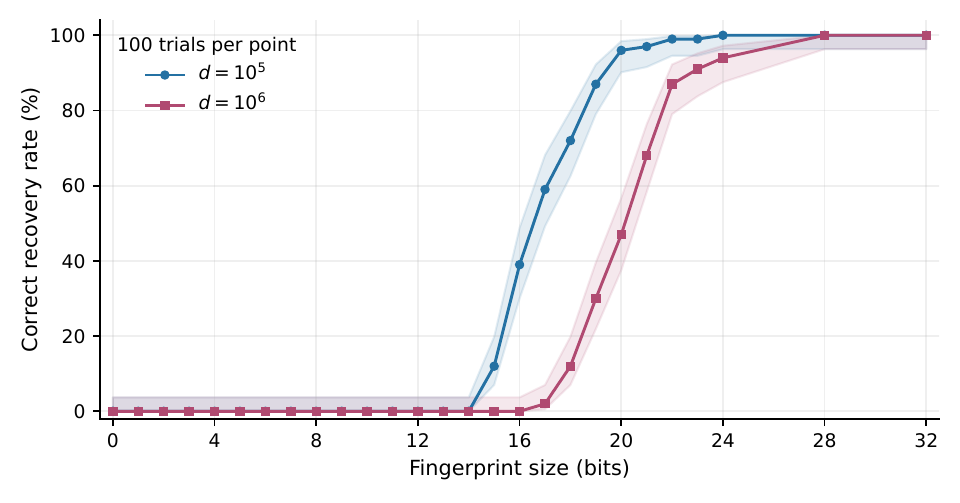}
    \caption{IBLT decoding success versus fingerprint length.
    We fix $k=4$, with $M=130{,}500$ for $d=10^5$ and
    $M=1{,}297{,}500$ for $d=10^6$.
    Each set contains $10^7$ elements.
    Each point uses 100 trials; shaded bands give
    Wilson $95\%$ confidence intervals.}
    \label{fig:iblt-fingerprint}
\end{figure}

\end{document}